\documentclass{article}

\usepackage{amsmath,amssymb,amsthm}
\usepackage{color}
\usepackage{graphicx}
\usepackage{tikz}
\usepackage{bbold}
\usepackage[hidelinks]{hyperref}
\usepackage{enumerate}

\newcommand {\R}{{\mathbb R}} 
\newcommand {\PP}{{\mathbb P}}
\newcommand {\EE}{{\mathbb E}}

\newcommand {\e}{{\varepsilon}}
\newcommand {\dd}{{\rm d}}

\usepackage{fullpage}

\usepackage{amsthm}
\newtheoremstyle{dotless}{}{}{\itshape}{}{\bfseries}{}{ }{}
\theoremstyle{dotless}

\newtheorem{theorem}{Theorem}[section]

\newtheorem{corollary}[theorem]{Corollary}
\newtheorem{assumption}[theorem]{Assumption}

\newtheorem{remark}{Remark}[section]
\newtheorem{definition}{Definition}[section]

\usepackage{titlesec}
\titleformat{\section}
   {\normalfont\fontsize{12pt}{0pt}\selectfont\bfseries}
   {\thesection}
   {1em}
   {}
      
\titleformat{\subsection}
   {\normalfont\fontsize{10pt}{0pt}\selectfont\bfseries}
   {\thesubsection}
   {1em}
   {}
   
\titleformat{\subsubsection}
   {\normalfont\fontsize{9pt}{0pt}\selectfont\bfseries}
   {\thesubsubsection}
   {1em}
   {}

\title{\bf Estimation of the average number of continuous\\
crossings for non-stationary non-diffusion processes}

\author{Romain Aza\"is\,$^a$ and Alexandre Genadot\,$^b$}

\date{\small
$^a$ Inria Nancy -- Grand Est, Team BIGS and Institut \'Elie Cartan de Lorraine, Nancy, France\\
$^b$ Institut de Math\'ematiques de Bordeaux and Inria Bordeaux -- Sud Ouest, Team CQFD, France
}

\begin{document}

\maketitle

\abstract{Assume that you observe trajectories of a non-diffusive non-stationary process and that you are interested in the average number of times where the process crosses some threshold (in dimension $d=1$) or hypersurface (in dimension $d\geq2$). Of course, you can actually estimate this quantity by its empirical version counting the number of observed crossings. But is there a better way? In this paper, for a wide class of piecewise smooth processes, we propose estimators of the average number of continuous crossings of an hypersurface based on Kac-Rice formulae. We revisit these formulae in the uni- and multivariate framework in order to be able to handle non-stationary processes. Our statistical method is tested on both simulated and real data.}

\section{Introduction}

We consider a random but non-diffusive trajectory $X$ that models some physical or biological phenomenon.
In order to set the ideas down, $X$ may be
issued from a piecewise deterministic Markov process (PDMP) \cite{Da93}
but our theoretical results will be established for a more general class of stochastic models.
In this paper, we are not interested in the estimation of the parameters that govern the underlying model that has generated the trajectory $X$ but in a functional of this trajectory.
Indeed, in numerous situations, the main quantity of interest is related to the average number of crossings of some level (in dimension $1$) or hypersurface (in dimension $d\geq2$) within a given time window.
For instance, if $X$ models the cumulative exposure to some food contaminant \cite{bertail2008}, the toxic effects depend on the time spent beyond a critical threshold. In reliability, $X$ may describe the size of a crack in a certain material \cite{doi:10.1177/1748006X16651170}. Exceeding a dangerosity threshold may lead to rupture and thus to a dramatic event. Unfortunately, the trajectory $X$ is often observed on a discrete temporal grid which time step size is imposed by the measuring devices. In this context, some crossings may be missed in such a way that the crude Monte Carlo-type estimator that only counts the number of observed crossings generally underestimates the theoretical quantity of interest. However, modern datasets often contain more information than the location at the measure time: they may provide the instantaneous velocity. This is typically the case when one studies spatial trajectories captured by a GPS device like the terrestrial and marine movements of lesser black-backed gulls investigated in \cite{Garthe16}. This additional information is crucial because it may give a clue on the likelihood of an unobserved crossing between two successive measure times. The main objective of this article is to show that one can take into account the velocity measurement to better estimate the average number of crossings.

\smallskip

\noindent
In the present article, we deal with piecewise smooth processes (PSP's) as defined in \cite{BL12}. They form a very general class of non-diffusion stochastic processes composed of deterministic trajectories following some differential equation punctuated by random jumps at random times. PDMP's are a particular case of PSP's because of the particular link between the inter-jumping times and the deterministic path that ensures the Markov property to hold. It should be noted that we do not impose any kind of Markov assumption in this paper. Kac-Rice formulae give a concise relation between the average number $C_S(H)$ of crossings of $X$ with a hypersurface $S$ within the time window $[0,H]$ and some features of the underlying stochastic model. They have already been stated for non-stationary one-dimensional PSP's \cite{DM15} and for stationary multidimensional PSP's \cite{BL12}. These papers \cite{BL12,DM15} and our approach are different and complementary. In the applications (typically spatial trajectories captured by GPS), the stochastic process of interest $X$ is often both multidimensional and non-stationary. In Theorem \ref{th:kr:multidim}, we establish under mild conditions the following Kac-Rice formula for multidimensional non-stationary PSP's,
\begin{equation}
\label{eq:csh1}
C_S(H)=\int_S|(r(x),\nu(x))|\int_0^H p_{X(s)}(x)\dd s\,\sigma_{d-1}(\dd x),
\end{equation}
where $r$ is the velocity of the deterministic motion, $\nu$ is a field of unit normals of $S$, $p_{X(t)}$ denotes the density of $X(t)$ and $\sigma_{d-1}$ stands for the Hausdorff measure (see Remark \ref{rem:hausdorff}). The strategy developed to state \eqref{eq:csh1} also gives a fresh look at Kac-Rice formula for one-dimensional processes. Consequently, we also provide a new proof of this formula for one-dimensional PSP's in Corollary \ref{cor:dmf}. In addition, we investigate in Corollaries \ref{crois:dim:1} and \ref{crois:dim:d} the Euclidean-mode setting often used in applications of hybrid Markov models (for instance in \cite{doi:10.1177/1748006X16651170}), which has never been studied from that perspective in the literature to the best of our knowledge.

\smallskip

\noindent
The distribution $p_{X(t)}$ appearing in the Kac-Rice formula \eqref{eq:csh1} is generally unknown but can be estimated (for example by kernel methods) from a dataset of trajectories observed within the time window $[0,H]$. In a wide range of applications, the deterministic motion is assumed to be known because it has been postulated by scientific laws, in particular in physical or biological models. This allows us to propose a new strategy for estimating the average number of crossings $C_S(H)$. If $\widehat{p}_{X(t)}$ denotes an estimate of $p_{X(t)}$, the number of crossings $C_S(H)$ can be approximated by the plug-in estimator
$$\widehat{C}_S(H) = \int_S|(r(x),\nu(x))|\int_0^H \widehat{p}_{X(s)}(x)\dd s\,\sigma_{d-1}(\dd x).$$
In the simulation study presented in Section \ref{s:simu}, we show that $\widehat{C}_S(H)$ better estimates $C_S(H)$ than the Monte Carlo estimator that only returns the empirical mean of observed crossings within the interval $[0,H]$, in particular when the time step size is large. In the real data application given in Section \ref{s:realdata}, we investigate the spatial trajectories of lesser black-backed gulls studied in \cite{Garthe16}. We do not assume any model for the velocity $r$ but we directly estimate the scalar product $|(r(x),\nu(x))|$ appearing in \eqref{eq:csh1} from instantaneous velocity measurements provided in the dataset \cite{Garthe16data}. Estimated Kac-Rice formulae allow us to approximate the average depth of marine and terrestrial trips of the birds within a one-day window, which helps to describe their daily habits.

\smallskip

\noindent
The paper is organized as follows. Section \ref{s:kr} is devoted to Kac-Rice formulae for non-stationary PSP's. The theoretical framework is presented in Subsection \ref{ss:kr:def}, while results for one-dimensional (multidimensional, respectively) processes are given in Subsection \ref{ss:freshlook} (Subsection \ref{ss:crossingsurf}, respectively). The proofs of the results stated in Section \ref{s:kr} have been deferred until Appendix \ref{app:proof2}. Section \ref{sec:sf} is dedicated to the statistical inference procedures. The proofs of the estimation results established in this section have been deferred until Appendix \ref{app:proof3}. A simulation study on PDMP's is provided in Section \ref{s:simu} through three examples: stationary one-dimensional telegraph process in Subsection \ref{tele:1d}, piecewise deterministic simulated annealing in Subsection \ref{ss:pdsa} and non-stationary two-dimensional telegraph process in Subsection \ref{tele:2d}. Finally, real data experiments are investigated in Section \ref{s:realdata}. 


\section{Kac-Rice formulae for piecewise smooth processes}
\label{s:kr}

\subsection{Definitions and background material}\label{ss:kr:def}

\subsubsection{A class of piecewise smooth processes}\label{psp}

A piecewise smooth process (PSP) $X$ on an open domain $\cal X$ of $\R^d$, endowed with the Euclidean norm $\|\cdot\|$ and associated scalar product $(\cdot,\cdot)$, involves continuous-time deterministic motions punctuated by random jumps at random times. Its dynamic is governed by a marked point process $(T_n,M_n)_{n\geq0}$ on $\R_+\times {\cal X}$ and a vector field $r : {\cal X}\to \R^d$, the $T_n$'s being the jump times of the continuous-time trajectory. The sequence of the jump times is assumed to be almost-surely increasing, i.e., $\PP(T_0=0< T_1\leq T_2\leq\ldots)=1$, and almost-surely going to infinity, i.e., $\PP(\lim_{n\to\infty} T_n=\infty)=1$. It should be noticed that the distribution of the $T_i$'s can depend on the marks. The following assumption on the jump times is quite classical in the context of hybrid models (see \cite[Standard conditions (24.8) and Remark (24.9)]{Da93}).

\begin{assumption}\label{assumption:time:mark}
The expected number of jump times over any bounded interval is finite.
\end{assumption}

\noindent
Between two successive jump times $T_n$ and $T_{n+1}$, the process $X$ evolves according to the vector field $r : {\cal X}\to \R^d$, which is assumed to be continuously differentiable on ${\cal X}$, and such that for any $\zeta\in {\cal X}$, the initial value problem
$$
\frac{\dd x}{\dd t}(t)=r(x(t)),\quad x(0)=\zeta,
$$ 
has a unique continuously differentiable solution $\phi(\zeta,\cdot)$ on $\R$. This solution satisfies the flow property,
$$
\forall\,(\zeta,t,s)\in\mathcal{X}\times\R^2,~\phi(\zeta,t+s)=\phi(\phi(\zeta,t),s).
$$
More precisely, the dynamic of the PSP $X$ can be described in an algorithmic fashion as follows. At time $T_0=0$, the process starts at point $X(0)=M_0\in {\cal X}$ and,
$$
\forall\,0\leq t<T_1,~X(t)=\phi(M_0,t-T_0).
$$
At time $T_1$, we set $X(T_1)=M_1$, and so on. We call $X$ a PSP. The process $X$ has the elementary properties to possess a c\`adl\`ag and piecewise continuously differentiable version and to have a number of jumps almost-surely finite and integrable over a finite time window. For technical purposes, the following regularity assumption is made on the density of the process.
\begin{assumption}
\label{assumption:density:new}
For any $t\geq0$, $X(t)$ admits a density $p_{X(t)}$ with respect to the Lebesgue measure on $\R^d$. 
\end{assumption}

\noindent
For a large class of applications, the considered PSP possesses a ``Euclidean-mode'' description. In such a case, we consider the process $(X,Y)$ where $X$ is the Euclidean variable, evolving in ${\cal X}$, and $Y$ is the mode variable, evolving in some discrete state space $\cal Y$. For any $y\in{\cal Y}$, we consider a vector field $r_y : {\cal X}\to \R^d$ which is continuously differentiable on ${\cal X}$. Then, for any $\zeta\in {\cal X}$, the initial value problem
$$
\frac{\dd x}{\dd t}(t)=r_y(x(t)),\quad x(0)=\zeta,
$$ 
has a unique continuously differentiable solution $\phi_y(\zeta,\cdot)$ on $\R$ satisfying the flow property
$$
\phi_y(\zeta,t+s)=\phi_y(\phi_y(\zeta,t),s).
$$
The process is then described as in the previous section except that the marks are now in the product space ${\cal X}\times {\cal Y}$, $M_n=(M_{X,n},M_{Y,n})$, $n\geq0$. At time $T_0=0$, the process starts at point $(X(0),Y(0))=M_0\in {\cal X}\times{\cal Y}$, and,
$$
\forall\,0\leq t<T_1,~(X(t),Y(t))=\left(\phi_{M_{Y,0}}(M_{X,0},t-T_0) , M_{Y,0}\right).
$$
At time $T_1$, we set $(X(T_1),Y(T_1))=M_1$, and so on. In this setting, $Y$ has a piecewise constant evolution and can be seen as some random parametrization of the evolution of $X$.

\subsubsection{A $\mathcal{C}^1$-hypersurface}\label{sbs:def:cross}

Let $S$ be a $\mathcal{C}^1$ and compact hypersurface being the boundary of a bounded domain ${\cal D}$ of $\cal X$ in dimension $d\geq2$. Following \cite[Chapter 1 1.2\,The concept of defining function]{K12}, we denote by $\rho: \R^d\to \R$ a $\mathcal{C}^1$-defining function for $\cal D$,
\begin{itemize}
\item $\rho$ is negative on $\cal D$, i.e., ${\cal D}=\{x\in{\cal X}\,:\,\rho(x)<0\}$,
\item $\rho$ is positive on the complementary of the closure of $\cal D$, i.e., $\bar{{\cal D}}^c=\{x\in{\cal X}\,:\,\rho(x)>0\}$,
\item and the gradient of $\rho$ is positive on $\partial{\cal D}$, i.e., $\nabla \rho(x)\neq0$ for all $x\in\partial{\cal D}=S$.
\end{itemize}
Then, for $x\in S$, $\nu(x)=\nabla\rho(x)/\|\nabla\rho(x)\|$ is the outward unit normal of $S$ at point $x$, well-defined and continuous in a neighborhood of $S$. Notice that on $S$, $n$ is independent of $\rho$. In dimension $d=1$, $S$ is reduced to a level, $S=\{x\}$ and $\nu\equiv1$ by convention.

\subsubsection{The interplay between the piecewise smooth process and the $\mathcal{C}^1$-hypersurface.}

Let $S$ be a surface as defined in the previous section. In this paper, we are interested in continuous crossings of $S$ as defined below.

\begin{definition}
Let $Z$ be an $\cal X$-valued process. We say that $Z$ has a continuous crossing of $S$ at time $s>0$, if $Z(s^-)=Z(s)\in S$ and there is $\delta>0$ such that $Z(t)\notin S$ for $t\in (s-\delta,s+\delta)\setminus\{s\}$. $S$ is then a crossing surface for $Z$.
\end{definition}

\noindent
$S$ is referred to as a crossing surface in the sequel. We denote by $c^Z_S(H)$ the (possibly random) number of continuous crossings within the interval $[0,H]$. In addition, $C^Z_S(H)$ denotes its expectation, $C^Z_S(H) = \EE[c^Z_S(H)]$. When there is no ambiguity, we will simply write $c_S(H)$ and $C_S(H)$, the dependency on $Z$ being implicit.

\smallskip

\noindent
Let $X$ be a PSP with flow $r$ as defined in the previous section. The assumption stating the relation between the deterministic behavior of $X$ and the crossing surface is the following.

\begin{assumption}
\label{assump:dd:tangent}
We assume that the trajectories of the flow are not tangent to $S$, 
\begin{equation}
\label{nt}
\forall\,x\in S,~(r(x),\nu(x))\neq 0.
\end{equation}
For any $x\in{\cal X}$, let $c^{\phi(x,\cdot)}_S(H)$ be the number of crossings of $S$ by the continuous deterministic process $\phi(x,\cdot)$ within the interval $[0,H]$. We assume that
\begin{equation}\label{fnt}
\sup_{x\in {\cal X}}c^{\phi(x,\cdot)}_S(H)<\infty.
\end{equation}
\end{assumption}

\noindent
Under condition \eqref{nt}, for any starting point $x\in{\cal X}$, the number $c^{\phi(x,\cdot)}_S(H)$ is always finite. To not overcomplicate proofs, we assume in condition \eqref{fnt} that this finiteness is uniform in $x$. These hypotheses are not difficult to verify in applications. In dimension $1$, these conditions take the following simple form, $r(x)\neq0$, where $S=\{x\}$. 

\smallskip

\noindent
For a real-valued smooth deterministic process, the following formula gives the number of crossings of a given level.
\begin{theorem}[Kac's counting formula \cite{Kac43}]
Let $x\in\R$, $H>0$, and $f : [0,H]\to\R$ a continuously differentiable function such that 
\begin{itemize}
\item either $f(0)\neq x$ or $f(H)\neq x$,
\item whenever $f(t)=x$ we have $f'(t)\neq0$. 
\end{itemize}
Then,
\begin{equation}\label{kac:formula}
c^f_x(H)=\lim_{\delta\to0}\frac{1}{2\delta}\int_0^H|f'(t)|\mathbb{1}_{|f(t)-x|\leq \delta}\dd t .
\end{equation}
\end{theorem}

\noindent
Our main probabilistic assumption about the relation between the surface and the piecewise deterministic smooth process is the following.
\begin{assumption}
\label{assumption:density}
The stochastic behavior of the PSP $X$ regarding the surface $S$ is as follows.
\begin{itemize}
\item Almost-surely, the PSP $X$ does not jump from or to $S$. In particular, the number of continuous crossings of $S$ by $X$ is almost-surely equal to the number of continuous crossings of $0$ by the process $(\rho(X(t))_{t\in[0,H]}$.
\item The family of density of the process is such that the map $(x,t)\mapsto p_{X(t)}(x)$ is continuous on ${\cal V}_S\times [0,H]$ with ${\cal V}_S$ a neighborhood of $S$, for any time horizon $H$.
\end{itemize}
\end{assumption}

\noindent
In the Euclidean-mode setting $Z=(X,Y)\in{\cal X}\times {\cal Y}$, the crossing surface $S$ is associated to the Euclidean domain $\cal X$. The conditions that we impose in this framework are the same as in the previous part but on the Euclidean component of the process for each value of the mode. We would like to point out that PSP's with such a decomposition are mentioned in \cite[Remark 3.4]{BL12}.

\begin{remark}
The formalism presented here is close to the setting of \cite[2.\,Model description]{BL12} except that we do not impose any kind of stationarity assumption.
\end{remark}

\begin{remark}
\label{rem:hausdorff}
In the present paper, integration over the hypersurface $S$ is performed with respect to the surface measure $\sigma_{d-1}$ of $\R^d$.  For regular hypersurfaces as considered in the present paper, when $d=1$, $\sigma_{0}$ is simply the counting measure and when $d=2$ ($d=3$, respectively), integration with respect to $\sigma_{1}$ ($\sigma_2$, respectively) reduces to line integrals (surface integrals, respectively). We refer the reader to \cite[Chapter 3]{K12} for more informations about surface measures (identical to Hausdorff measures in our context).
\end{remark}

\subsection{Average number of continuous level crossings}

\label{ss:freshlook}

This section aims at giving a new insight on the Kac-Rice formula counting the average number of crossings of one-dimensional PSP's. We begin this section with the presentation of a formula linking the number of continuous crossings with the local time.

\begin{theorem}[Local time-crossing relation]
\label{crois:d1}
For any level $x$ and any time horizon $H$, we have almost-surely
\begin{equation}
\label{crois:dim:1:formula}
c_x(H)=|r(x)|l_x(H),
\end{equation}
where $l_x(H)$ is the local time spent at level $x$ by the process $X$ between times $0$ and $H$,
$$l_x(H) = \lim_{\delta\to0}\frac{1}{2\delta}\int_0^H\mathbb{1}_{|X(t)-x|\leq \delta}\dd t .$$
\end{theorem}
\begin{proof}The proof lies in Appendix \ref{app:ss:crois:d1}.\end{proof}

\noindent
The Dalmao-Mordecki formula for the expectation of the number of crossings of one-dimensional PSP's, presented below but first established in \cite[5.\,Generalization of Borovkov-Last's formula]{DM15}, can be stated as a corollary of Theorem \ref{crois:d1}.

\begin{corollary}[Dalmao-Mordecki formula]
\label{cor:dmf}
For any level $x$ and any time horizon $H$, we have
\begin{equation}
\label{Kac-Rice:dim:1}
C_x(H)=|r(x)|\EE(l_x(H))=|r(x)|\int_0^H p_{X(s)}(x)\dd s.
\end{equation}
The average number of upward (downward, respectively) continuous crossings is obtained by replacing the absolute value of $r$ by its positive part (negative part, respectively) in the stated formulae.
\end{corollary}
\begin{proof}The proof lies in Appendix \ref{app:ss:cor:dmf}.\end{proof}

\noindent
Let us notice that the value of the average number of continuous crossings $C_x(H)$ can be actually greatly affected  by the initial distribution of the process. For instance, the linear process defined for $t\in[0,H]$ by
$$
X(t)= (1+t)\mathbb{1}_{X(0)>1}+(1-t)\mathbb{1}_{X(0)\leq1},
$$
that is increasing starting from $x>1$ but decreasing starting from $x\leq 1$, have crossings of the threshold $2$ that clearly depends on the initial distribution of the process. In formula \eqref{Kac-Rice:dim:1}, this dependence is given implicitly by the actual expression of the density $p_{X(s)}$ of $X$ at time $s$ whose actual expression can possibly be highly dependent on $X(0)$ as in the considered example.

\smallskip

\noindent
We now state the formula for the average number of continuous crossings in the Euclidean-mode setting. Since the state space of the mode is finite, the proof is, up to obvious changes, identical to the proofs of Theorem \ref{crois:d1} and Corollary \ref{cor:dmf}.

\begin{corollary}[Local time-crossing relation and Kac-Rice formula in the Euclidean-mode setting]
\label{crois:dim:1}
In the Euclidean-mode setting, for any level $x$ and any time horizon $H$, we have almost-surely
\begin{equation*}
c_x(H)=\sum_{y\in\cal Y}|r_y(x)|l_{(x,y)}(H),
\end{equation*}
where $l_{(x,y)}(H)$ is the local time spent at level $x$ by $X$ when the mode is $y$,
$$
l_{(x,y)}(H)=\lim_{\delta\to0}\frac{1}{2\delta}\int_0^H\mathbb{1}_{|X(t)-x|\leq \delta,\,Y(t)=y}\dd t.
$$
The average number of crossings is then
\begin{equation*}
C_x(H)=\sum_{y\in\cal Y}|r_y(x)|\int_0^H p_{(X(s),Y(s))}(x,y)\dd s.
\end{equation*}
\end{corollary}

\noindent
In the stationary case, Theorem \ref{cor:dmf} and Corollary \ref{crois:dim:1} take the following simple form.
\begin{corollary}[Stationary case]
Under the hypotheses of Corollary \ref{cor:dmf}, if the density $p_{X(0)}$ is stationary for the process $X$, we have
$$
C_x(H)=H|r(x)|p_{X(0)}(x).
$$ 
In the same line, under the hypotheses of Corollary \ref{crois:dim:1}, if the measure $p_{(X(0),Y(0))}$ is stationary for the process $(X,Y)$, we have
$$
C_x(H)=H\sum_{y\in\cal Y}|r_y(x)|p_{(X(0),Y(0))}(x,y).
$$ 
\end{corollary}

\noindent
Assumptions \ref{assump:dd:tangent} and \ref{assumption:density} are done in order that no tangency occur. However, in the Euclidean-mode setting, we can still handle situations where tangencies occur in a parallel fashion, that is $r_y(x)=0$ for some level $x$ with some mode $y$. But in this case, our formula does not count crossings of the form:
\begin{center}
\begin{tikzpicture}[scale=0.5]
\draw[dashed] (-1,1) node[left]{{\rm level}}--(19,1);
\draw (0,0)--(1,1)--(2,1)--(3,2);
\draw (5,0)--(6,1)--(7,1)--(8,0);
\draw (10,2)--(11,1)--(12,1)--(13,0);
\draw (15,2)--(16,1)--(17,1)--(18,2);
\end{tikzpicture}
\end{center}
For some processes, such situations may occur but with zero probability, making the stated formula still useful in practice.

\begin{remark}
The law of the PSP is involved in the expression of $C_x(H)$ through the integral
\begin{equation}
\label{eq:exp:lt}
\int_0^H p_{X(s)}(x)\dd s.
\end{equation}
From a pratical point of view, the actual computation of such a term may be desirable. However, even for very simple cases, analytical expressions of such a quantity can be either unreachable or untractable. In the Markov case, numerical approximations of \eqref{eq:exp:lt} can be built through the Fokker-Planck equation which reads,
$$
\forall\,x\in{\cal X},~\partial_t p_{X(t)}(x)={\cal L}^\ast p_{X(t)}(x),
$$
with initial condition $p_{X(0)}(x)$, where ${\cal L}^\ast$ is the adjoint of the generator of $X$.
\end{remark}

\subsection{Average number of hypersurface continuous crossings}
\label{ss:crossingsurf}

This section is devoted to the generalization of the formulae of Subsection \ref{ss:freshlook} to the multi-dimensional case. Let $(X(t),t\in[0,H])$ be a PSP in dimension $d\geq2$ and $S$ a hypersurface as described in Subsection \ref{psp}. The following theorem is complementary to the main theorem of \cite{BL12}.

\begin{theorem}[Kac-Rice formula in the multidimensional case]
\label{th:kr:multidim}
For any crossing surface $S$ and time horizon $H$, we have
\begin{equation}\label{KacRice:dd}
C_S(H)=\int_S|(r(x),\nu(x))|\int_0^H p_{X(s)}(x)\dd s\,\sigma_{d-1}(\dd x).
\end{equation}
\end{theorem}
\begin{proof}The proof lies in Appendix \ref{app:ss:th:kr:multidim}.\end{proof}

\noindent
When $d=1$, the hypersurface $S$ reduces to the point $\{x\}$. The measure $\sigma_0$ being the counting measure, the multidimensional formula \eqref{KacRice:dd} and the one-dimensional formula \eqref{Kac-Rice:dim:1} agree if we set the normal to a point in $\R$ to be equal to the scalar $1$. We also notice that the term $\int_0^H p_{X(s)}(x)\dd s$, corresponding to the expectation of the local time at point $x$, also appears as in dimension $1$. However, we are not able to give a concise relation between crossing and local time in the multidimensional case. We simply notice that we have the following relation
$$
c_S(H)\leq\sup_{x\in S}|(r(x),\nu(x))|l_S(H),
$$
where $l_S(H)$ is the local-time spent in $S$ by $X$ within the interval $[0,H]$,
$$
l_S(H)=\lim_{\delta\to0}\frac{1}{2\delta}\int_0^H \mathbb{1}_{{\cal T}(S,\delta)}(X(s))\dd s,
$$
with ${\cal T}(S,\delta)$ the $\delta$-tube of elements of $\cal X$ being at distance less than $\delta$ from $S$.

\smallskip

\noindent
In the stationary case, Theorem \ref{th:kr:multidim} takes the following simple form.
\begin{corollary}
\label{rem:md:statio}
When $p_{X(0)}$ is a stationary density for the process $X$, we have
\begin{equation}
\label{eq:csh:md:statio}
C_S(H)=H \int_S|(r(x),\nu(x))|p_{X(0)}(x) \sigma_{d-1}(\dd x).
\end{equation}
\end{corollary}

\noindent
In the Euclidean-mode setting, the following result holds.
\begin{corollary}[Kac-Rice formula in the Euclidean-mode setting in the multidimensional case]
\label{crois:dim:d}
We consider the Euclidean-mode setting with $\cal Y$ a finite state space for the mode. For any crossing surface $S$ and time horizon $H$, we have
\begin{equation*}
C_S(H)=\sum_{y\in\cal Y}\int_S|(r_y(x),\nu(x))|\int_0^H p_{(X(s),Y(s))}(x,y)\dd s\,\sigma_{d-1}(\dd x).
\end{equation*}
\end{corollary}

\noindent
Kac-Rice formulae take into account the instantaneous velocity to provide a tractable form for the quantity of interest $C_S(H)$, which contains two characteristics of the model calculated at any point $x$ of the surface $S$: the vector field $r(x)$ and the distribution $p_{X(t)}(x)$. These quantities are unknown and thus remain to be estimated. In many applications, the deterministic motion is supposed to be known and thus only the distribution has to be estimated from data. This is the framework investigated in Section \ref{sec:sf} and in the simulation study presented in Section \ref{s:simu}. In the real data application presented in Section \ref{s:realdata}, we also show how to directly estimate the scalar product $|(r(x),\nu(x))|$ depending on the unknown underlying vector field.


\section{Statistical framework and estimation procedures}\label{sec:sf}

We observe a dataset ${\cal T}=\{{\cal T}_i\,:\,1\leq i\leq n\}$ made of $n$ discrete trajectories within the time window $[0,H]$. Each trajectory $\mathcal{T}_i$ is itself composed of $n_H$ points in $\R^d$, ${\cal T}_i=\{z_{i}^{j}\,:\,1\leq j\leq n_H\}$ given on the regular temporal grid
$$\left\{\frac{H(j-1)}{n_H-1}\,:\,1\leq j\leq n_H\right\},$$
which is common to all the trajectories. We assume that the $n$ trajectories come from the same underlying piecewise deterministic Markov process $X$ and are independent. Our goal is to estimate the average number of continuous crossings $C_S(H)$ of the hypersurface $S$ by the process $X$ from the dataset $\cal T$.

\smallskip

\noindent
As noticed in Subsection \ref{ss:freshlook}, $C_S(H)$ implicitly depends on the initial value of the process $X_0$, in particular through the family $(p_{X(t)})_{0\leq t\leq H}$ in \eqref{KacRice:dd} that also depends on the initial condition. We want to emphasize that we propose to estimate the average number $C_S(H)$ of crossings from $X_0=x$ from a dataset ${\cal T}$ of trajectories that all start from $x$. This allows us to estimate the densities $p_{X(t)}$ without considering the question of the initial value of the process.

\smallskip

\noindent
In the sequel, we introduce three estimation procedures of the average number of crossings $C_S(H)$, namely the Monte Carlo method, the non-stationary Kac-Rice estimator and the stationary Kac-Rice estimator. The Monte Carlo method only returns the empirical mean of the number of observed crossings. This technique does not exploit the characteristics of the underlying model. The non-stationary Kac-Rice estimator consists in plugging an estimator of the distribution of $X$ in the Kac-Rice formula \eqref{KacRice:dd}, taking into account the knowledge of the velocity $r$. The stationary Kac-Rice method assumes further that the trajectory is stationary by plugging an estimator of the invariant measure in the stationary Kac-Rice formula \eqref{eq:csh:md:statio}.

\paragraph{Monte Carlo method} The most naive estimator of the average number of crossings is the Monte Carlo method. For each trajectory ${\cal T}_i$, $1\leq i\leq n$, we count the number of crossings of the given level or hypersurface. We then average over the $n$ trajectories and denote the result by $\check{C}_S(H)$. Denoting $[AB]$ the line segment from point $A$ to point $B$, we have
$$
\check{C}_S(H)=\frac1n\sum^n_{i=1}\sum_{j=1}^{n_H-1}\#[z_{i}^{j}z_{i+1}^{j}]\cap S.
$$ 
In dimension $1$, for a level $x$, the latter quantity reads
$$
\check{C}_x(H)=\frac1n\sum^n_{i=1}\sum_{j=1}^{n_H-1}\mathbb{1}_{(-\infty,0)}((z_{i}^{j}-x)(z_{i+1}^{j}-x)).
$$
As mentioned in the introduction, this estimator has a natural inclination to be biased. Indeed, $\check{C}_S(H)$ might miss one or several continuous crossings between two consecutive temporal steps given the discrete nature of the temporal grid, especially when the time step size is large.

\paragraph{Non-stationary Kac-Rice method} This technique consists in replacing the distribution of $X$ appearing in \eqref{KacRice:dd} by some non-parametric estimator. More precisely, we estimate the density $p_{X(h_j)}$ at time $h_j=H(j-1)/(n_H-1)$ by kernel smoothing methods using the whole dataset, i.e., all the observed values $z_i^j$ of $X(h_j)$,
\begin{equation}\label{eq:def:phat}
\forall\,x\in S,~\widehat{p}_{X(h_j)}(x)=\frac{1}{n \sqrt{\det(B_{n})}}\sum_{i=1}^{n}\mathbb{K}\left(B^{-1/2}_{n}(x-z_i^j)\right),
\end{equation}
with $\mathbb{K}$ some kernel function and $B_{n}$ the bandwith, i.e., a $d\times d$ symmetric positive-definite matrix.
We then compute the plug-in estimator (see Theorem \ref{th:kr:multidim})
\begin{equation}
\label{eq:estim:kr}
\widehat{C}_S(H)=\int_S|(r(x),\nu(x))|\,\frac{H}{n_H-1}\sum_{j=1}^{n_H}\widehat{p}_{X(h_j)}(x)\,\sigma_{d-1}(\dd x),
\end{equation}
where $H(n_H-1)^{-1}\sum_{j=1}^{n_H}\widehat{p}_{X(h_j)}(x)$ estimates (by the rectangle method) the integral $\int_0^H p_{X(s)}(x)\dd s$. In dimension $1$, this quantity is reduced to 
$$
\widehat{C}_x(H)=|r(x)|\frac{H}{n_H-1}\sum_{j=1}^{n_H}\widehat{p}_{X(h_j)}(x).
$$

\noindent
The theoretical consistency of the non-stationary Kac-Rice estimator is established from the asymptotic properties of the kernel estimator \eqref{eq:def:phat}. In the theorem and subsequent proof below, we write $\widehat{C}^{(n,n_H)}_S(H)$ for $\widehat{C}_S(H)$ and $\widehat{p}^{(n)}_{X(h_j)}(x)$ for $\widehat{p}_{X(h_j)}(x)$ to stress out the dependence of the estimator in the number $n$ of trajectories and number $n_H$ of time steps.
\begin{theorem}\label{thm:nonstatkr:consis}
The non-stationary Kac-Rice estimator $\widehat{C}^{(n,n_H)}_S(H)$ is strongly consistent in the sense that, with probability one, 
$$
\lim_{n_H\to\infty}\lim_{n\to\infty}\widehat{C}^{(n,n_H)}_S(H)=C_S(H),
$$
as long as the bandwith $B_n$ and the kernel $\mathbb{K}$ appearing in \eqref{eq:def:phat} are selected as in \cite[Theorem 1]{DW80}. For instance, it is sufficient to assume that there exists some $\delta>0$ such that,
$$
\lim_{\|x\|_d\to\infty}\|x\|^{d+\delta}_d \mathbb{K}(x)=0,\quad\lim_{n\to\infty}h_n=0,\quad\text{and}\quad \lim_{n\to\infty} \frac{n h^{2d}_n}{\log(n)}=0,
$$
where $h_n=\sqrt{{\rm det}(B_n)}$ and $\|\cdot\|_d$ is any norm on $\R^d$.
\end{theorem}
\begin{proof}The proof lies in Appendix \ref{app:ss:nonstat}.\end{proof}

\noindent
We refer the reader to \cite{Chacon2010} for recent developments on multivariate kernel density estimation with plug-in bandwidth selection. In the simulation study, we exploit the implementation of the latter algorithm in the \verb+R+ package \verb+ks+ \cite{JSSv021i07}.

\paragraph{Stationary Kac-Rice method} We now assume that the underlying process $X$ is stationary. For each trajectory ${\cal T}_i$, $1\leq i\leq n$, we estimate the invariant distribution $\mu$ on the hypersurface $S$ by kernel smoothing methods,
$$
\forall\,x\in S,~\widehat{\mu}_i(x)=\frac{1}{n_H \sqrt{\det(B_{n_H})}}\sum_{j=1}^{n_H}\mathbb{K}\left(B^{-1/2}_{n_H}(x-z_{i}^{j})\right).
$$
We then compute the plug-in estimator given by the Kac-Rice formula in the stationary case (see Remark \ref{rem:md:statio})
\begin{equation}
\label{eq:estim:kr:statio}
\widetilde{C}_{S,i}(H)=H \int_S|(r(x),\nu(x))|\widehat{\mu}_i(x)\sigma_{d-1}(\dd x).
\end{equation}
In dimension $1$, this quantity is simply
$$
\widetilde{C}_{x,i}(H)=H |r(x)|\widehat{\mu}_i(x).
$$
We would like to emphasize that we then obtain a good estimator of $C_S(H)$ from each trajectory $\mathcal{T}_i$. For the sake of comparison (with Monte Carlo and non-stationary Kac-Rice methods), the estimation from the whole set of trajectories is then given by the empirical mean
$$\widetilde{C}_S(H) = \frac{1}{n}\sum_{i=1}^n \widetilde{C}_{S,i}(H).$$

\noindent
The convergence of the stationary Kac-Rice estimator \eqref{eq:estim:kr:statio} is deduced from the asymptotic properties of the kernel estimator of $\nu$ as in the non-stationary case (see for instance \cite{nguyen1979} for results on the estimation of the stationary distribution of a continuous-time Markov process). In the theorem and proof below, we write $\widetilde{C}^{(n,n_H)}_S(H)$ for $\widetilde{C}_S(H)$ in order to stress out the dependence of the estimator in the number $n$ of available trajectories and the number $n_H$ of time steps.
\begin{theorem}\label{thm:statkr:consis}
In the stationary case, let us assume that the process $X$, the kernel $\mathbb{K}$ and the stationary measure $\mu$ satisfy the hypotheses of  \cite[Theorem 6]{ADP16}. Then, the stationary Kac-Rice estimator is consistent in the sense that, in $\PP_\mu$-probability, for any $n\geq1$,
$$
\lim_{n_H\to\infty}\widetilde{C}^{(n,n_H)}_S(H)=C_S(H) .
$$
\end{theorem}
\begin{proof}The proof lies in Appendix \ref{app:ss:stat}.\end{proof}

\noindent
Only the non-stationary Kac-Rice estimator requires to observe all the trajectories on the same temporal grid. Indeed, this condition ensures that the distribution of $X(t)$ may be easily estimated from the observed dataset. In addition, the stationary Kac-Rice method is the only one that allows to estimate the average number of crossings from only one trajectory under the stationarity assumption. Furthermore, it also allows us to estimate the average number of crossings on a time window larger than the observation one.


\section{Simulation study on piecewise deterministic Markov processes}
\label{s:simu}

The following simulation study has been performed on some particular PSP's in the Euclidean-mode setting and having the property to be Markov, namely PDMP's as introduced in \cite{Da93}. A PSP with mode $(X,Y)$ is a PDMP if and only if the following conditions are satisfied. For any integer $n$, conditionally on $T_n$ and $M_n$, with $M_n=(M_{X,n},M_{Y,n})$, the next jump time $T_{n+1}$ of the process $(X,Y)$ has the following survival function,
$$
t\mapsto1\wedge\exp\left(-\int_{T_n}^t\lambda\left(\phi_{M_{Y,n}}(M_{X,n},s-T_n)\right)\dd s\right),
$$
where $\lambda: {\cal X}\to \R_+$ denotes the jump rate, assumed to be path-integrable.
In the same line, conditionally on $T_n$, $T_{n+1}$ and $M_n$, the distribution of the next mark $M_{n+1}$ depends (only) on $\phi_{M_{Y,n}}(M_{X,n},T_{n+1}-T_n)$.
In this Markov setting, if $X(0)$ and the conditional distribution of the marks have densities with respect to some measure $\mu$, then, for any time $t$, $X(t)$ has a density with respect to this same measure $\mu$. We refer the reader to \cite{TK09} on these questions. As a consequence, if $\mu$ is the Lebesgue measure, Assumption \ref{assumption:density} is satisfied for PDMP's under these mild conditions.

\smallskip

\noindent
In Subsection \ref{tele:1d}, we aim at showing through an example the convergence of our estimation procedures when the number of trajectories increases, as well as the limited impact of applying them to trajectories observed on a coarse temporal grid, while the Monte Carlo method suffers from this kind of setting. For the sake of completeness, in Subsection \ref{ss:pdsa} (see in particular Figure \ref{pdsa:level0}), we present a situation where the Monte Carlo estimator always gives better results than the Kac-Rice methods. The methodology is presented in a two-dimensional setting in Subsection \ref{tele:2d}.

\subsection{Stationary one-dimensional telegraph process}
\label{tele:1d}

We consider a variant of the telegraph process \cite{Fontbona20163077} that is a quite simple piecewise deterministic Markov process $(X,Y)$ in the Euclidean-mode setting, with value in $\mathbb{R}\times\{-1,+1\}$ and evolving according to the generator
$$\mathcal{L}f(x,y) = y \partial_xf(x,y) +\left(a\,\mathbb{1}_{\{xy\leq 0\}}+b\,\mathbb{1}_{\{xy>0\}}\right)\left(f(x,-y)-f(x,y)\right),$$ 
where $a$ and $b$ are two positive parameters. The behavior of the process can be described as follows. The Euclidean variable $X$ represents the evolution of a particle on the real line with velocity jumping between $Y=-1$ and $Y=+1$ at rate $a$ ($b$, respectively) if the particle approaches (goes away from, respectively) the origin. When $b>a$, the process exhibits a stationary behavior \cite[Theorem 1.1]{fontbona2012quantitative} with convergence in law towards the invariant measure $\mu$ on $\mathbb{R}\times\{-1,+1\}$, defined by
$$\mu(\dd x\times\dd y)  = \frac{b-a}{2}\exp\left(-(b-a)|x|\right)\dd x \otimes  \frac{1}{2}(\delta_{-1}+\delta_{+1})(\dd y).$$
Thus, at equilibrium, the mode is equidistributed between $-1$ and $+1$ whereas the Euclidean variable is symetrically distributed on both sides of the origin. In our simulation study, we assume that the process is stationary and thus starts from its invariant measure. For this purpose, given a random variable $U$ with uniform distribution on the interval $[0,1]$, let us remark that a random variable with density
$$x\mapsto\frac{b-a}{2}\exp\left(-(b-a)|x|\right)$$
is equal in law to
$$
\frac{\log(2U)}{b-a}\mathbb{1}_{(-\infty,0)}(\log(2U))-\frac{\log(2(1-U))}{b-a}\mathbb{1}_{(-\infty,0)}(\log(2(1-U))),
$$
which can be easily simulated. A stationary trajectory of the Euclidean variable $X$ (with $b=2$ and $a=1$) until time $H=100$ is displayed in Figure \ref{traj_telegraph_1d} together with the empirical distribution of a trajectory until time $H=500$ on a discrete grid compared to the theoretical invariant measure.

\begin{figure}[t]
\centering
\includegraphics[width=7cm]{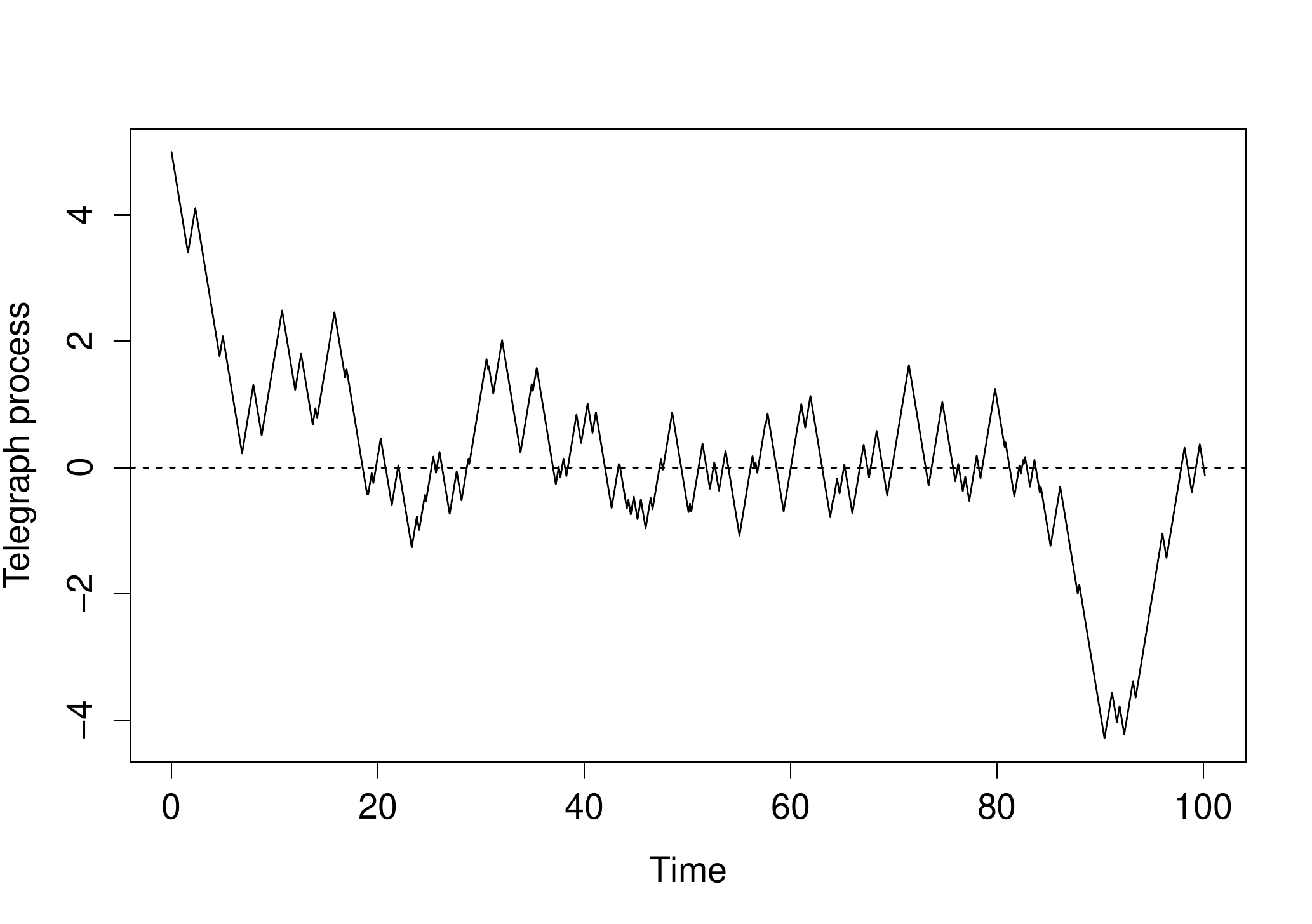}\qquad\qquad\includegraphics[width=7cm]{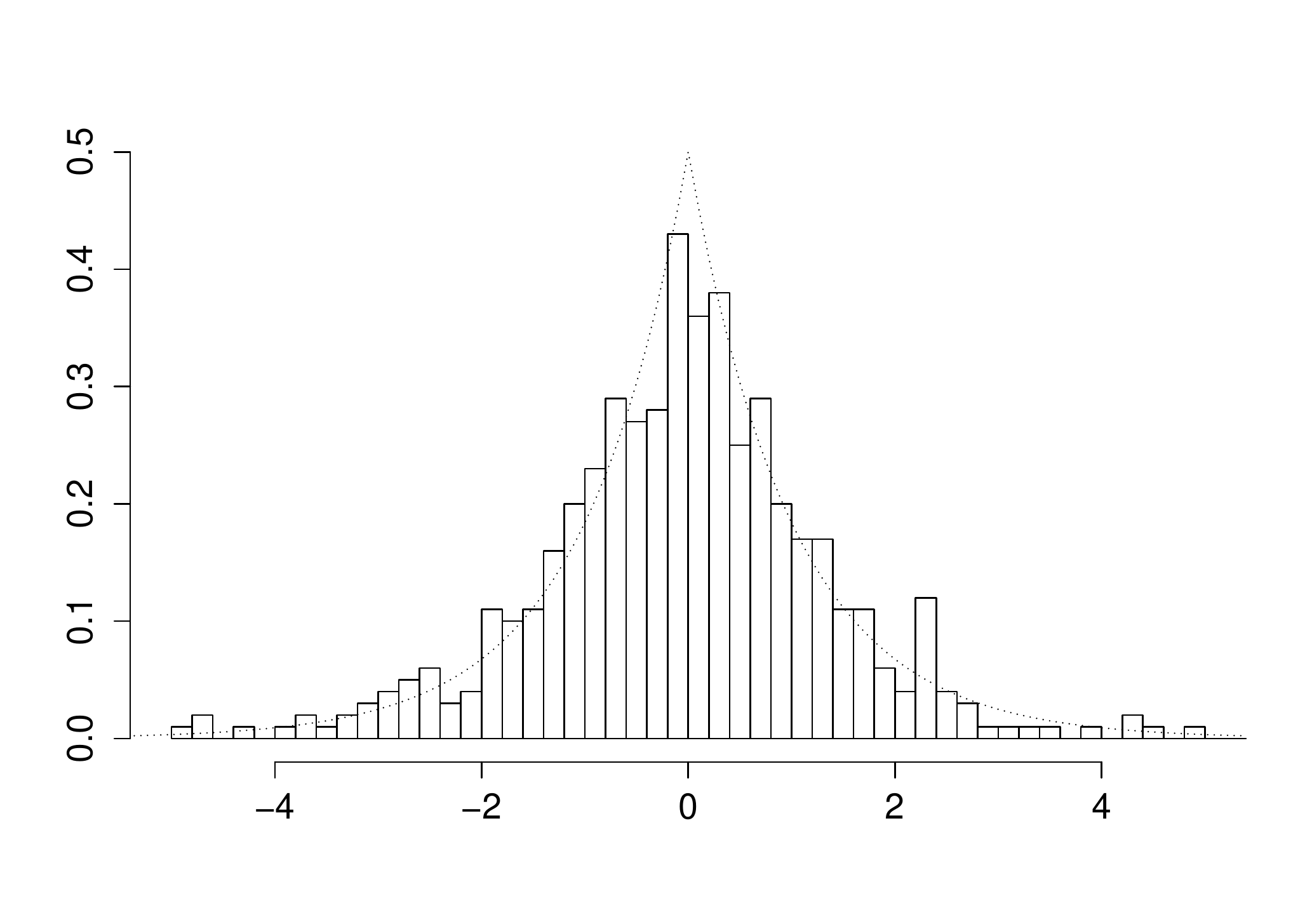}
\caption{A trajectory of the stationary one-dimensional telegraph process until time $H=100$ with parameters $b=2$ and $a=1$ (left) and histogram of a trajectory until time $H=500$ with the theoretical invariant measure (right).}\label{traj_telegraph_1d}
\end{figure}

\smallskip

\noindent
We have estimated the average number of crossings $C_2(50)$ from different numbers $n$ of trajectories ($n=50$, $n=100$, $n=200$, $n=500$ and $n=1\,000$) and with different time step sizes $h$ ($h=0.01$, i.e., $n_H=10\,001$, $h=0.1$, i.e., $n_H=1001$, $h=1$, i.e., $n_H=101$ and $h=2$, i.e., $n_H=51$) by the three procedures mentioned above. The numerical results are presented in Figure \ref{fig:cvsimu}. They show that the Kac-Rice estimators tend to be concentrated around the true value of $C_2(50)$ when the number of trajectories increases, which shows the consistency of our inference procedures. Furthermore, the thiner the step, the better the kernel estimator of the stationary measure behaves (a large number of data is available) and the smaller the bias of the Kac-Rice estimators is. Even from a coarse temporal grid, the Kac-Rice estimators perform pretty well and present a small bias, while the Monte Carlo estimator misses a large number of crossings. This proves the great interest of our approach. Finally, it should be noticed that the stationary and the non-stationary Kac-Rice estimators provide similar performances even in this stationary setting increasing the attractiveness of the latter.

\smallskip

\noindent
In Figure \ref{bp_step_size_s_0_telegraph_1d} are displayed boxplots corresponding to both estimators of the average number of crossings of the level $x=0$, $C_0(100)=50$ with an increasing time step size $h\in\{0.1,1,2\}$. Due to the shape of the density, the stationary Kac-Rice estimator underestimates the average number of crossings. Indeed the stationary density is underestimated at $0$ as illustrated in Figure \ref{bp_step_size_s_0_telegraph_1d}. We remark that in such a case, for small step size, the empirical estimator is the closest to the real value on average whereas for greater step size, the Kac-Rice-based estimator outperforms the Monte Carlo estimator, although still quite far from the true value.

\begin{figure}[p]
\centering
\includegraphics[width=5.5cm]{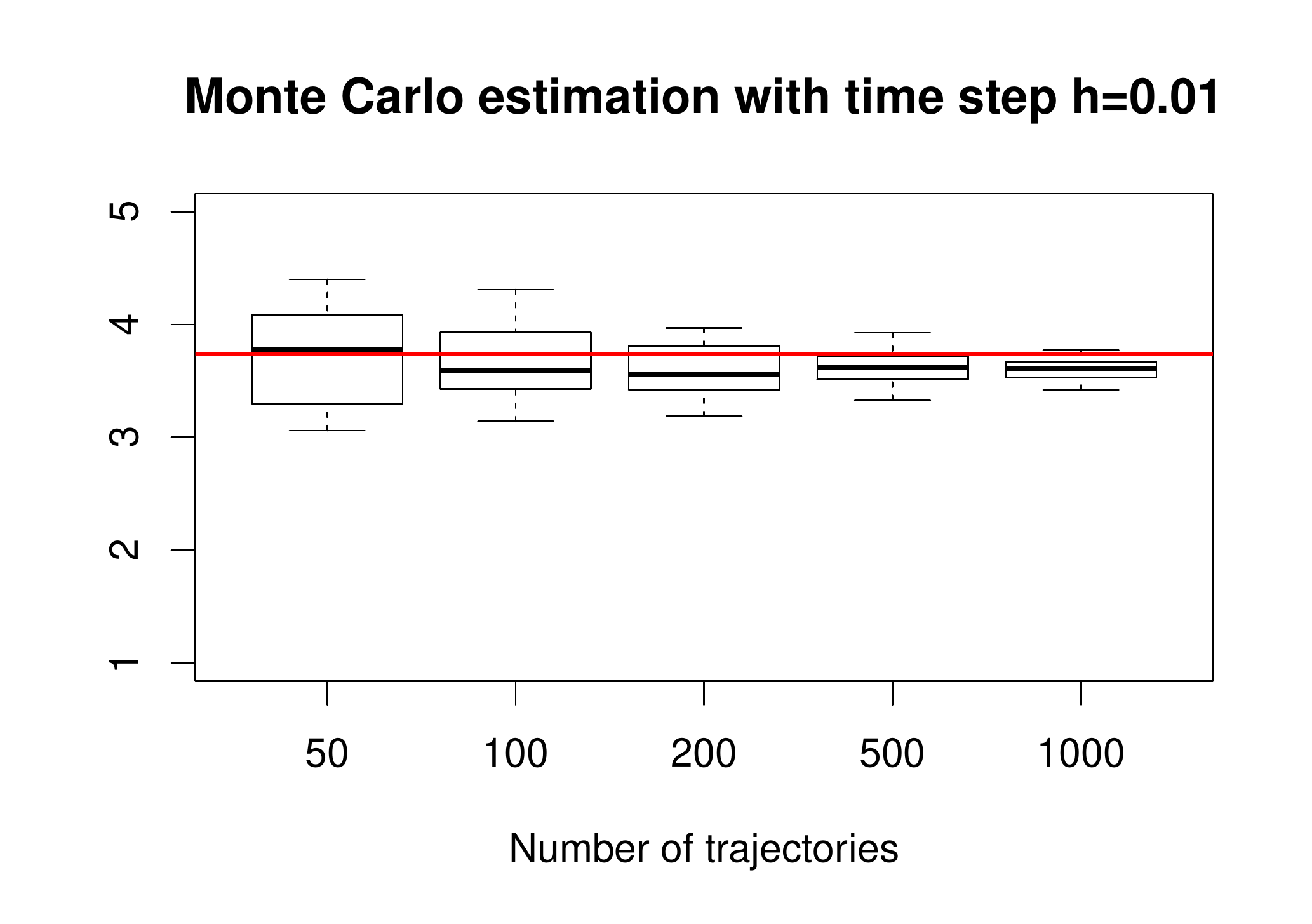}\includegraphics[width=5.5cm]{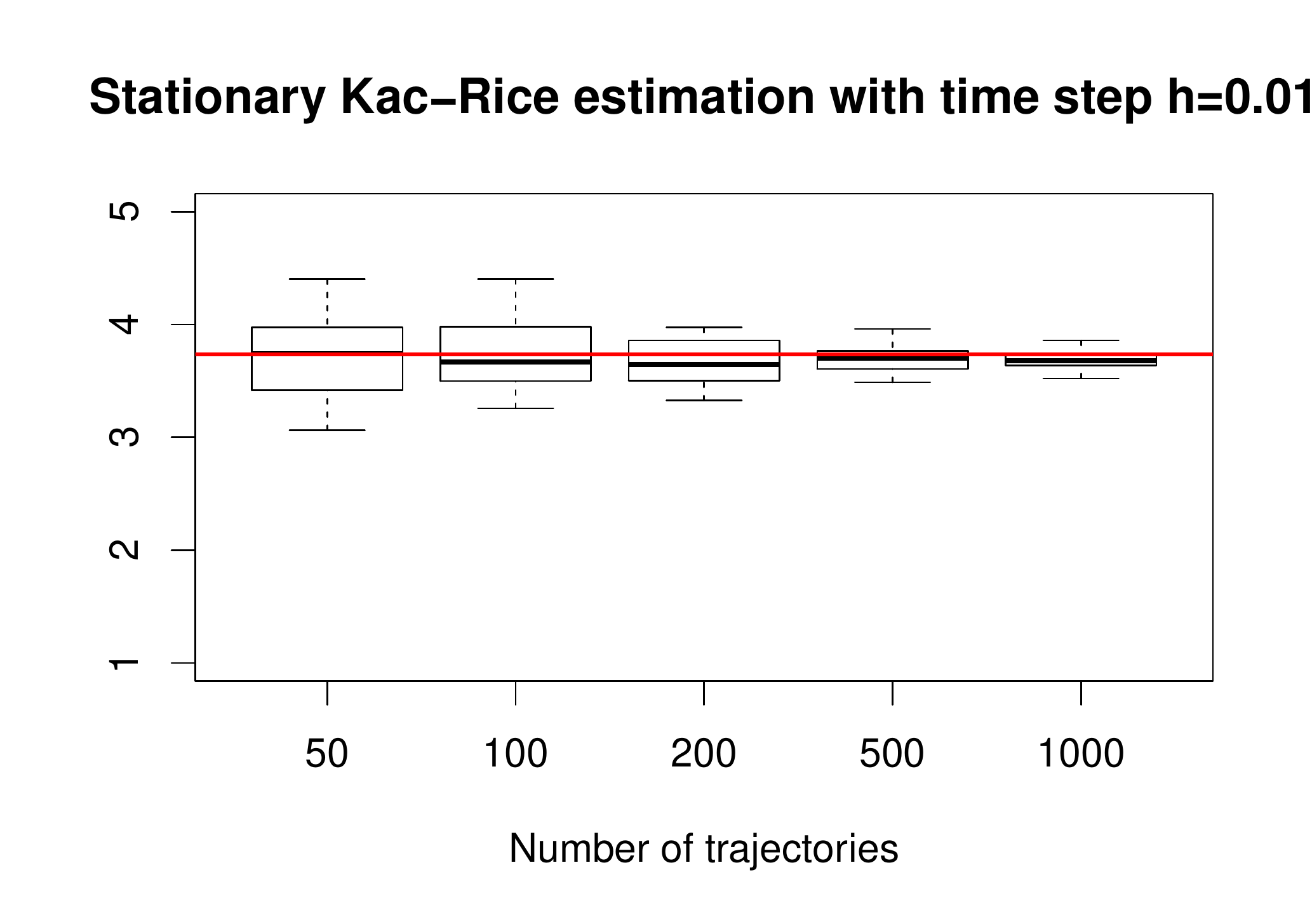}\includegraphics[width=5.5cm]{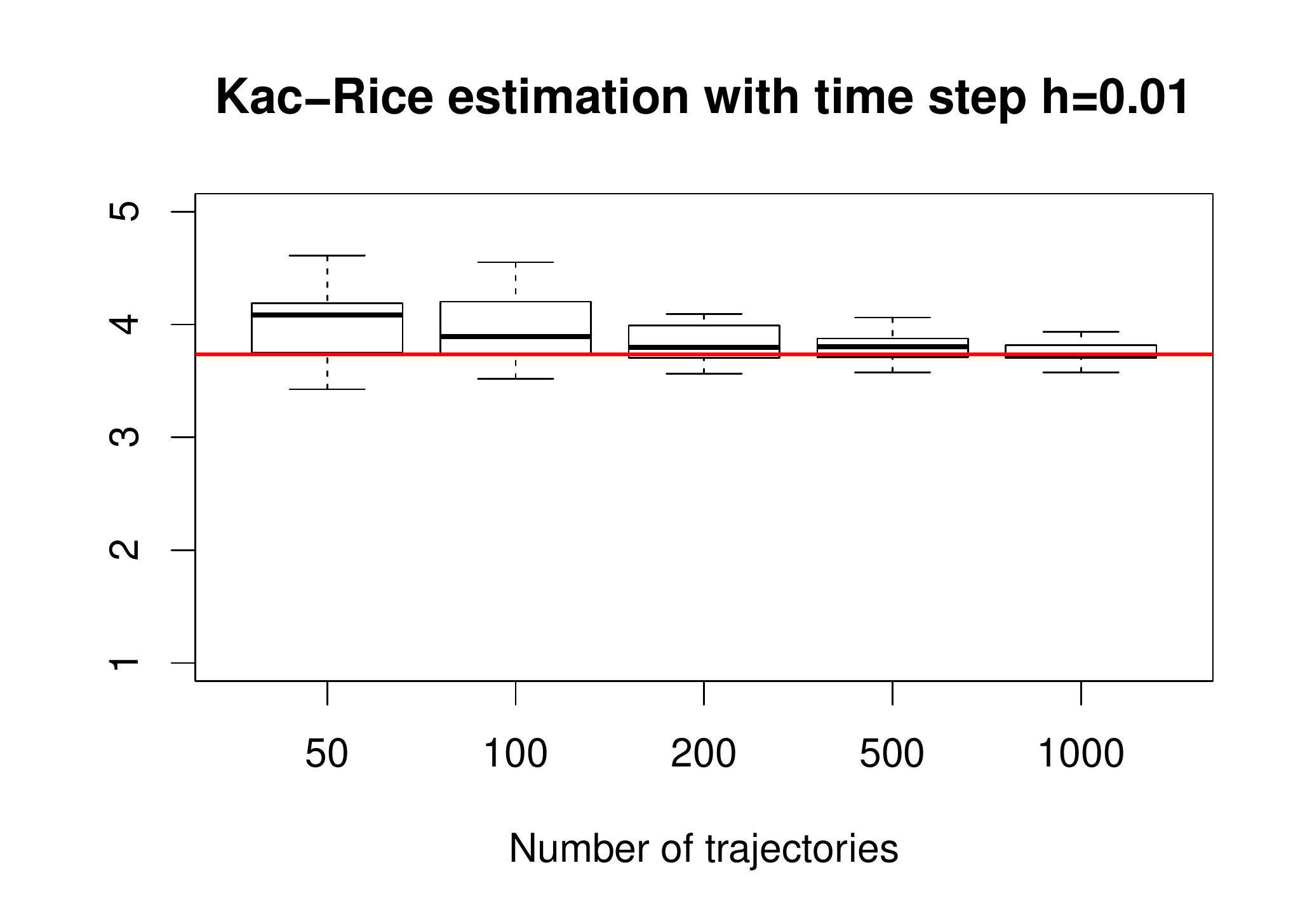}\\
\includegraphics[width=5.5cm]{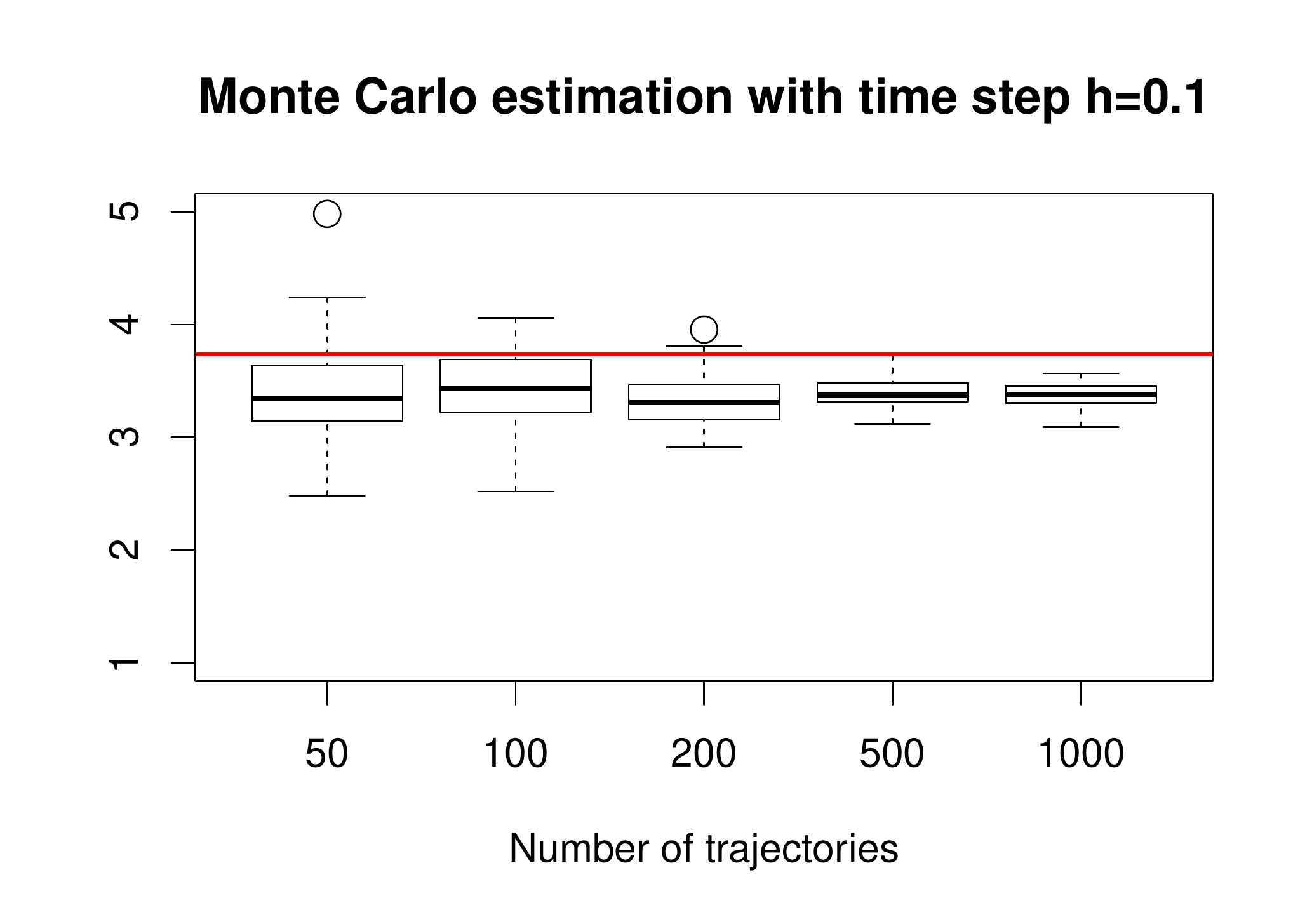}\includegraphics[width=5.5cm]{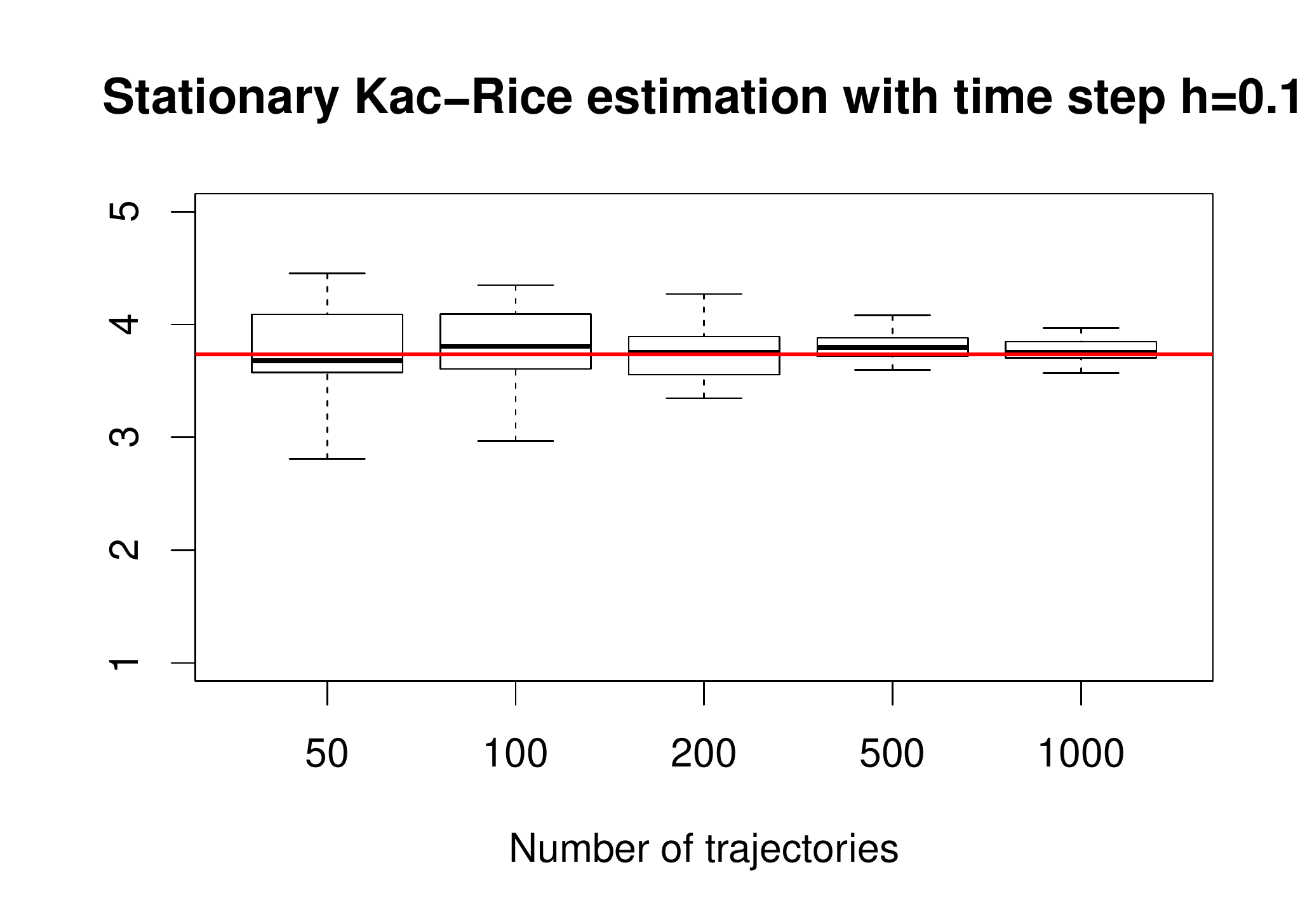}\includegraphics[width=5.5cm]{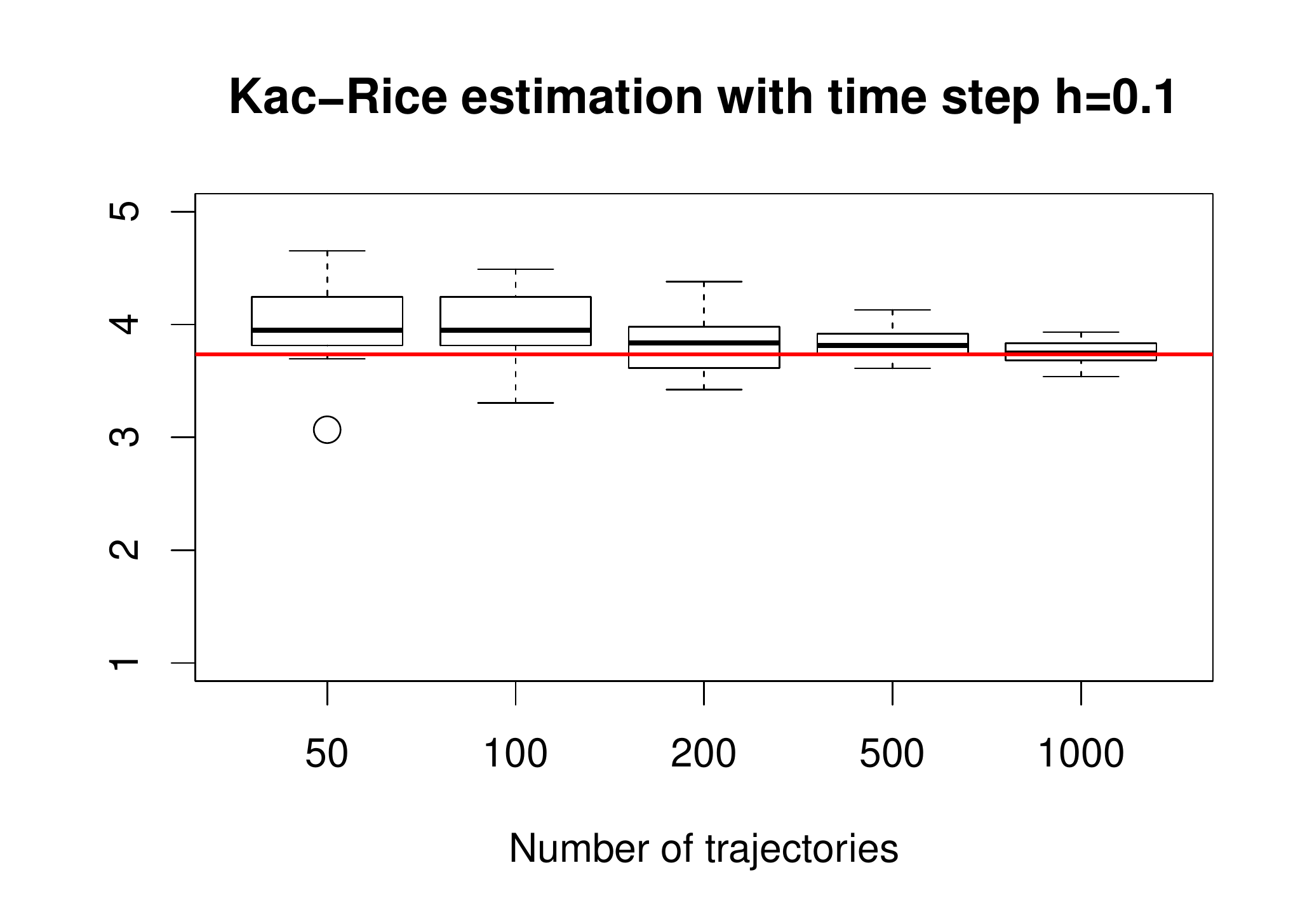}\\
\includegraphics[width=5.5cm]{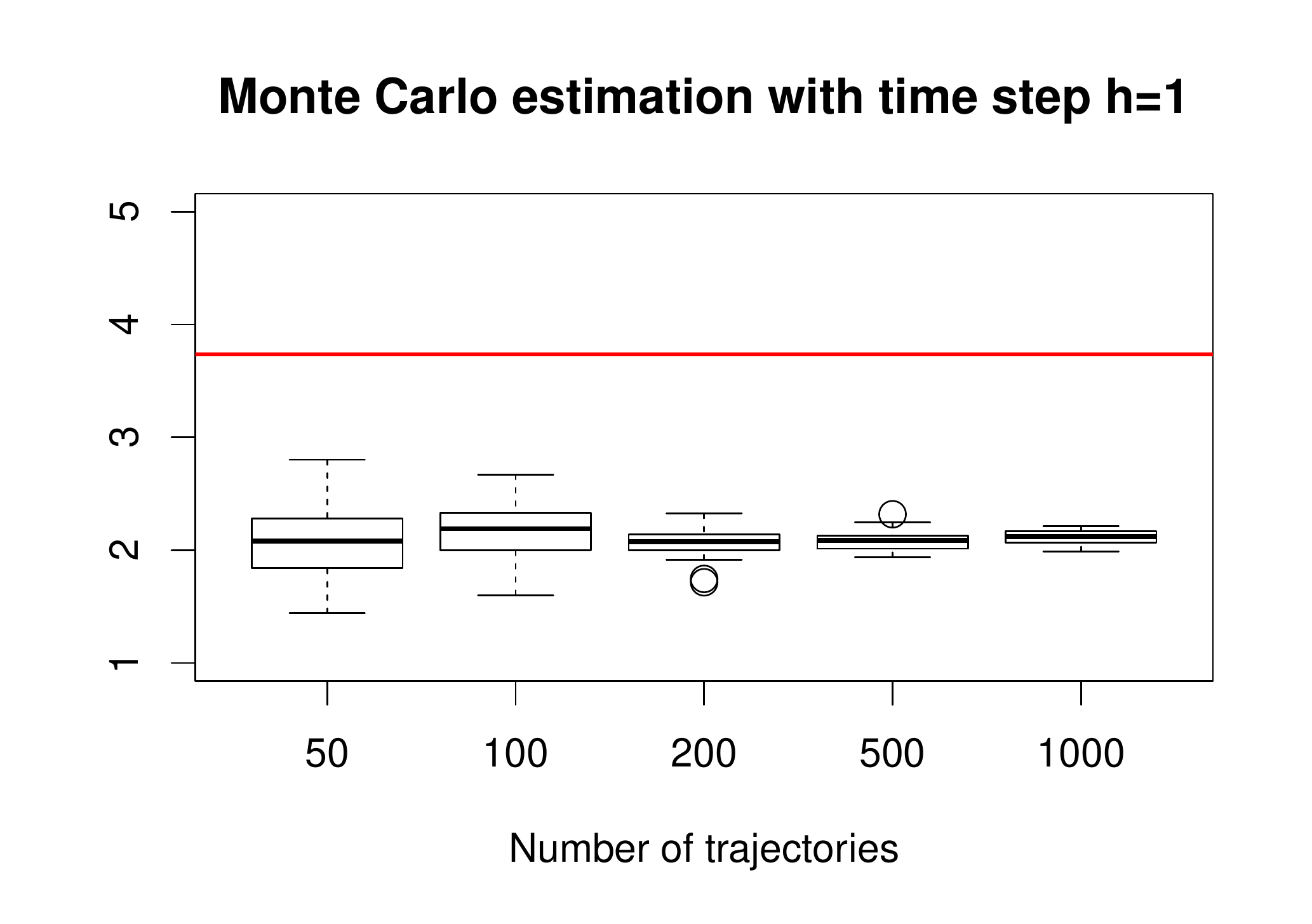}\includegraphics[width=5.5cm]{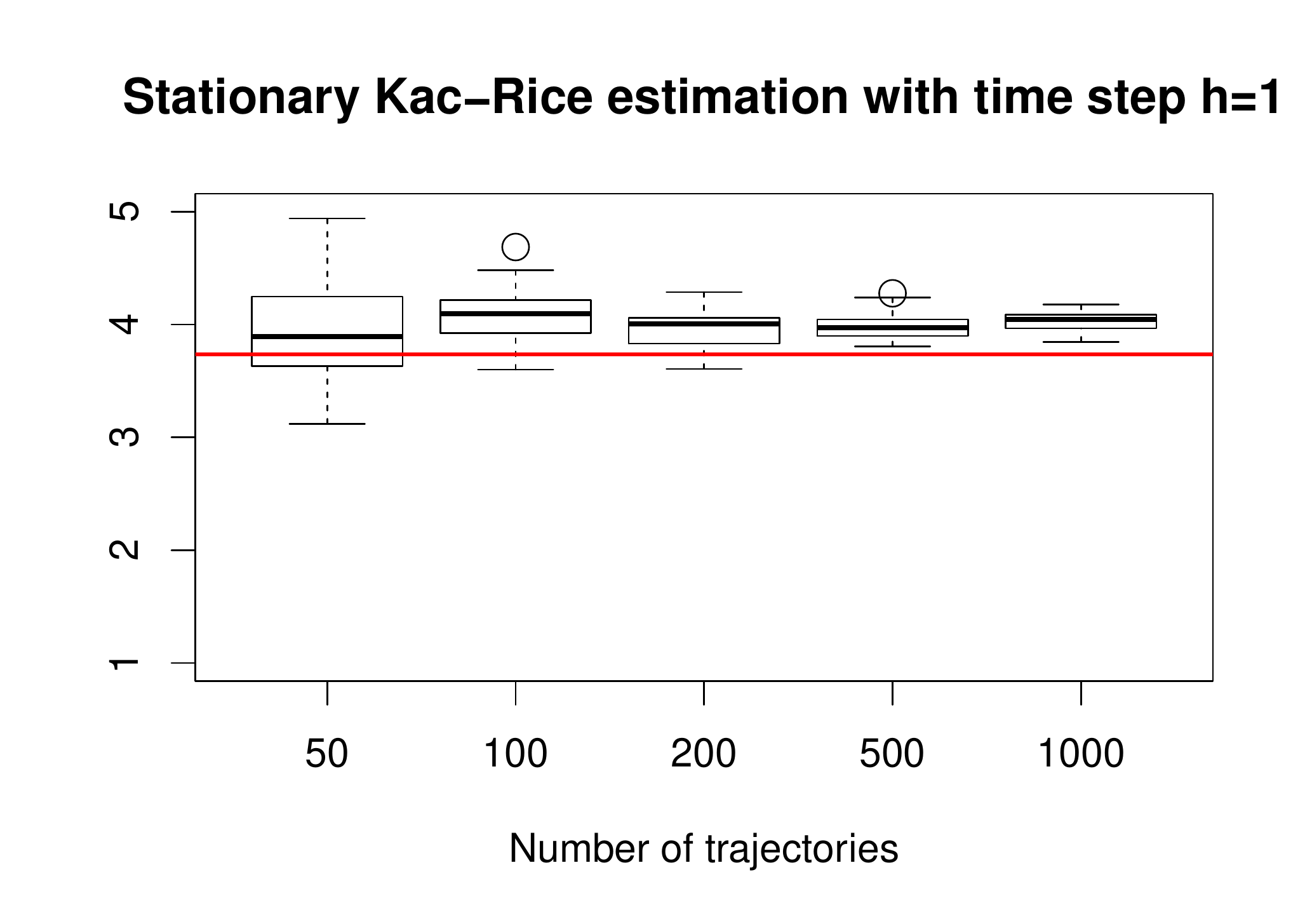}\includegraphics[width=5.5cm]{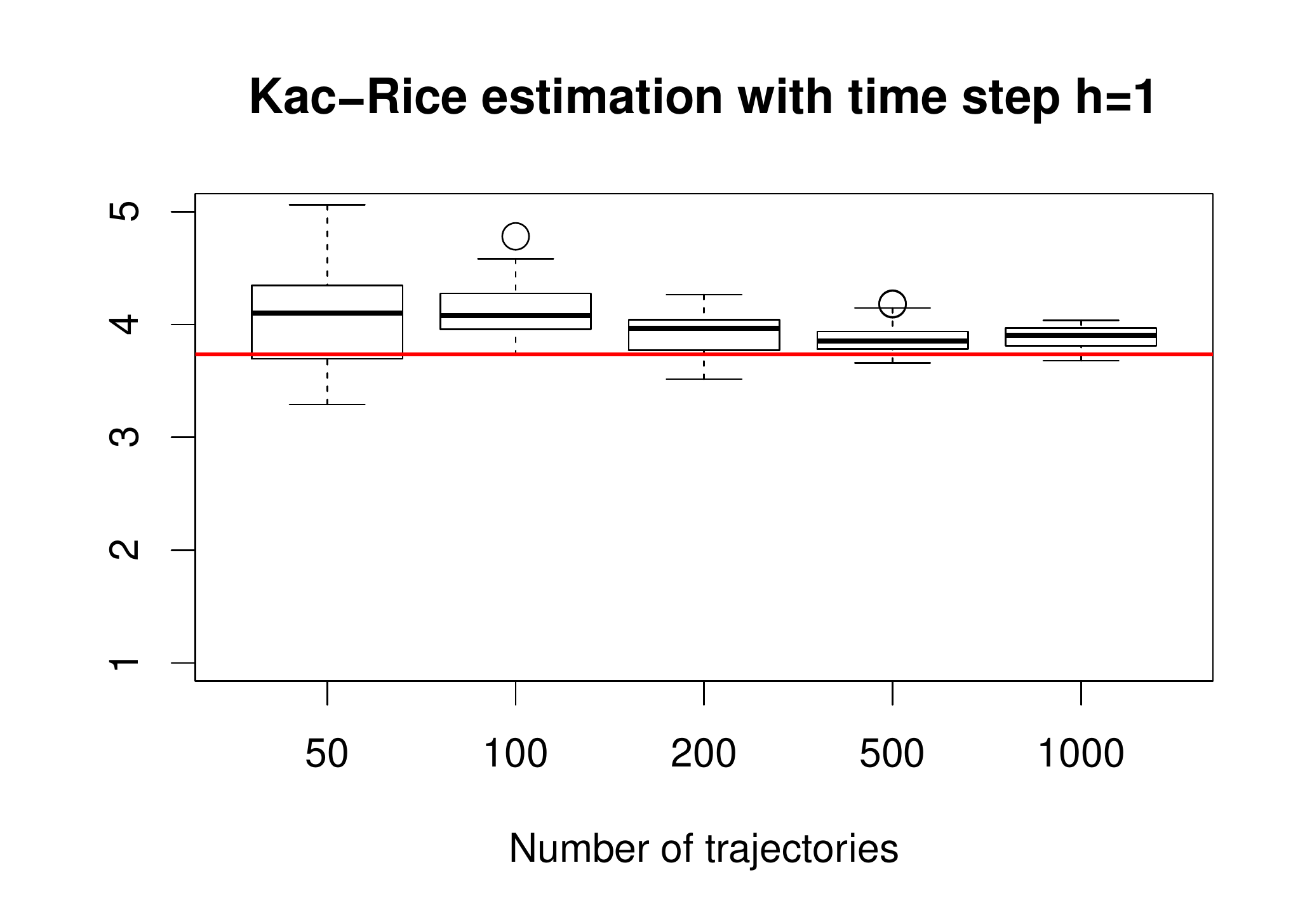}\\
\includegraphics[width=5.5cm]{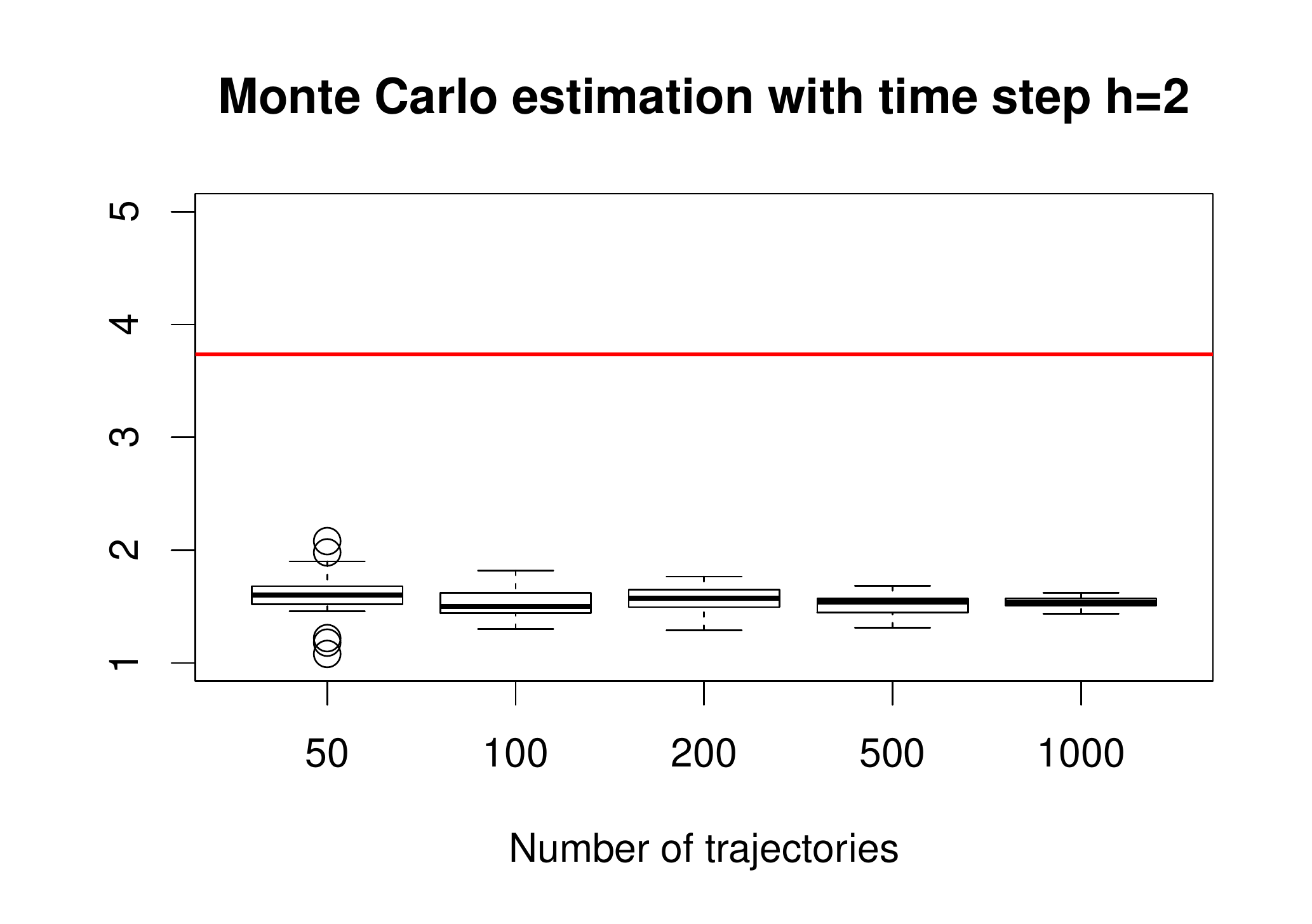}\includegraphics[width=5.5cm]{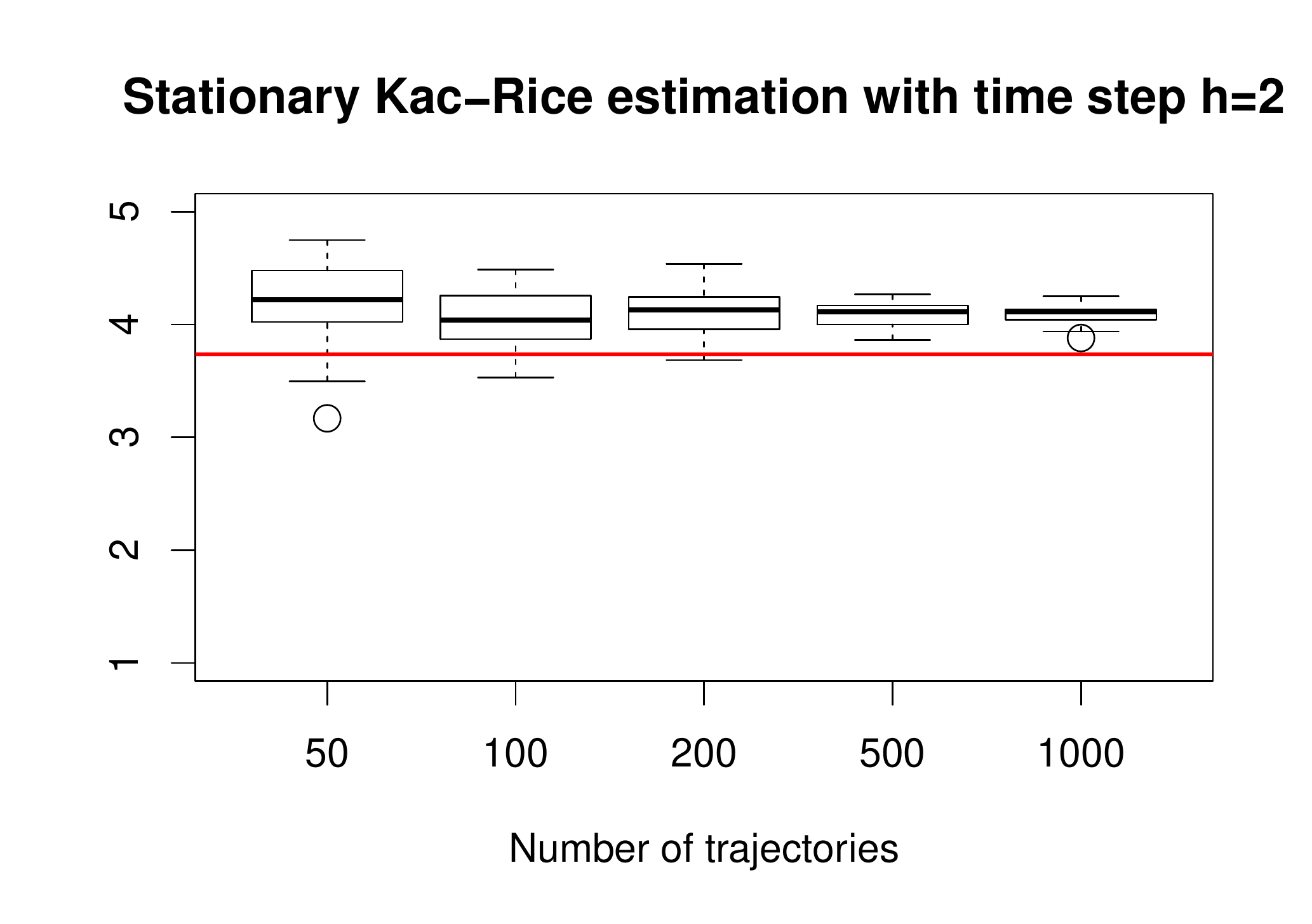}\includegraphics[width=5.5cm]{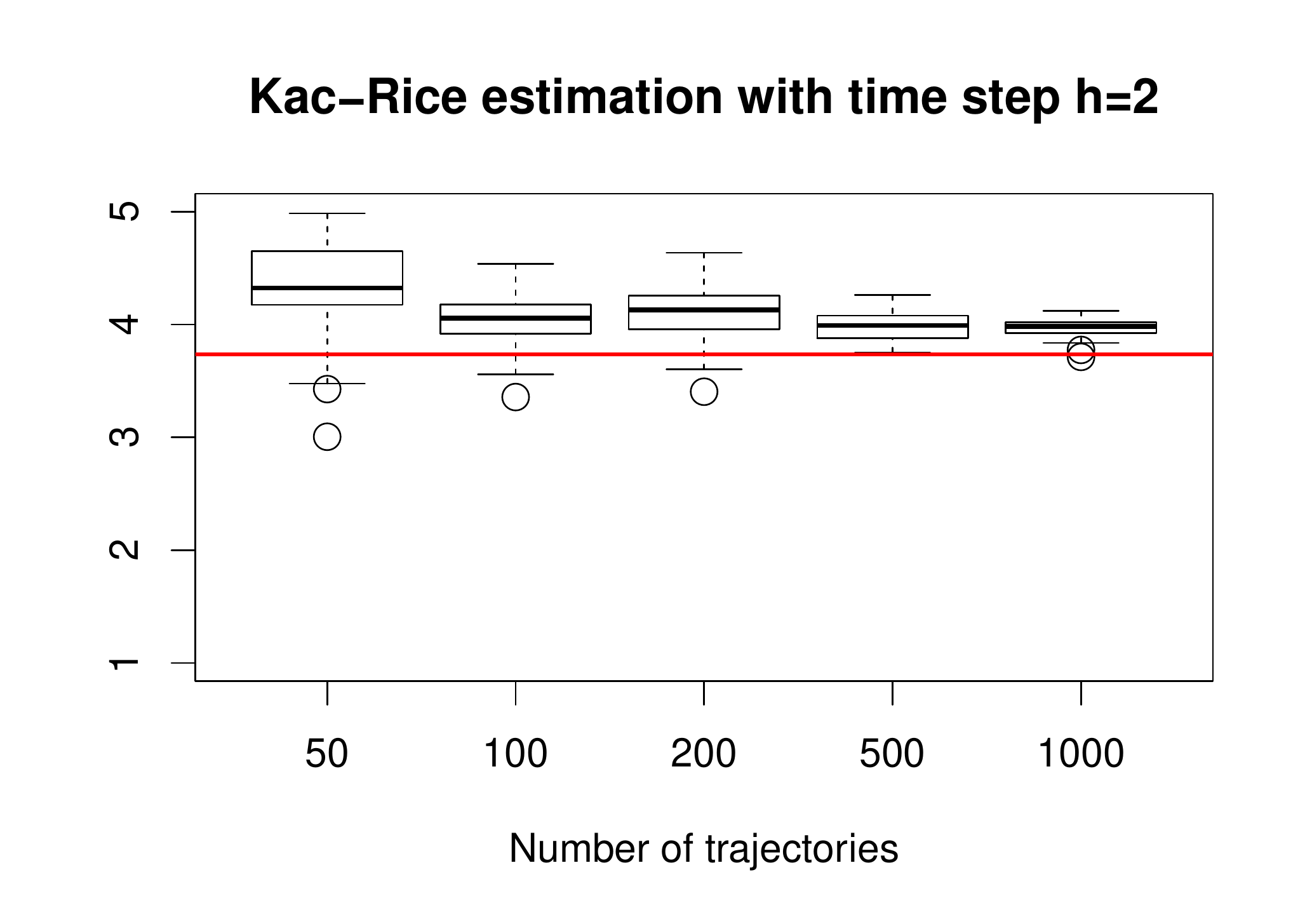}\\
\caption{Boxplots over $100$ replicates of Monte Carlo (left), stationary (middle) and non-stationary (right) Kac-Rice estimators of $C_2(50)$ for the telegraph process from data extracted from a temporal grid with an increasing step size: $h=0.01$, $h=0.1$, $h=1$ and $h=2$ from top to bottom.}
\label{fig:cvsimu}
\end{figure}

\begin{figure}[t]
\centering
\includegraphics[width=6.5cm]{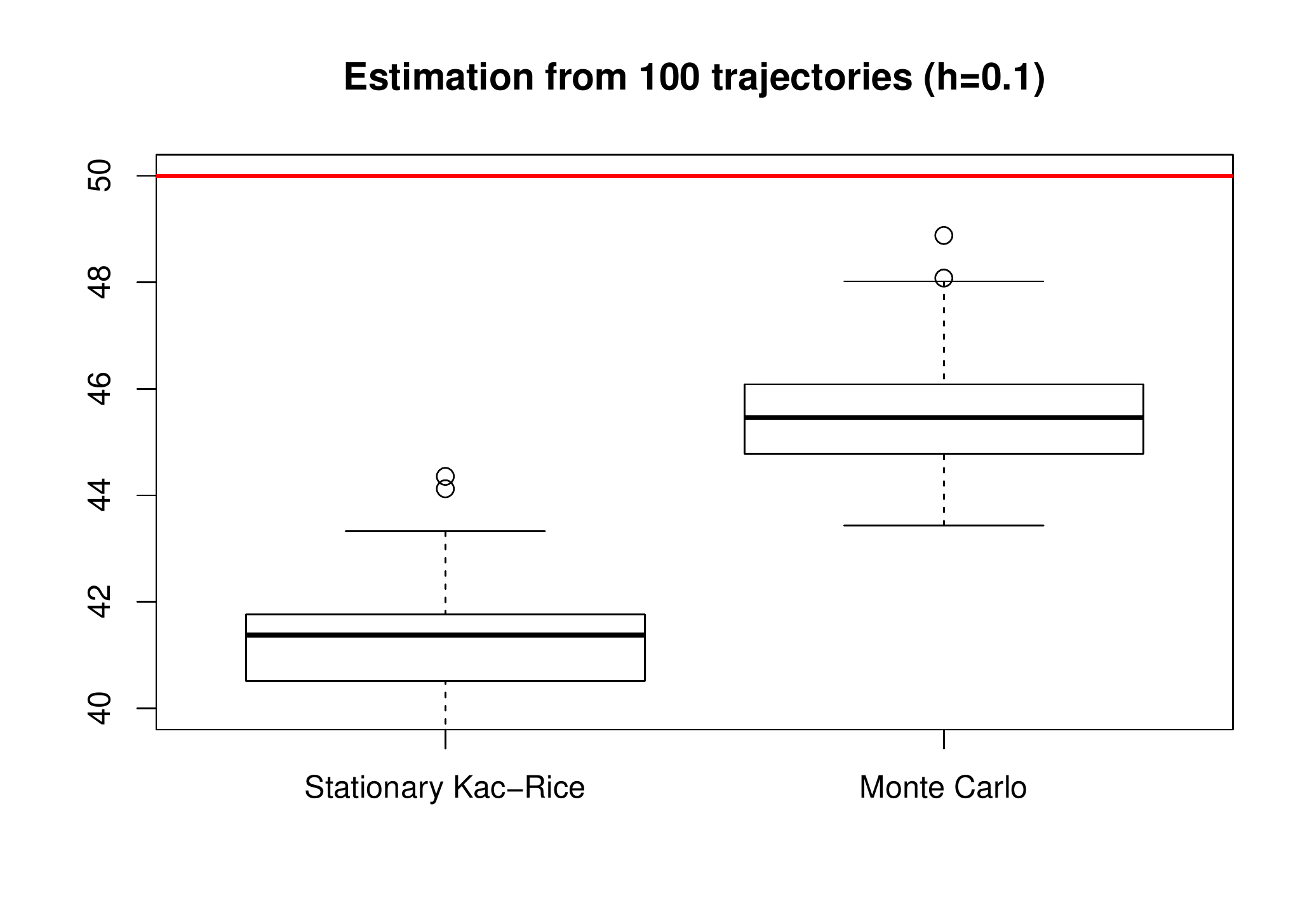}\qquad\qquad\includegraphics[width=6.5cm]{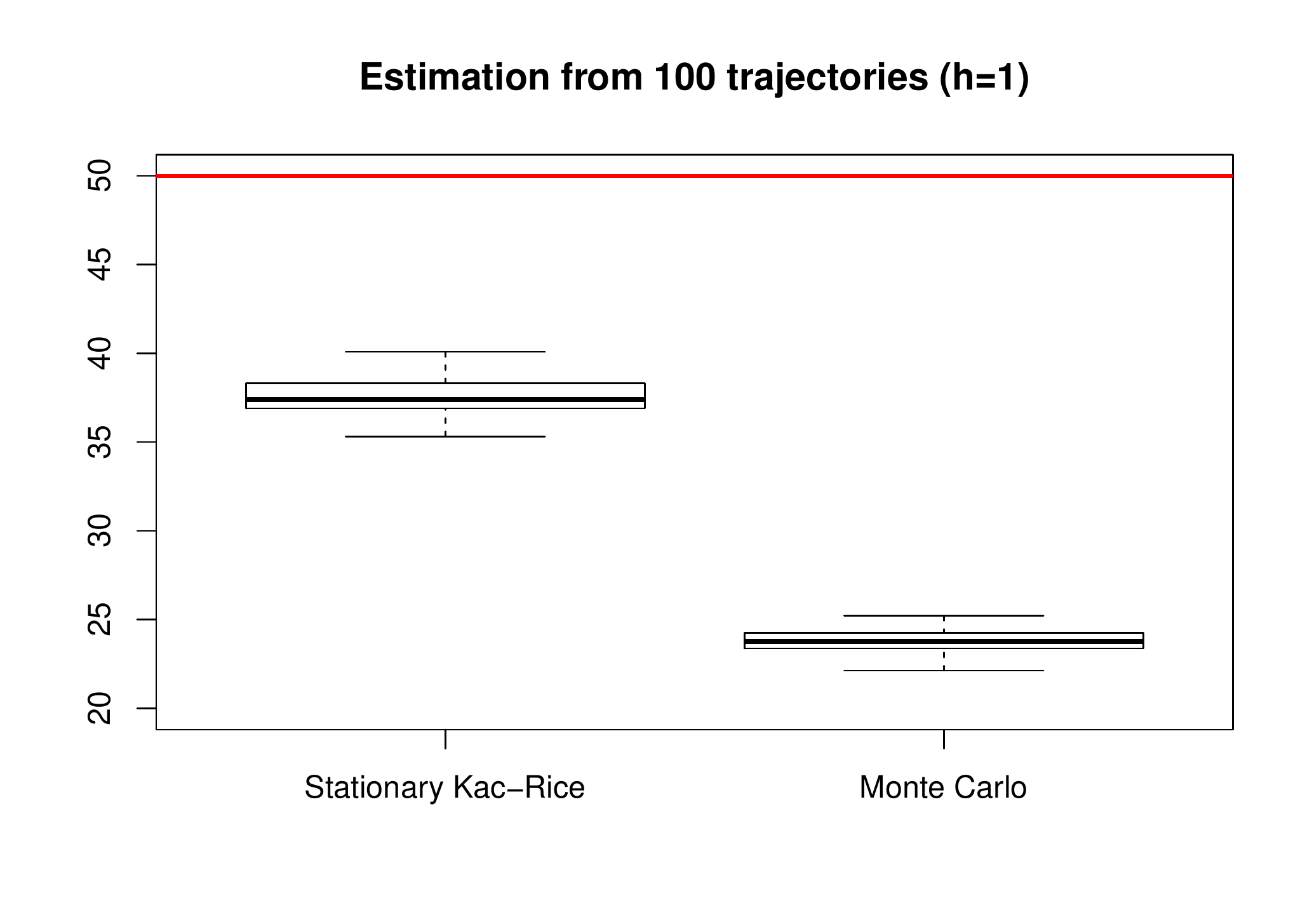}\\
\includegraphics[width=6.5cm]{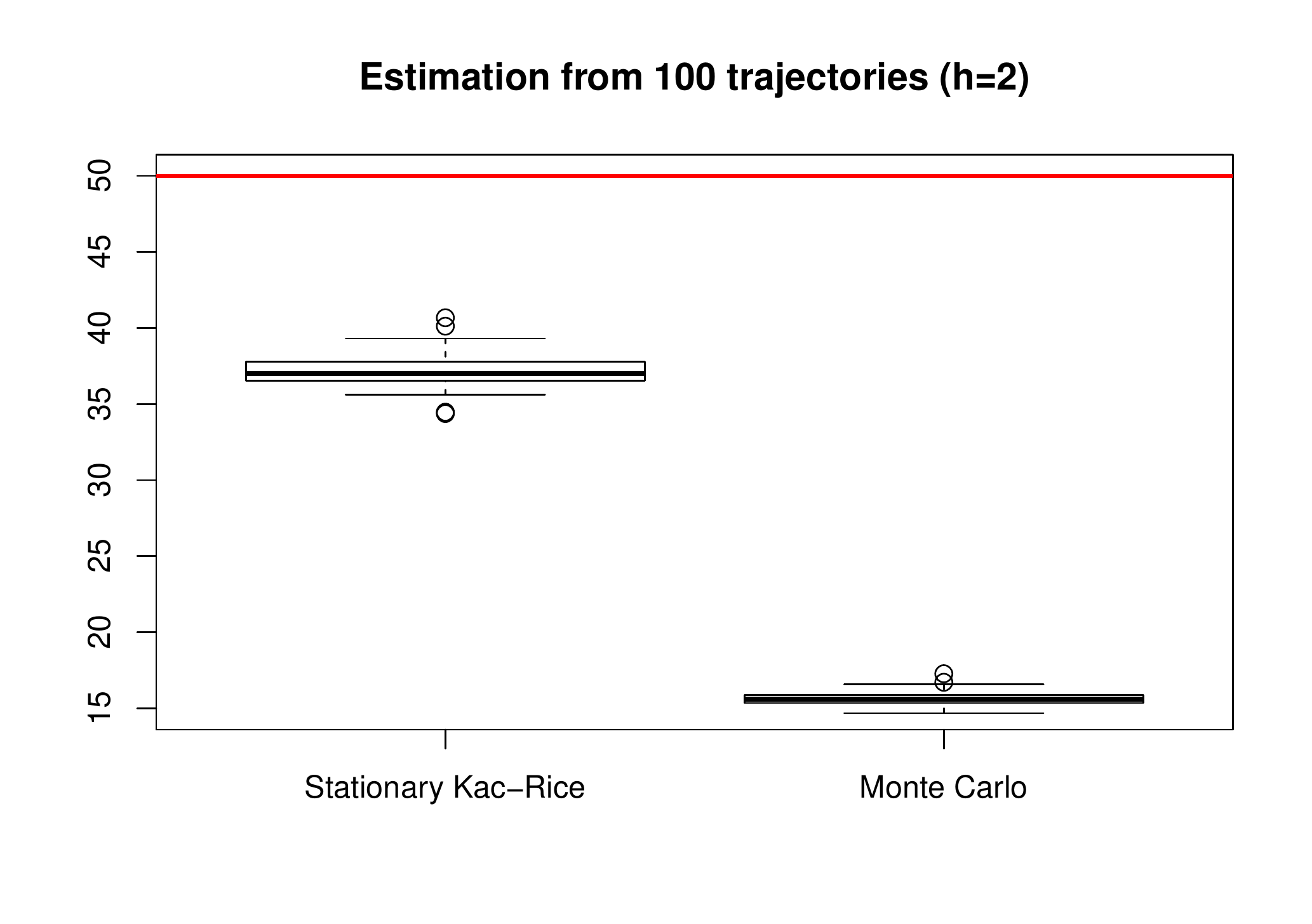}\qquad\qquad\includegraphics[width=6.5cm]{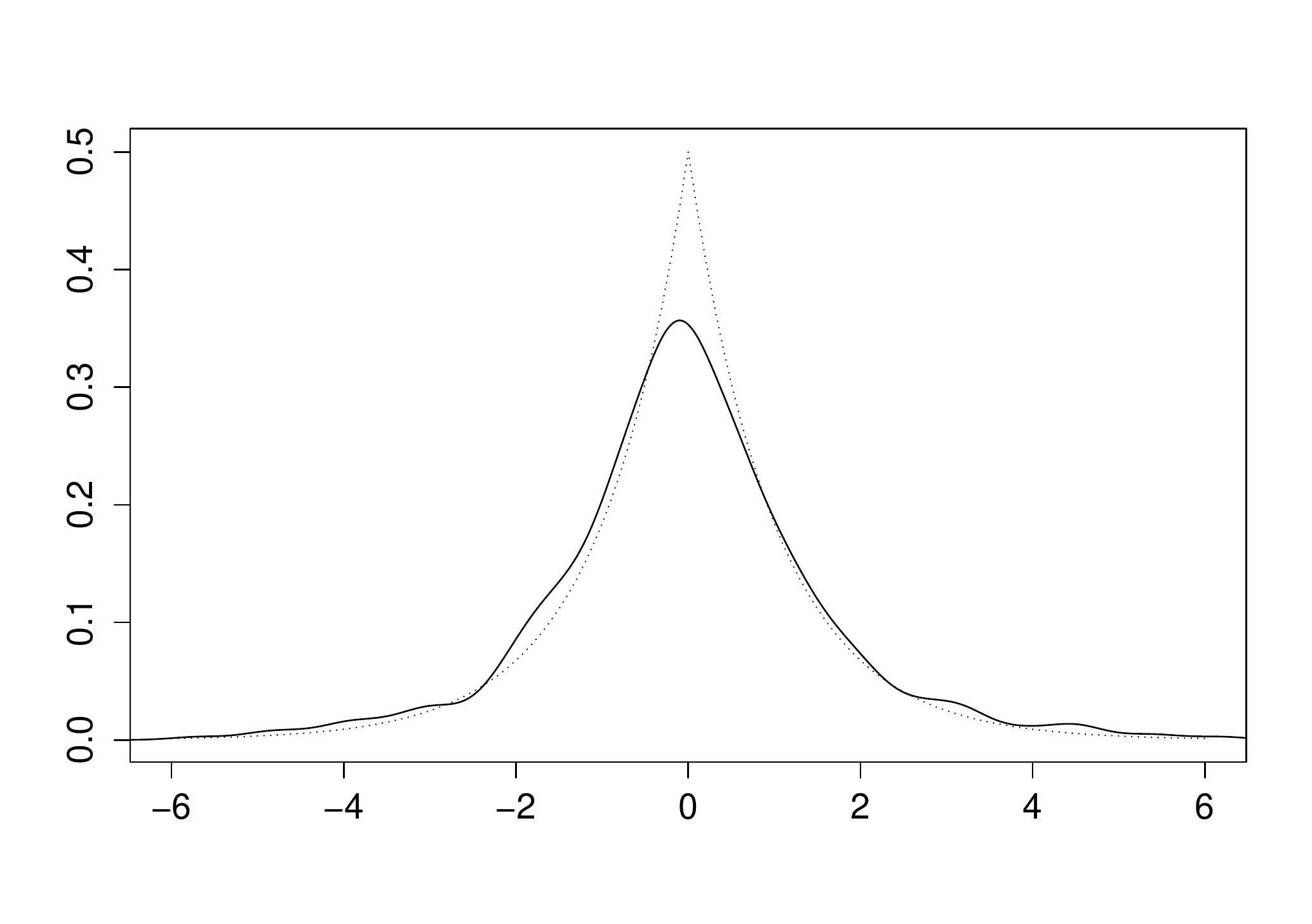}
\caption{Boxplots over $100$ replicates of Monte Carlo and stationary Kac-Rice estimators of $C_0(100)$ for the telegraph process from data extracted from a temporal grid with an increasing step size: $h=0.1$ (top left), $h=1$ (top right) and $h=2$ (bottom left), and comparison between the theoretical invariant distribution (dotted line) and its estimation (bottom right).}
\label{bp_step_size_s_0_telegraph_1d}
\end{figure}


\subsection{Piecewise deterministic simulated annealing algorithm}
\label{ss:pdsa}

We consider a piecewise deterministic version of the simulated annealing algorithm. This process, stated and studied in \cite{M16}, is a PDMP $(X,Y)$ valued in $\mathbb{R}\times\{-1,+1\}$ and evolving according to the generator,
$$\mathcal{L}f(x,y) = y \partial_xf(x,y) +\beta\left[yU'(x)\right]_+ \left(f(x,-y)-f(x,y)\right),$$
where $\beta$ is a positive parameter called the inverse of the temperature, and $U$ an energy potential. The Euclidean part $X$ evolves linearly with velocity switching between $-1$ and $+1$ according to the mode $Y$. These switchings occur at an $X$-dependent rate: there is no jump when the potential gradient and the mode have the same sign and there is a jump with intensity the absolute value of this same gradient times the inverse of the temperature $\beta$ otherwise, in the spirit of gradient descent algorithms. The process relaxes to equilibrium \cite[Theorem 1.4]{M16} with convergence in law towards the measure $\mu$ on $\mathbb{R}\times\{-1,+1\}$ defined by
$$\mu(\dd x\times\dd y)  = Z_\beta^{-1} \exp\left(-\beta U(x)\right)\dd x \otimes  \frac{1}{2}(\delta_{-1}+\delta_1)(\dd y),$$
where $Z_\beta$ is a renormalization constant. Thus, at equilibrium, the mode is equidistributed between $-1$ and $+1$ whereas the Euclidean variable samples the Gibb's measure associated to the energy potential $U$. In Figure \ref{pdsa:traj} is displayed a trajectory of $X$ until time $100$ when $U$ is the double-well potential defined on $\R$ by $U(x) = 0.05\,\left(x^4+x^3-4x^2\right)$. The figure also illustrates the concentration of the Euclidean variable around the minima of $U$.

\begin{figure}[th!]
\centering
\includegraphics[width=6.5cm]{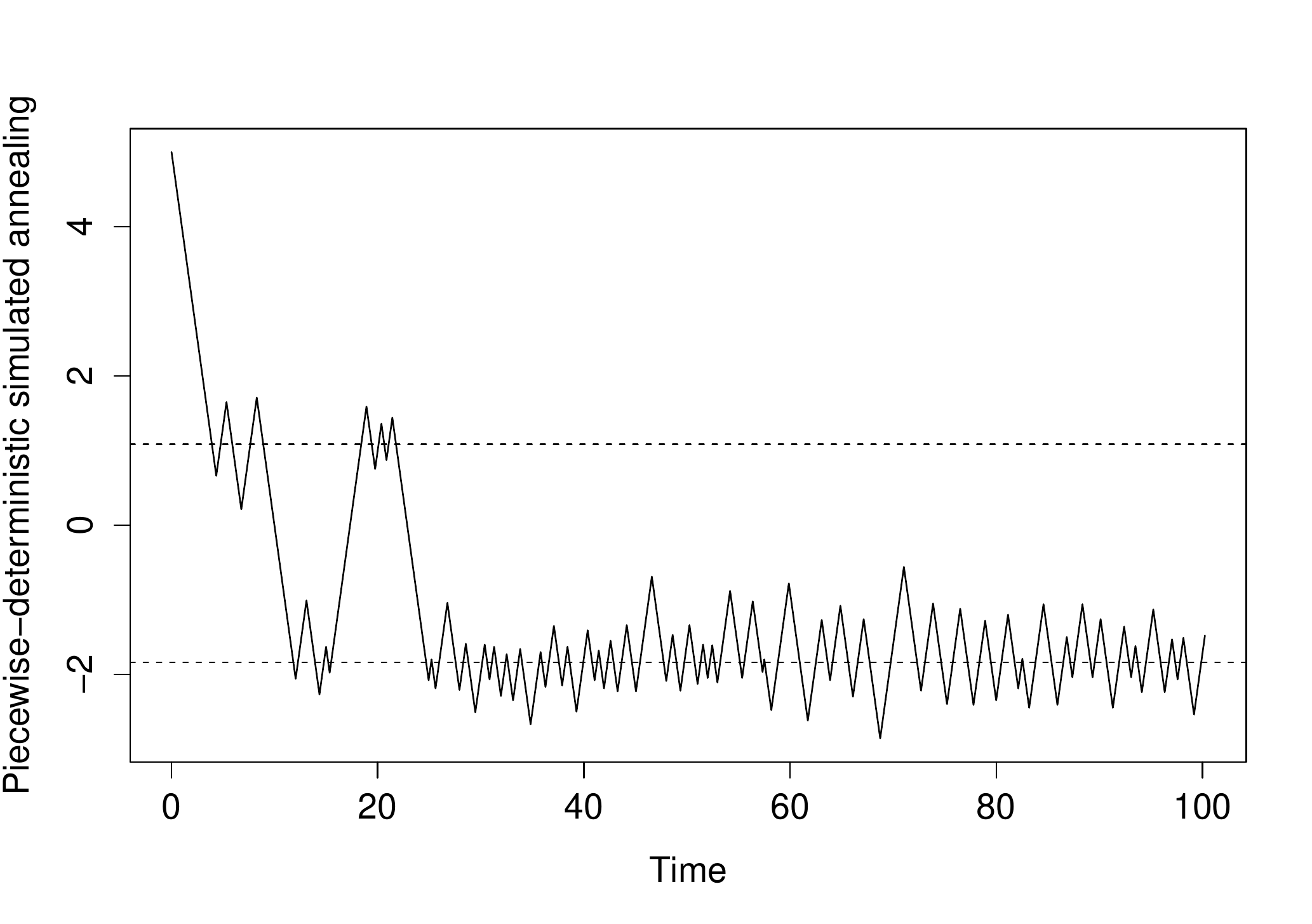}
\qquad\qquad\includegraphics[width=6.5cm]{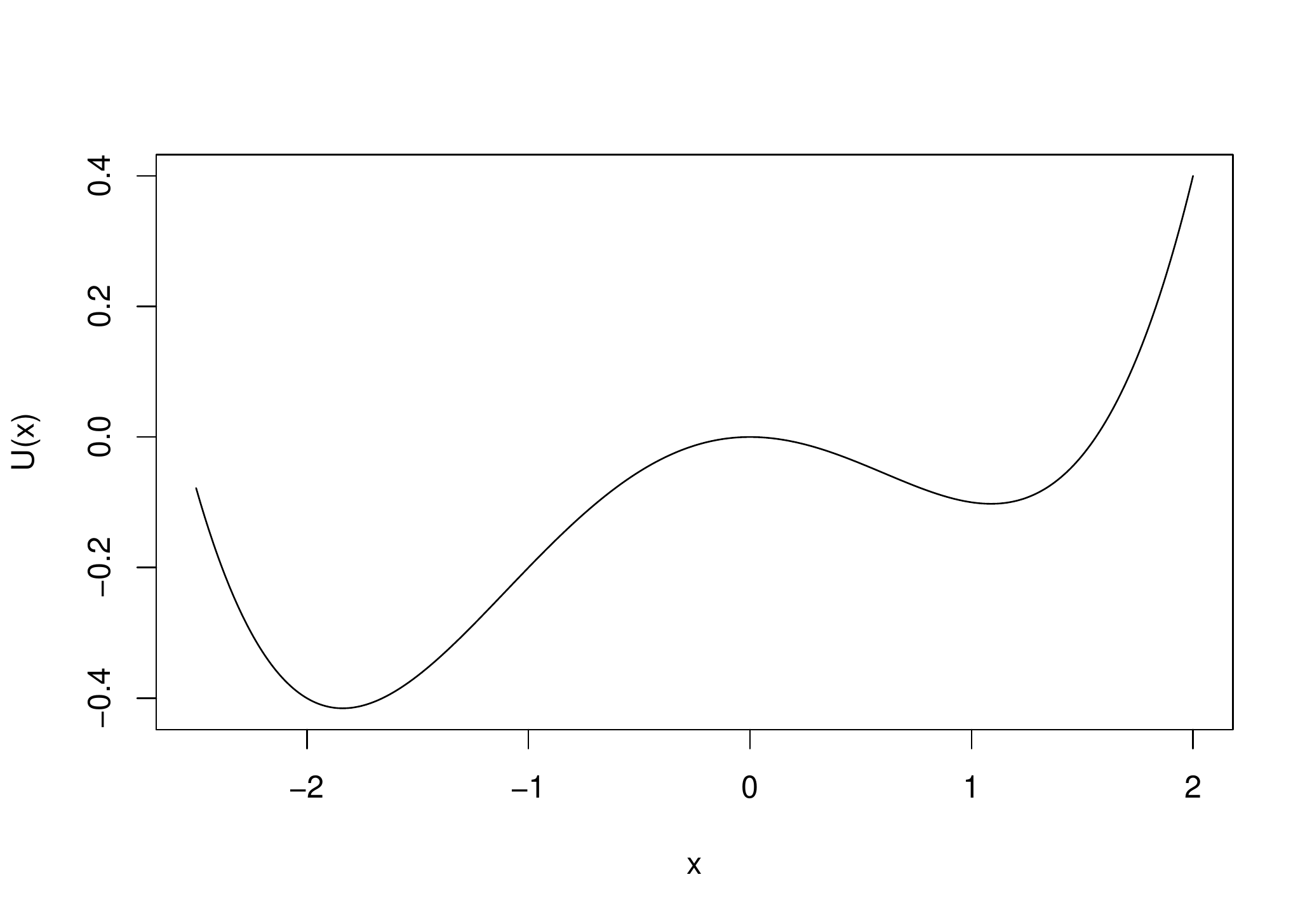}
\caption{A trajectory of the piecewise deterministic simulated annealing with inverse of temperature $\beta=7$ (left) and double-well potential $U:x\mapsto 0.05\,(x^4+x^3-4x^2)$ (right).\label{pdsa:traj}}
\end{figure}

\smallskip

\noindent
The behavior of the process when the Euclidean variable is in a neighborhood of $0$ is quite peculiar: there is no jump. Indeed, the gradient of $U$ at $0$ is null and therefore $X$ does not change its direction around $0$. As a consequence, the Monte Carlo estimator should perform quite well at this level, even with reasonably large discretization step sizes, since no crossings should be missed. This is what is indeed observed in Figure \ref{pdsa:level0} where are displayed boxplots over $100$ replicates of Monte Carlo and non-stationary Kac-Rice estimators of the average number of crossings of the level $0$, $C_0(100)$, with an increasing time step size $h\in\{0.1,1,2\}$. In this case, the counting technique performs better than the Kac-Rice-based methods.

\begin{figure}[th!]
\centering
\includegraphics[width=6.5cm]{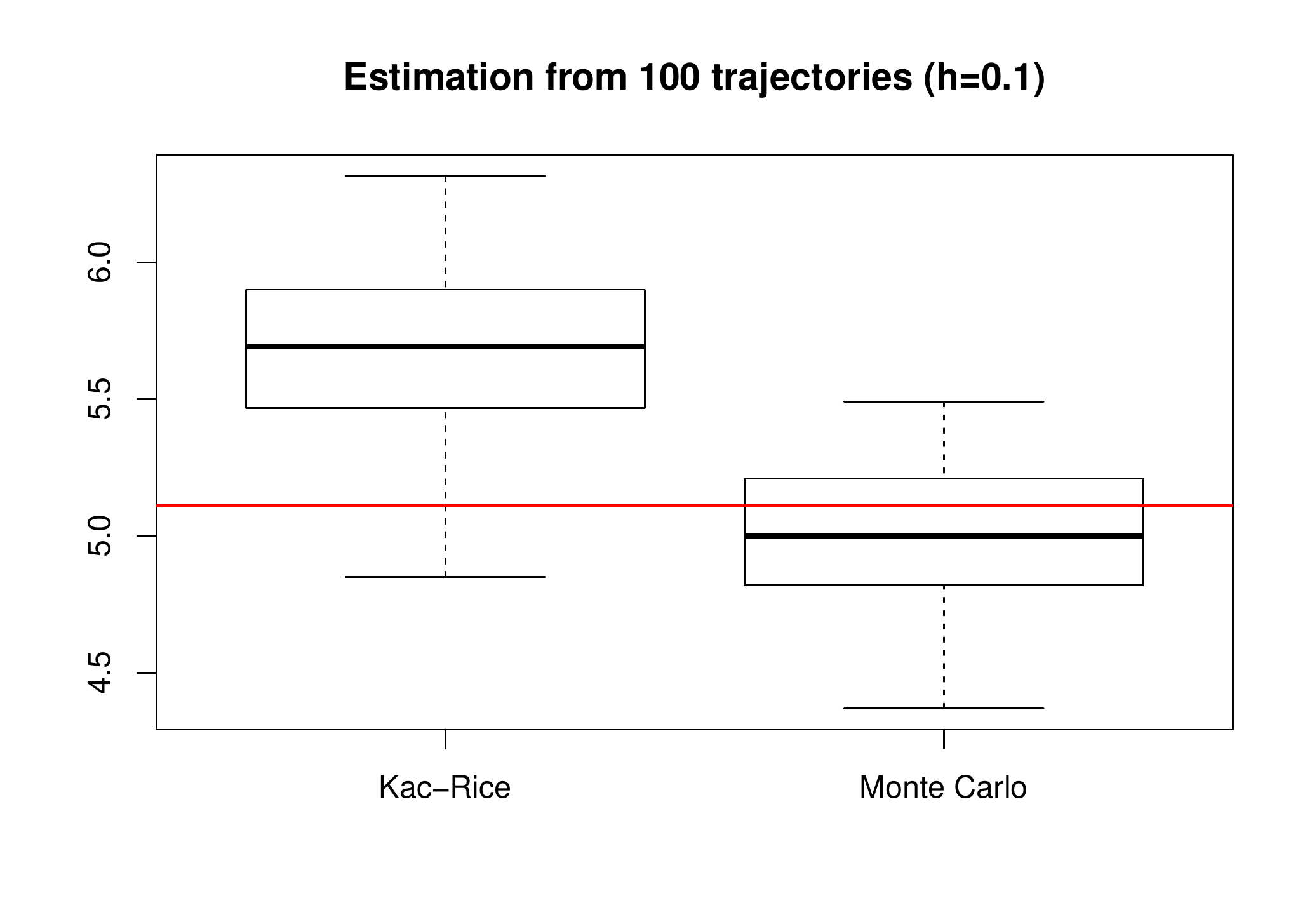}\qquad\qquad\includegraphics[width=6.5cm]{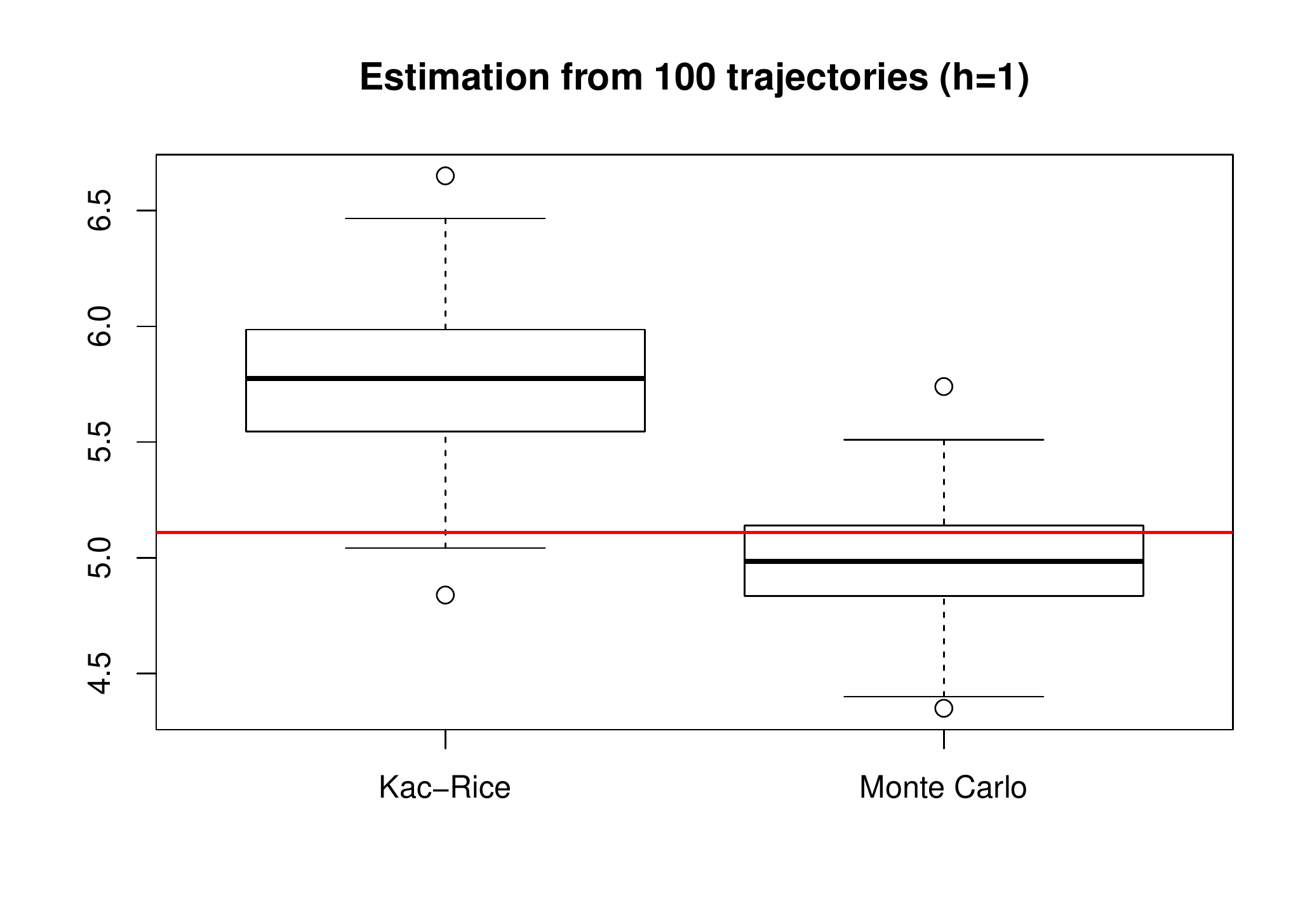}\\
\qquad\quad\includegraphics[width=6.5cm]{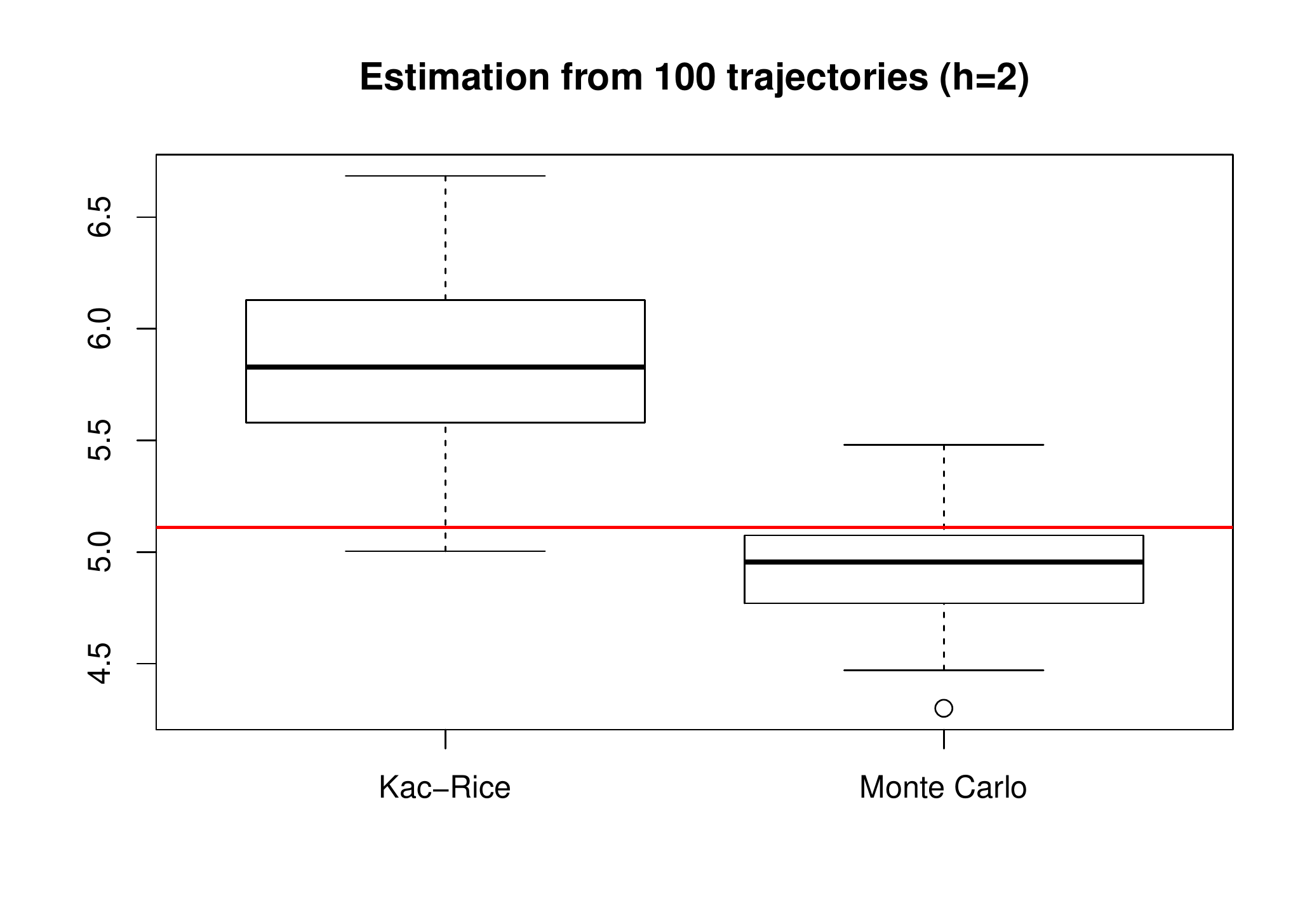}\hfill$~$
\caption{Boxplots over $100$ replicates of Monte Carlo and non-stationary Kac-Rice estimators of $C_0(100)$ for the piecewise deterministic simulated annealing process from data extracted from a temporal grid with an increasing step size: $h=0.1$ (top left), $h=1$ (top right) and $h=2$ (bottom left).}
\label{pdsa:level0}
\end{figure}

\smallskip

\noindent
In Figure \ref{pdsa:level1} are displayed boxplots of Monte Carlo and non-stationary Kac-Rice estimators of the average number of crossings of the level $1$ $C_1(100)$ with an increasing time step size $h\in\{0.1,1,2\}$ and an increasing number of trajectories $n\in\{50,100\}$. For a small time step size and a small number of trajectories, the Monte Carlo estimator better estimates the average number of crossings of $1$ than the Kac-Rice-based method. This observation is reversed when the step size and the number of observation grow: when $h=2$ and $n=200$, the empirical estimator is dramatically far from the true average number of crossings whereas the non-stationary Kac-Rice estimator is still accurate.

\begin{figure}[th!]
\centering
\includegraphics[width=6.5cm]{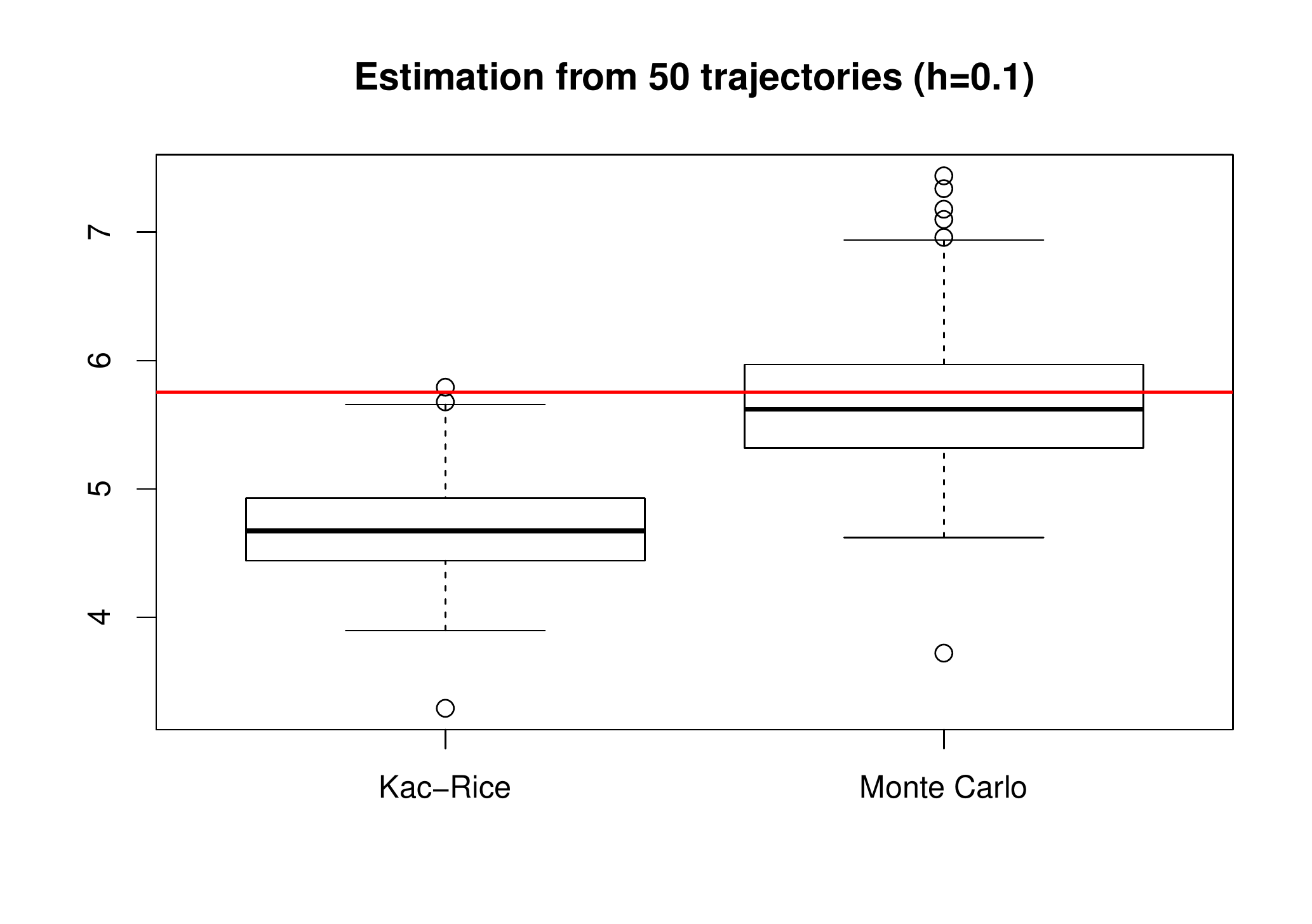}\qquad\qquad\includegraphics[width=6.5cm]{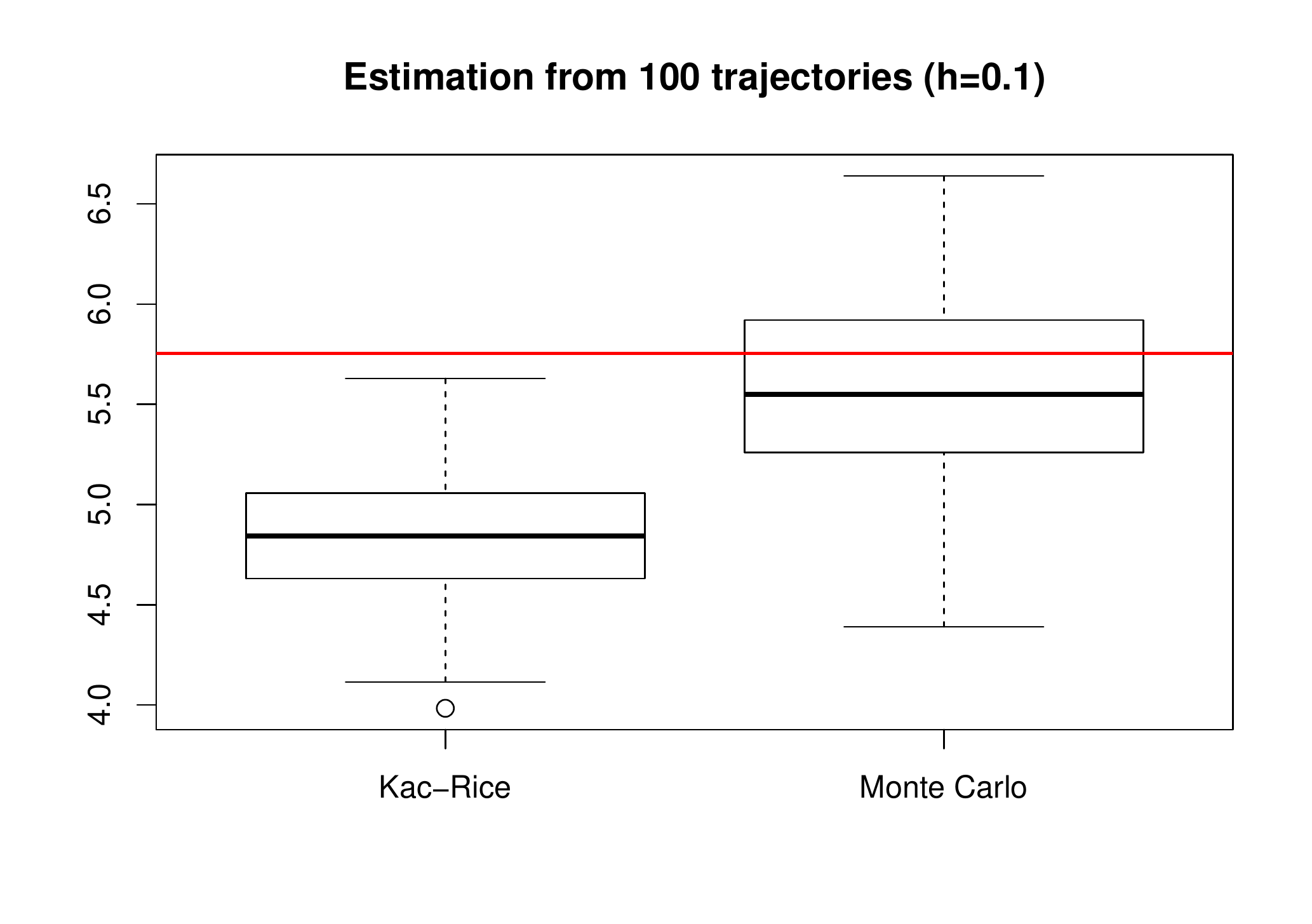}\\
\includegraphics[width=6.5cm]{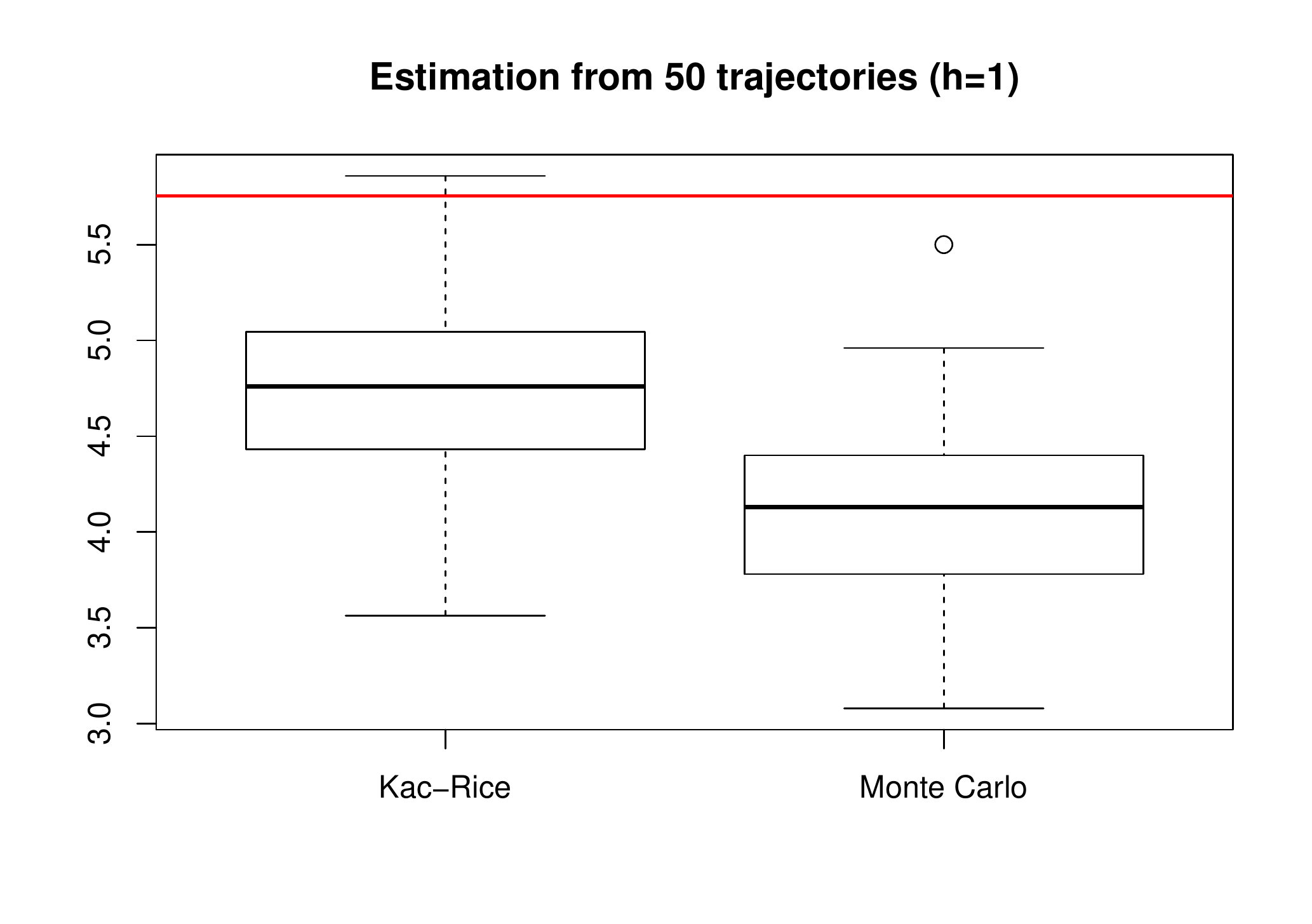}\qquad\qquad\includegraphics[width=6.5cm]{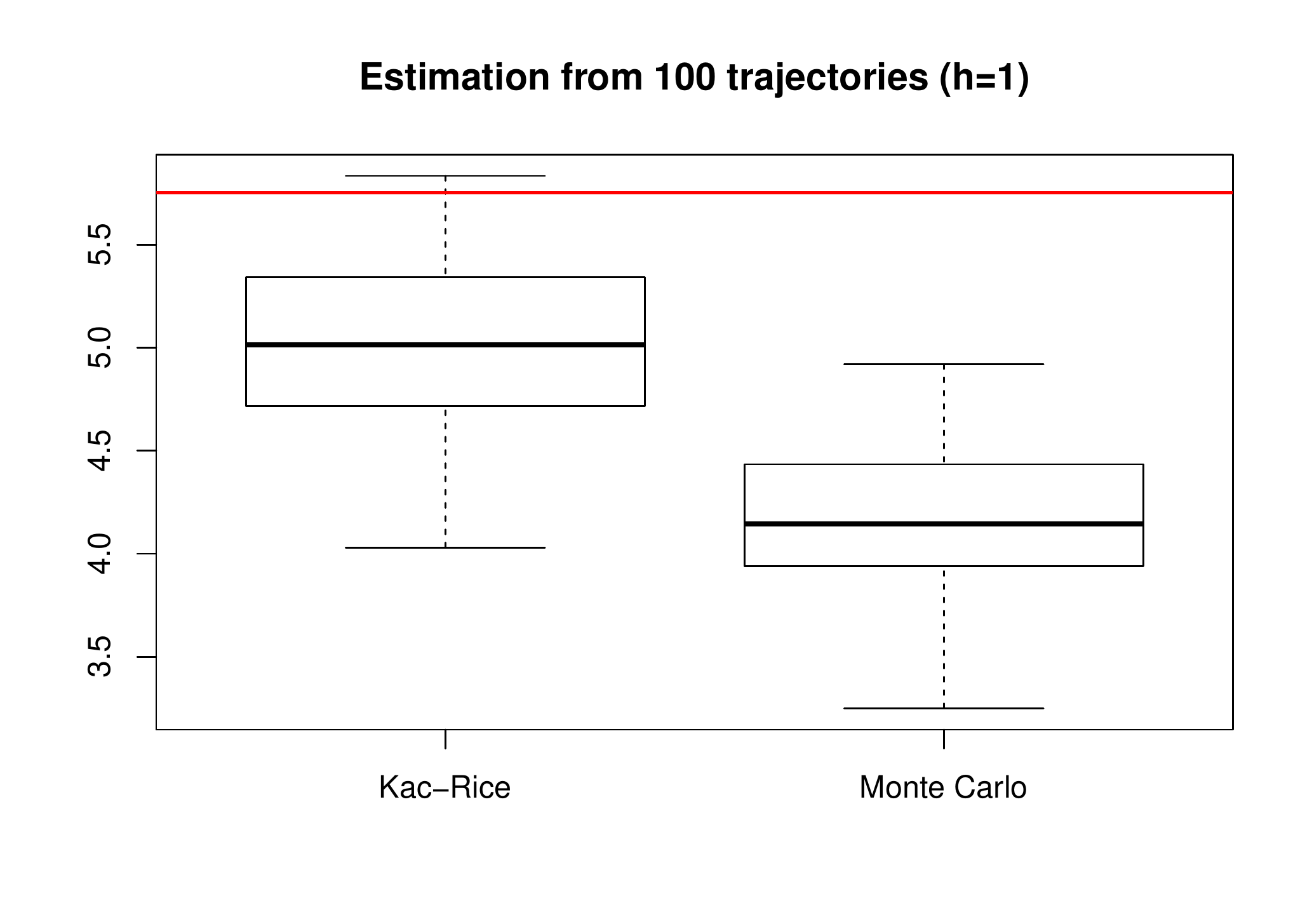}\\
\includegraphics[width=6.5cm]{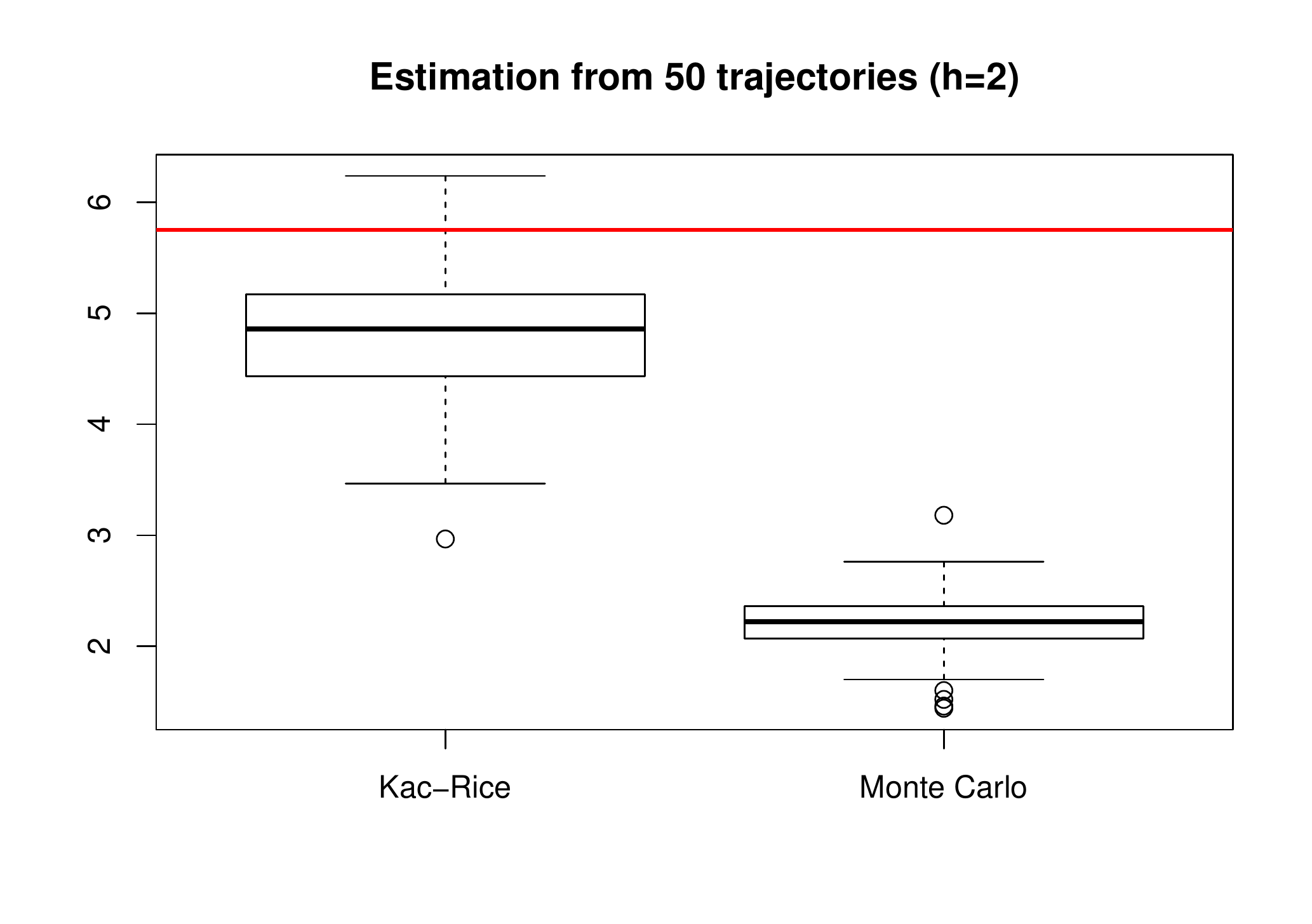}\qquad\qquad\includegraphics[width=6.5cm]{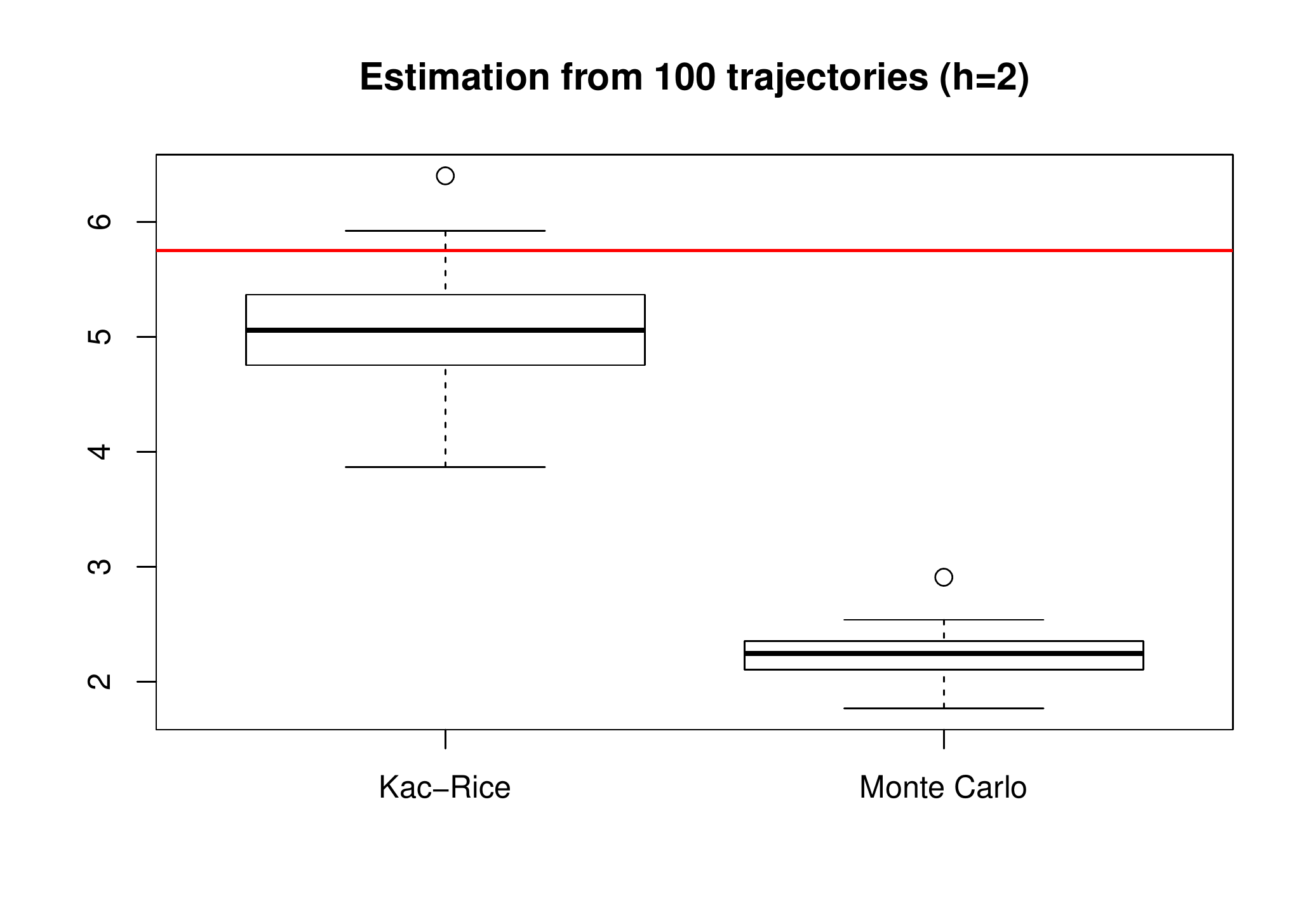}
\caption{Boxplots over $100$ replicates of Monte Carlo and non-stationary Kac-Rice estimators of $C_1(100)$ for the piecewise deterministic simulated annealing process from data extracted from a temporal grid with an increasing step size: $h=0.1$ (top), $h=1$ (middle) and $h=2$ (bottom), and an increasing number of trajectories: $n=50$ (left) and $n=100$ (right).}
\label{pdsa:level1}
\end{figure}

\subsection{Two-dimensional telegraph process}
\label{tele:2d}

We consider a two-dimensional variant of the telegraph process. Still in the Euclidean-mode setting, the considered process is valued in ${\cal X}\times{\cal Y}=\mathbb{R}^2\times\{\text{N},\text{S},\text{E},\text{W}\}$. The four letters $\{\text{N},\text{S},\text{E},\text{W}\}$ stand for the cardinal directions, that is, in the same order, for the vectors
$$\{(0,1), (0,-1),(1,0),(-1,0)\}$$
of $\mathbb{R}^2$. For the sake of readability, we also denote by $\text{NE}$, $\text{SE}$, $\text{SW}$ and $\text{NW}$ the four quadrants
$$
\left\{\begin{array}{ccc}
\text{NE}&=&\{(x,y)\in\mathbb{R}^2~:~x\geq0, y\geq0\},\\
\text{SE}&=&\{(x,y)\in\mathbb{R}^2~:~x\geq0, y<0\},\\
\text{SW}&=&\{(x,y)\in\mathbb{R}^2~:~x<0, y\leq 0\},\\
\text{NW}&=&\{(x,y)\in\mathbb{R}^2~:~x<0, y>0\}.
\end{array}\right.
$$
With these notations, the considered process is a piecewise deterministic Markov process evolving according to the generator
$$\mathcal{L}f(x,y) = y\cdot\partial_xf(x,y) +\sum_{y'\in{\cal Y}}\lambda(x,y\rightarrow y')\left(f(x,y')-f(x,y)\right),$$
where the rate functions are given by
$$
\left\{
\begin{array}{ccccc}
\lambda(x,y\rightarrow\text{N})&=&(\mathbb{1}_{\text{W}}(y)+b\,\mathbb{1}_{\text{S}}(y))\,\mathbb{1}_{\text{SW}}(x)&+&a\mathbb{1}_{\text{S}}(y)\mathbb{1}_{\text{NE}}(x),\\
\lambda(x,y\rightarrow\text{S})&=&(\mathbb{1}_{\text{E}}(y)+b\,\mathbb{1}_{\text{N}}(y))\,\mathbb{1}_{\text{NE}}(x)&+&a\mathbb{1}_{\text{N}}(y)\mathbb{1}_{\text{SW}}(x),\\
\lambda(x,y\rightarrow\text{E})&=&(\mathbb{1}_{\text{N}}(y)+b\,\mathbb{1}_{\text{W}}(y))\,\mathbb{1}_{\text{NW}}(x)&+&a\mathbb{1}_{\text{W}}(y)\mathbb{1}_{\text{SE}}(x),\\
\lambda(x,y\rightarrow\text{W})&=&(\mathbb{1}_{\text{S}}(y)+b\,\mathbb{1}_{\text{E}}(y))\,\mathbb{1}_{\text{SE}}(x)&+&a\mathbb{1}_{\text{E}}(y)\mathbb{1}_{\text{NW}}(x),
\end{array}\right.
$$
with parameters $b>a>0$. The Euclidean component of the process evolves in parallel to one of the main axes and in a similar fashion in each quadrant of the plane. For instance, when it enters the $\text{NE}$ quadrant, it can only do it following the $\text{E}$ direction. Then, its direction switches to the $\text{S}$ direction at rate $\lambda$ and then switches between $\text{N}$ and $\text{S}$ at rate $a$ and $b$, respectively. A trajectory of the process is displayed in Figure \ref{traj_telegraph_2d}.

\begin{figure}[ht]
\centering
\includegraphics[width=7cm]{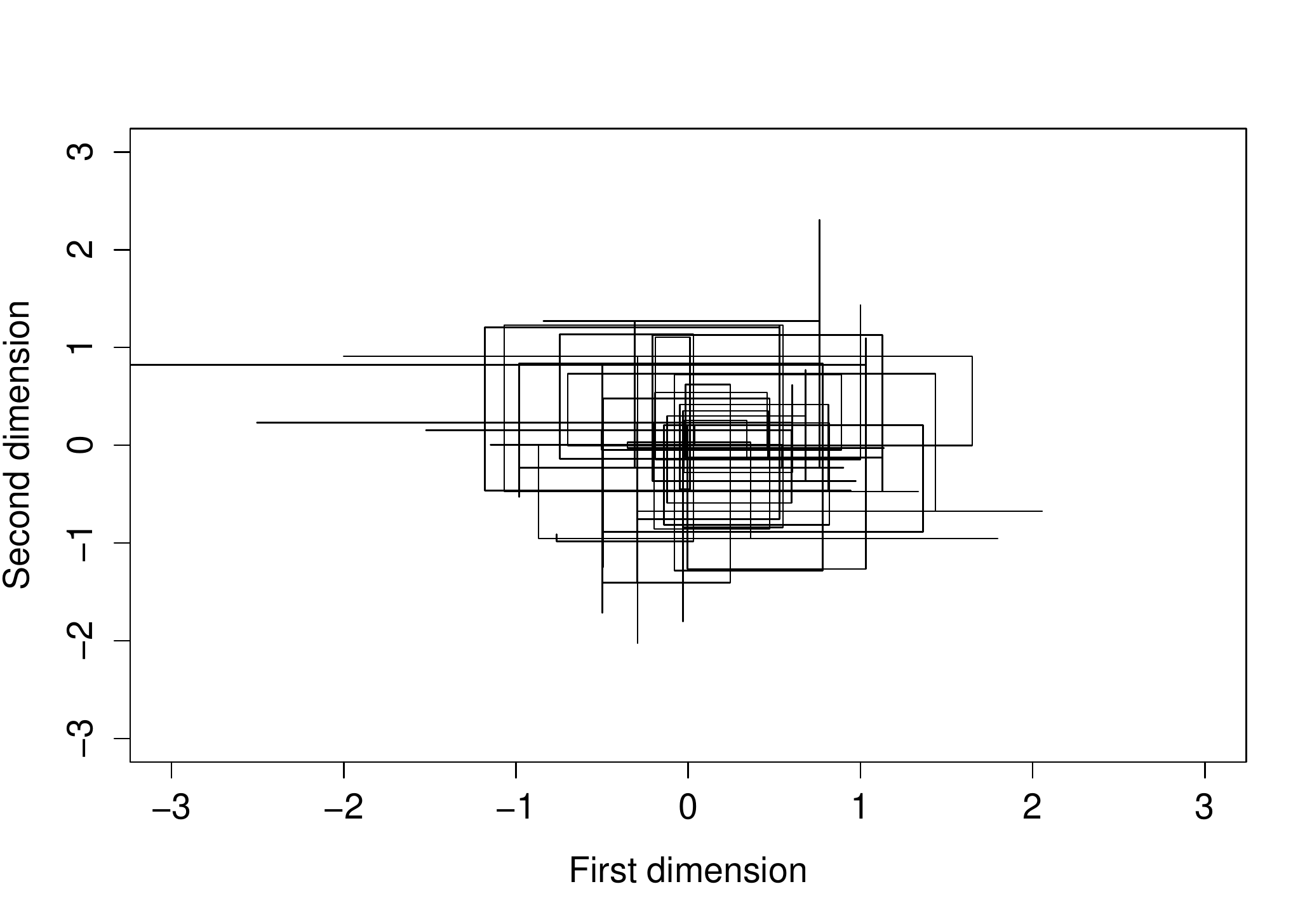}
\caption{A trajectory of the two-dimensional telegraph process until time $H=200$.}\label{traj_telegraph_2d}
\end{figure}

\noindent
We consider the average number of crossings of the square $S_c$ with vertices 
$$\{(-c,-c),(-c,c),(c,-c),(c,c)\}.$$
We compare the Monte Carlo and non-stationary Kac-Rice estimators when the horizon time is $H=20$, the number of trajectories is either $n=200$ or $n=400$ and the time step size is $h=1$. It should be noted that integrals over the square $\int_{S_c} f(x) \sigma_1(\dd x)$ have been computed as four line integrals,
$$
\int_{S_c} f(x) \sigma_1(\dd x)=\int_{-c}^c f(-c,x_2)\,\dd x_2+\int_{-c}^{c} f(c,x_2)\,\dd x_2+\int_{-c}^c f(x_1,c)\,\dd x_1+\int_{-c}^c f(x_1,-c)\,\dd x_1.
$$
In Figure \ref{bp_step_size_s_2_telegraph_2d} are displayed boxplots over $100$ replicates of Monte Carlo and non-stationary Kac-Rice estimators of the average number of crossings of the square $S_2$ from an increasing number of available trajectories $n\in\{200,400\}$. The Kac-Rice-based method performs for both numbers of trajectories better than the counting estimator: the Monte Carlo estimator misses some crossings given the discrete nature of the temporal grid. It should be remarked that the bias does not depend on the number of data.
The same conclusions hold true for the crossings of the square $S_3$ presented in Figure \ref{bp_step_size_s_3_telegraph_2d}.

\begin{figure}[t!h]
\centering
\includegraphics[width=6.5cm]{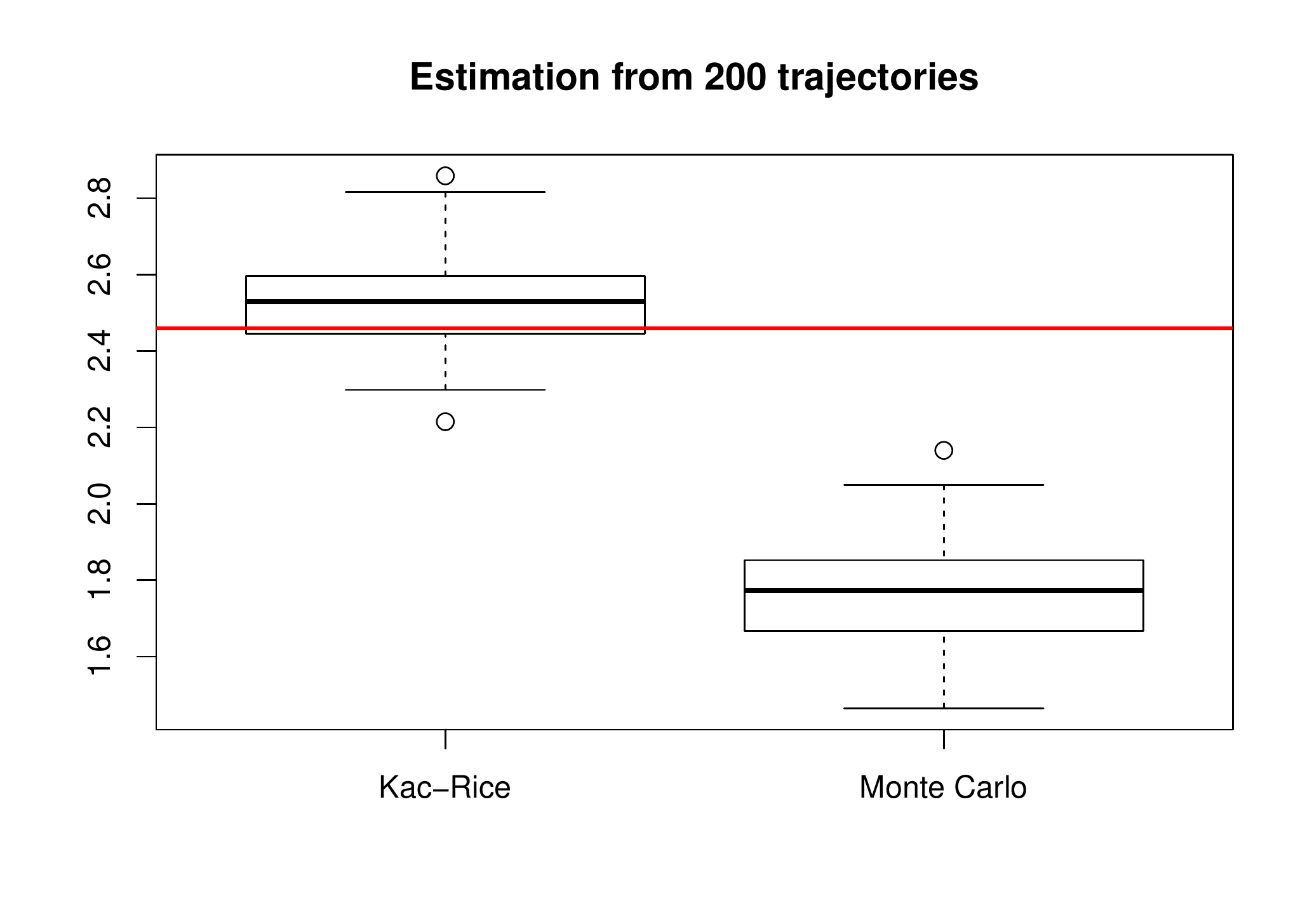}\qquad\qquad\includegraphics[width=6.5cm]{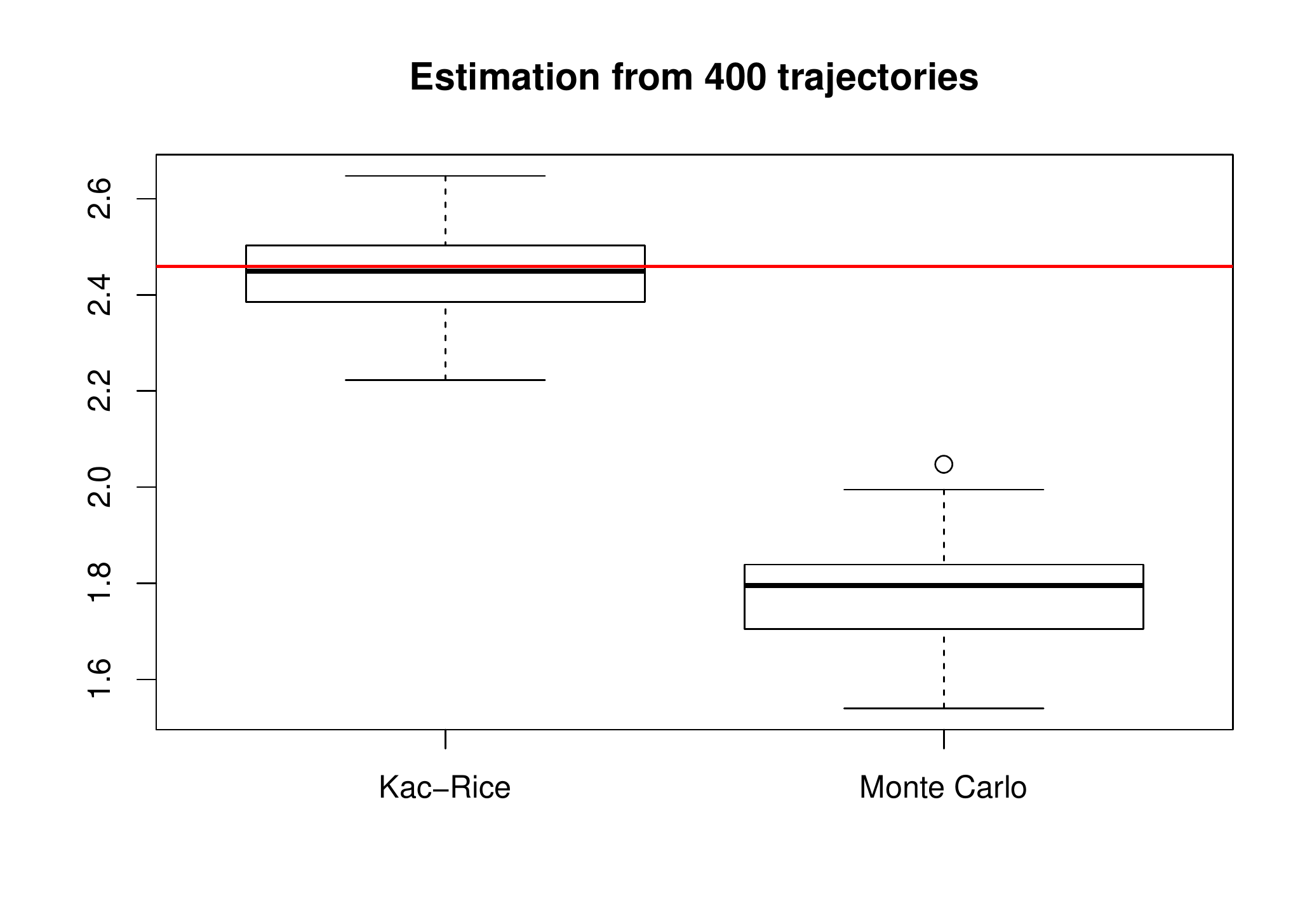}
\caption{Boxplots over $100$ replicates of Monte Carlo and non-stationary Kac-Rice estimators of $C_{S_2}(20)$ for the two-dimensional telegraph process from $n=200$ (left) and $n=400$ (right) trajectories extracted from a temporal grid with step size $h=1$. The (spatial) integration step size is $\Delta=0.1$.}
\label{bp_step_size_s_2_telegraph_2d}
\end{figure}

\begin{figure}[t!h]
\centering
\includegraphics[width=6.5cm]{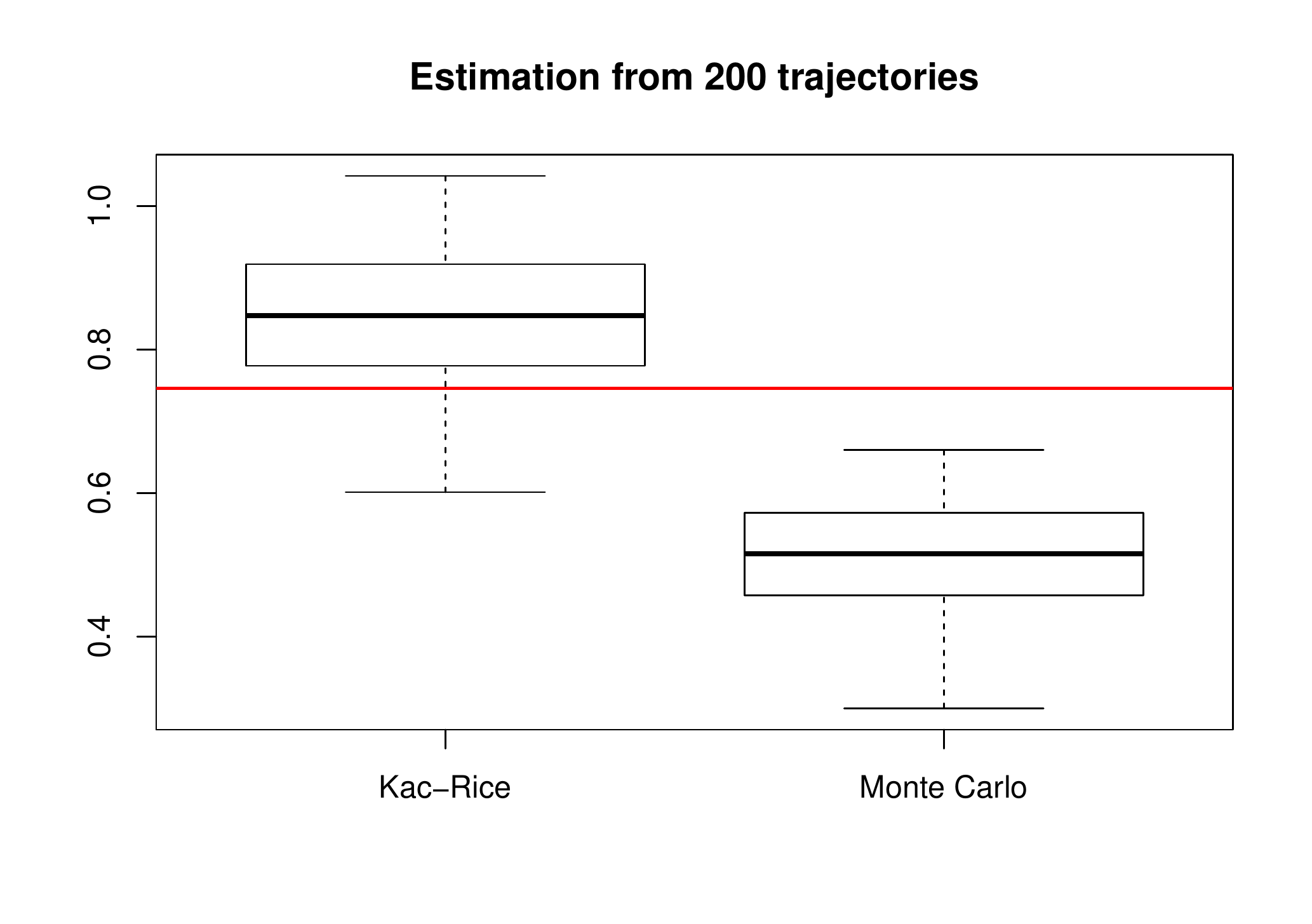}\qquad\qquad\includegraphics[width=6.5cm]{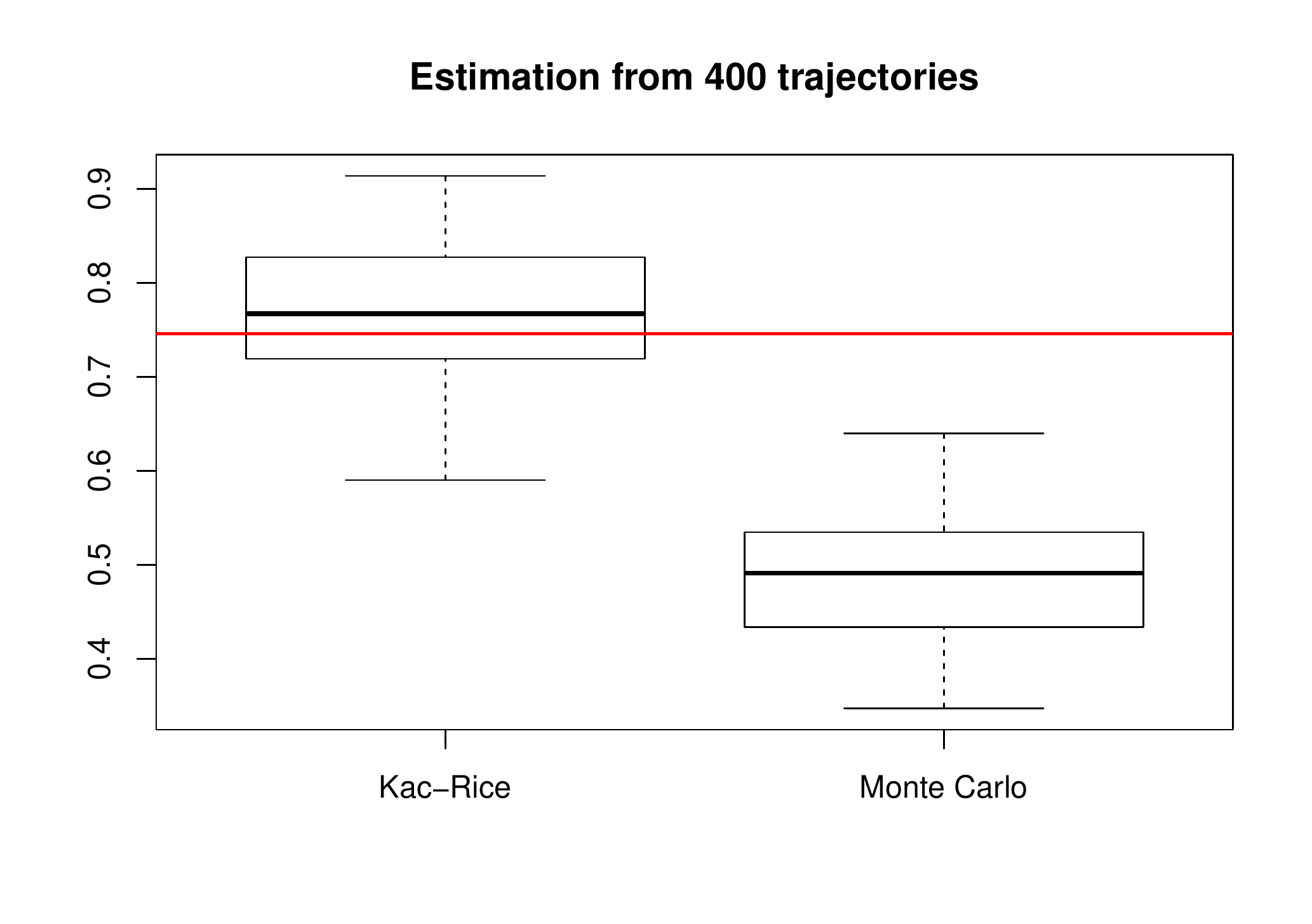}
\caption{Boxplots over $100$ replicates of Monte Carlo and non-stationary Kac-Rice estimators of $C_{S_3}(20)$ for the two-dimensional telegraph process from $n=200$ (left) and $n=400$ (right) trajectories extracted from a temporal grid with step size $h=1$. The (spatial) integration step size is $\Delta=0.1$.}
\label{bp_step_size_s_3_telegraph_2d}
\end{figure}

\section{Terrestrial and marine behaviors of a seabird species}
\label{s:realdata}

This section is devoted to the application of the methodology developed in this paper to a real dataset involving marine and terrestrial trajectories of lesser black-backed gulls around the island of Spiekeroog in Germany.

\subsection{Description of the dataset}

The authors of \cite{Garthe16} compare the marine and terrestrial foraging behaviors in lesser black-backed gulls breedind along the coast of the southern North Sea, more precisely at Spiekeroog. For this purpose, they attached GPS data loggers to eight incubating birds, recording GPS data from $10$ to $19$ days, every $3$ minutes approximatively, allowing for flight-path reconstruction. The authors have made the data from this experiment openly available at the Movebank Data Repository \cite{Garthe16data}. A GPS data is made of multiple entries. As long as we are concerned, a GPS data is composed of
\begin{itemize}
\item a timestamp: the date and time at which a sensor measurement was taken, in Coordinate Universal Time (UTC);
\item a latitude: the geographic latitude, in decimal degrees, of the location of the animal;
\item a longitude: the geographic longitude, in decimal degrees, of the location of the animal;
\item a ground speed: the estimated ground speed, in meters per second;
\item a heading: the direction in which the animal moved, in decimal degrees clockwise from north.
\end{itemize}
Some other entries are available, but will not be used in the present study, we refer the interested reader to the \verb+README.txt+ file associated to the dataset \cite{Garthe16data}.

\smallskip

\noindent
For each of the eight birds, we sliced the recording by days and retained only the days with from $440$ to $467$ data. There is a hundred of such days. We standardized the data on a common time grid with time step $\dd t=24\times 3600/467$ seconds, replacing the possible missing data by interpolation. The hundred resulting trajectories are displayed in Figure \ref{fig:gull_traj}.

\begin{figure}[ht]
\centering
\includegraphics[width=8cm]{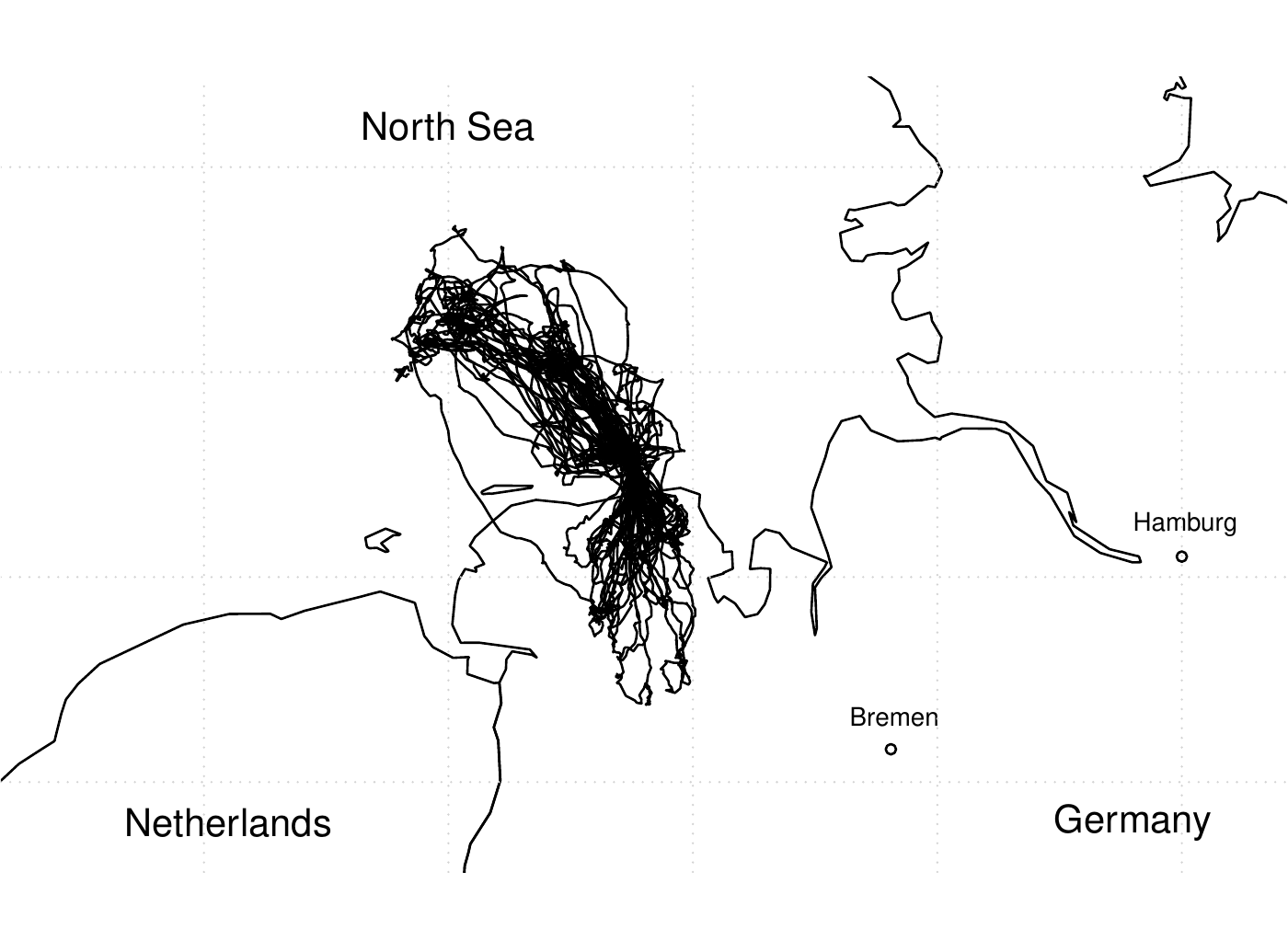}
\caption{The hundred lesser black-backed gulls' daily trajectories starting from the island of Spiekeroog.}
\label{fig:gull_traj}
\end{figure}

\subsection{Average depth of terrestrial and marine trips}

We propose to analyze the marine and terrestrial trip behaviors of the considered colony of gulls by means of the averaged number of crossings of some segments distant from the colony nests localized at Spiekeroog. Even if the simulation experiments of Section \ref{s:simu} clearly show the pertinence of the use of the Kac-Rice estimator for the estimation of the average number of crossings, for comparison purpose, we implement both the Kac-Rice and Monte Carlo methods for studying this real data set. We begin with the presentation of the implementation of the non-stationary Kac-Rice method and then proceed briefly to the presentation of the implementation of the Monte Carlo method.

\paragraph{Implementation of the non-stationary Kac-Rice estimator} Our time unit is the hour. We model the one day trajectory of a gull by a PSP $(X(t))_{t\in[0,24]}$ as defined in Subsection \ref{psp}. The velocity vector $r$ associated to $X$ is unfortunately unknown and has to be inferred from data. More precisely, in formula \eqref{KacRice:dd}, only the scalar product of the velocity vector with the normal to the considered hypersurface is needed to compute the average number of crossings. For a segment $[AB]$, we propose to estimate the scalar product (referred to as the speed projection in what follows) $(r(\rho),\nu_{[AB]})$, where $\rho\in[AB]$ and $\nu_{[AB]}$ is the normal to the segment $[AB]$, in the following way. We browse the segment $[AB]$ at points
$$
\rho_k=A+k\,\dd x(B-A)
$$
with $\dd x$ a space step size and $0\leq k\leq \lfloor1/ \dd x\rfloor$. Let $\e>0$ be an exploration distance. For each point $\rho_k$, we consider the set ${\cal C}_{\e,k}$ of points of the dataset that are $\e$-close to $\rho_k$. We then define the two sets of scalar products
$$
{\cal V}^+_{\e,k}=\{(r(\rho),\nu_{[AB]})_+~:~\rho\in {\cal C}_{\e,k}\}\quad\text{and}\quad {\cal V}^-_{\e,k}=\{(r(\rho),\nu_{[AB]})_-~:~\rho\in {\cal C}_{\e,k}\}
$$
of positive and negative speed projections where, for a point $\rho$ of the dataset,
$$
r(\rho)=v_{\rho}\begin{pmatrix}\cos\,\theta_\rho\\\sin\,\theta_\rho\end{pmatrix},
$$
with $v_\rho$ the recorded ground speed (converted in decimal degrees per hour) and $\theta_\rho$ the recorded heading (converted in radian, clockwise from East). To the point $\rho_k$ of the segment is then associated the positive speed projection ${\text{mean}}({\cal V}^+_{\e,k})$ and the negative speed projection ${\text{mean}}({\cal V}^-_{\e,k})$ (the mean of the empty set being zero). The sets of positive and negative speed projections
$$
\{(r(\rho),\nu_{[AB]})_+~:~\rho\in [AB]\}\quad\text{and}\quad \{(r(\rho),\nu_{[AB]})_-~:~\rho\in [AB]\}
$$
on the whole segment $[AB]$ are then approximated by locally-weighted polynomial regression (performed with \verb+R+ function \verb+lowess+ in our data experiments) on the sets
$$
\{{\text{mean}}({\cal V}^+_{\e,k})~:~0\leq k\leq \lfloor1/ \dd x\rfloor\}\quad\text{and}\quad \{{\text{mean}}({\cal V}^-_{\e,k})~:~0\leq k\leq \lfloor1/ \dd x\rfloor\},
$$
as illustrated in Figure \ref{gull_velocity}. This method allows us to consider the positive and negative speed projection functions
$$
s^{\pm}_{\e,\dd x}:\rho\in[AB]\mapsto (r(\rho),\nu_{[AB]})_{\pm}
$$
associated to $X$ for any segment $[AB]$ (and actually for any one-dimensional space endowed with a normal vector field).

\begin{figure}[ht]
\centering
\includegraphics[width=6.5cm]{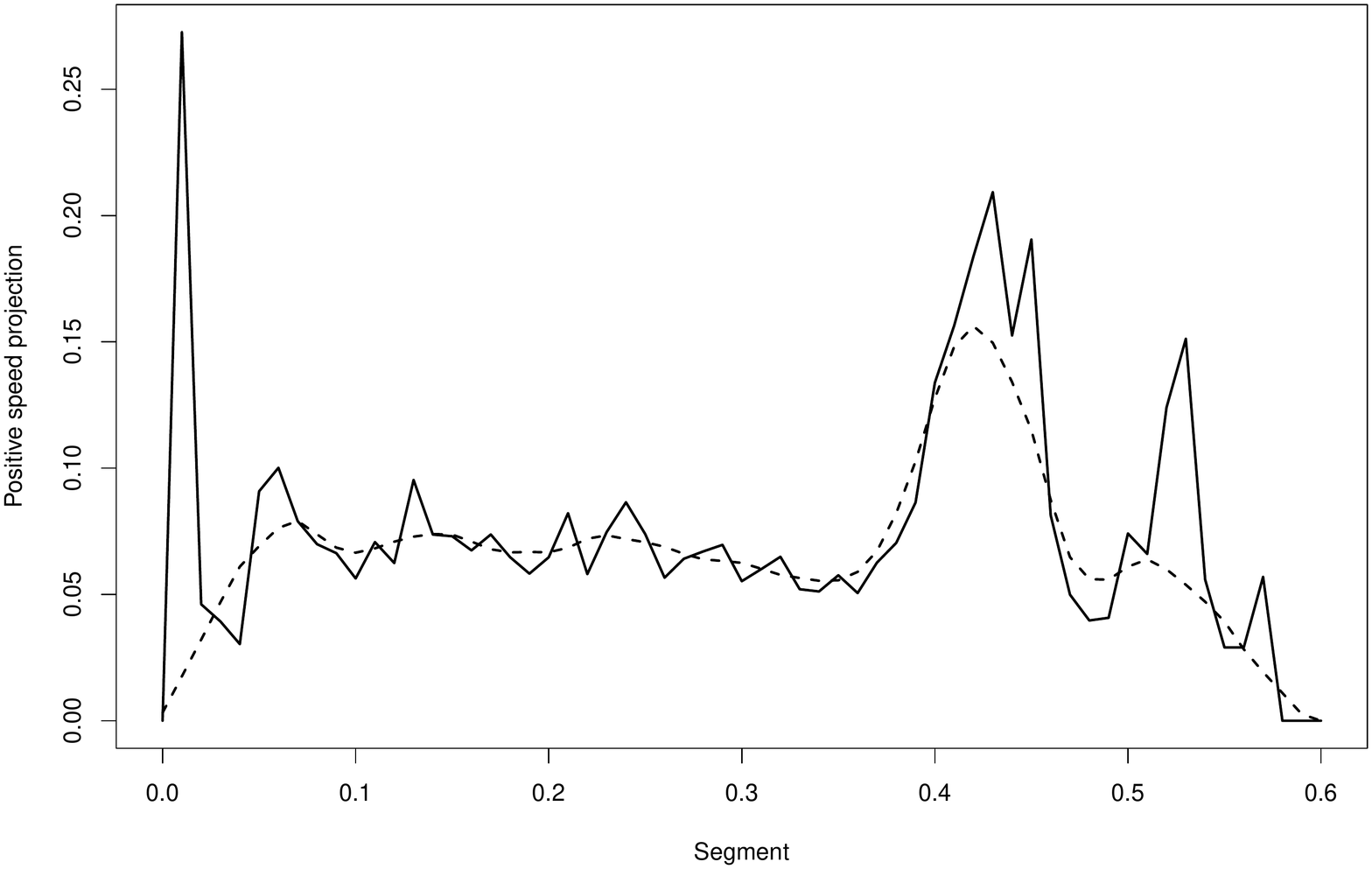}
\qquad\qquad\includegraphics[width=6.5cm]{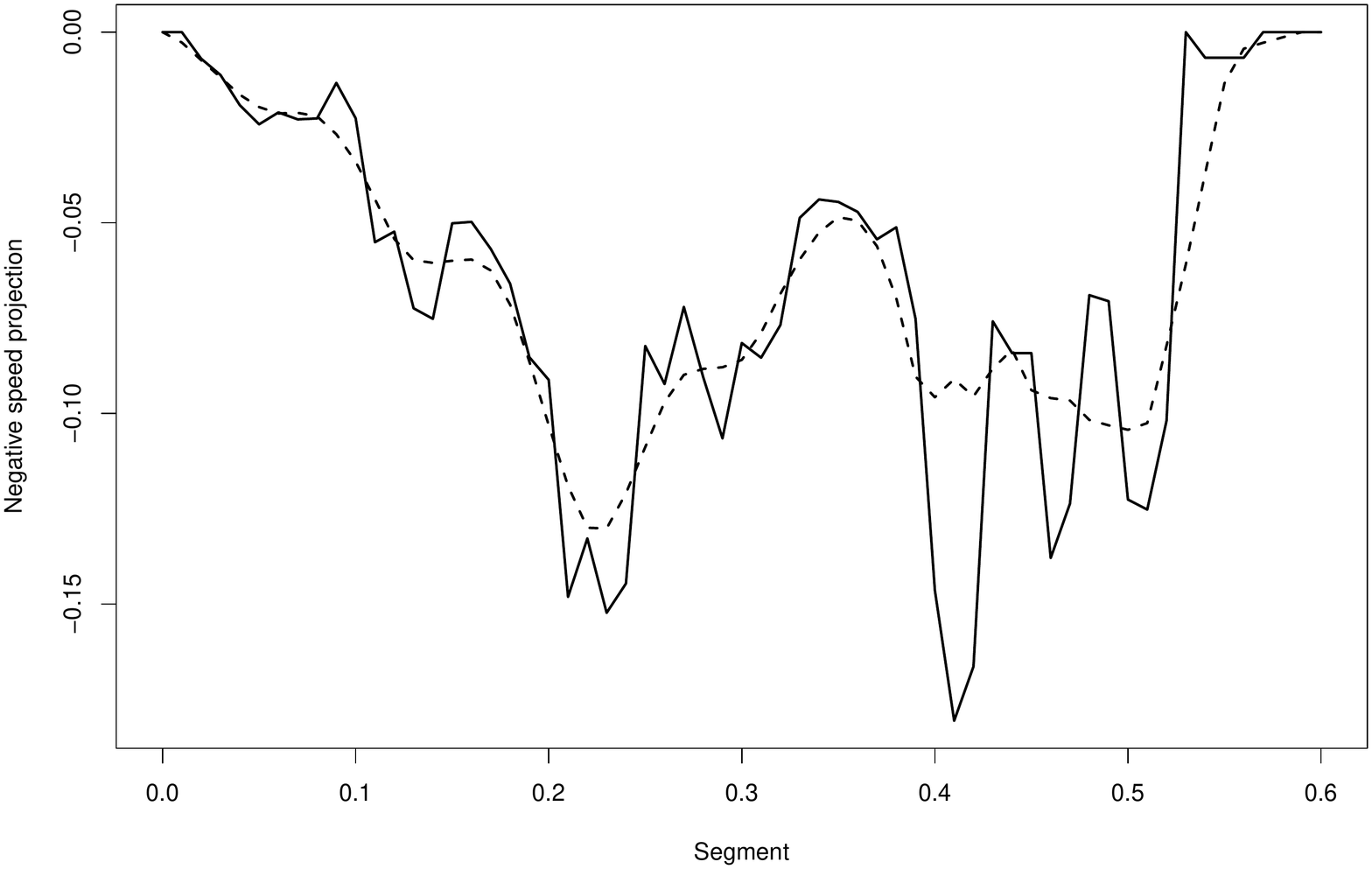}
\caption{Positive (left) and negative (right) speed projections over the normal to the segment $[AB]$ with $A=(7.45,53.54)^T$ and $B=(8.05,53.54)^T$ (solid curve) together with their locally-weighted polynomial regression over the whole segment (dashed curve).}
\label{gull_velocity}
\end{figure}

\begin{remark}
We would like to emphasize that the procedure developed to estimate the scalar products $(r(\rho),\nu(\rho))_{\rho\in S}$ on a crossing surface $S$ is quite general and does not rely on the application context considered in this section. The ingredients required to compute an estimation of $(r(\rho),\nu(\rho))$ at a specific point $\rho\in S$ are the existence and knowledge of the normal $\nu(\rho)$ and the existence of points $\e$-close to $\rho$ with observed velocity vectors. 
\end{remark}

\smallskip

\noindent
Since a bird always goes back to the colony location (at least in our dataset), there is theoretically as many negative speed projections as positive ones and the average number of crossings of the segment $[AB]$ by the process $X$ is then approximated, using the non-stationary Kac-Rice method, by
$$
\int_{[AB]} \frac{s^+_{\e,\dd x}(\rho)+s^{-}_{\e,\dd x}(\rho)}{2}\int_0^{24} p_{X(t)}(\rho)\dd t\,\dd \rho
$$
which is then evaluated by plugging a non-parametric estimator of the distribution, assuming that the hundred available gull's trajectories are independent.

\paragraph{Implementation of the Monte Carlo estimator} For the Monte Carlo method, we have to count, for a given segment [AB], the number of times a trajectory actually crosses this segment. For a fragment of trajectory $[ X(t_i)X(t_{i+1})]$ with $1\leq i\leq n_H-1$, this involves the computation of the following determinants:
\begin{align*}
d_{1,i}=\rm{det}(B-A,X(t_i)-A),\quad &d_{2,i}=\rm{det}(B-A,X(t_{i+1})-A),\\
d_{3,i}=\rm{det}(X(t_{i+1})-X(t_i),A-X(t_i)),\quad &d_{4,i}=\rm{det}(X(t_{i+1})-X(t_i),B-X(t_i)).
\end{align*}
If $d_{1,i}d_{2,i}<0$ and $d_{3,i}d_{4,i}<0$, then there is indeed one crossing between $[AB]$ and $[ X(t_i)X(t_{i+1})]$. We then average on the total number of trajectories to obtain the average number of crossings according to the Monte Carlo method.

\paragraph{The crossing segments} For both methods, we compute the average number of crossings for two sets of segments: one set being more and more distant to Spiekeroog in the sea direction
$$
{S}_{\text{sea}} = \left\{[A_i B_i]~:~0\leq i\leq 60,~A_i=A+i\,\dd x\,R^i\nu_{[AB]},~B_i=B+R^i\left(B+i\,\dd x\,\nu_{[AB]}-A\right)\right\},$$
with $A=(6.9,53.7)^T$, $B=(7.5,53.77)^T$ and $R$ the rotation matrix of angle $\theta=\pi/247$, and the other set of segments being also more and more distant to Spiekeroog but in the inland direction
$$
{S}_{\text{inland}}=\left\{[A'_i B'_i]~:~0\leq i\leq 60,~A'_i=A'+i\,\dd x\,\nu_{[A'B']},~B'_i=B'+i\,\dd x\,\nu_{[A'B']}\right\},
$$
with $A'=(7.45 , 53.7)^T$ and $B'=(8.05,53.7)^T$, as presented in Figure \ref{traj_seg}.

\begin{figure}[ht]
\centering
\includegraphics[width=6.5cm]{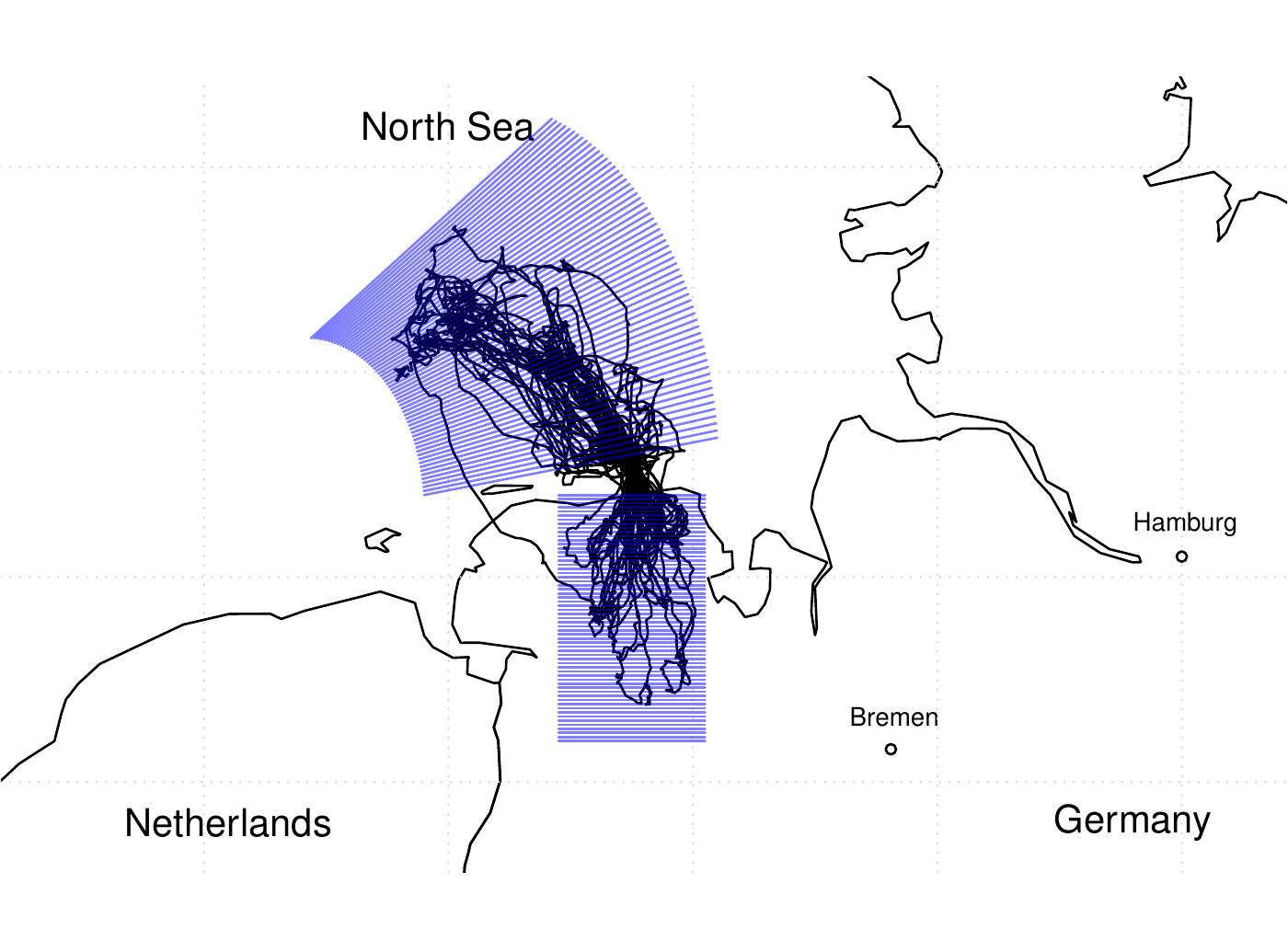}\qquad\qquad\includegraphics[width=6.5cm]{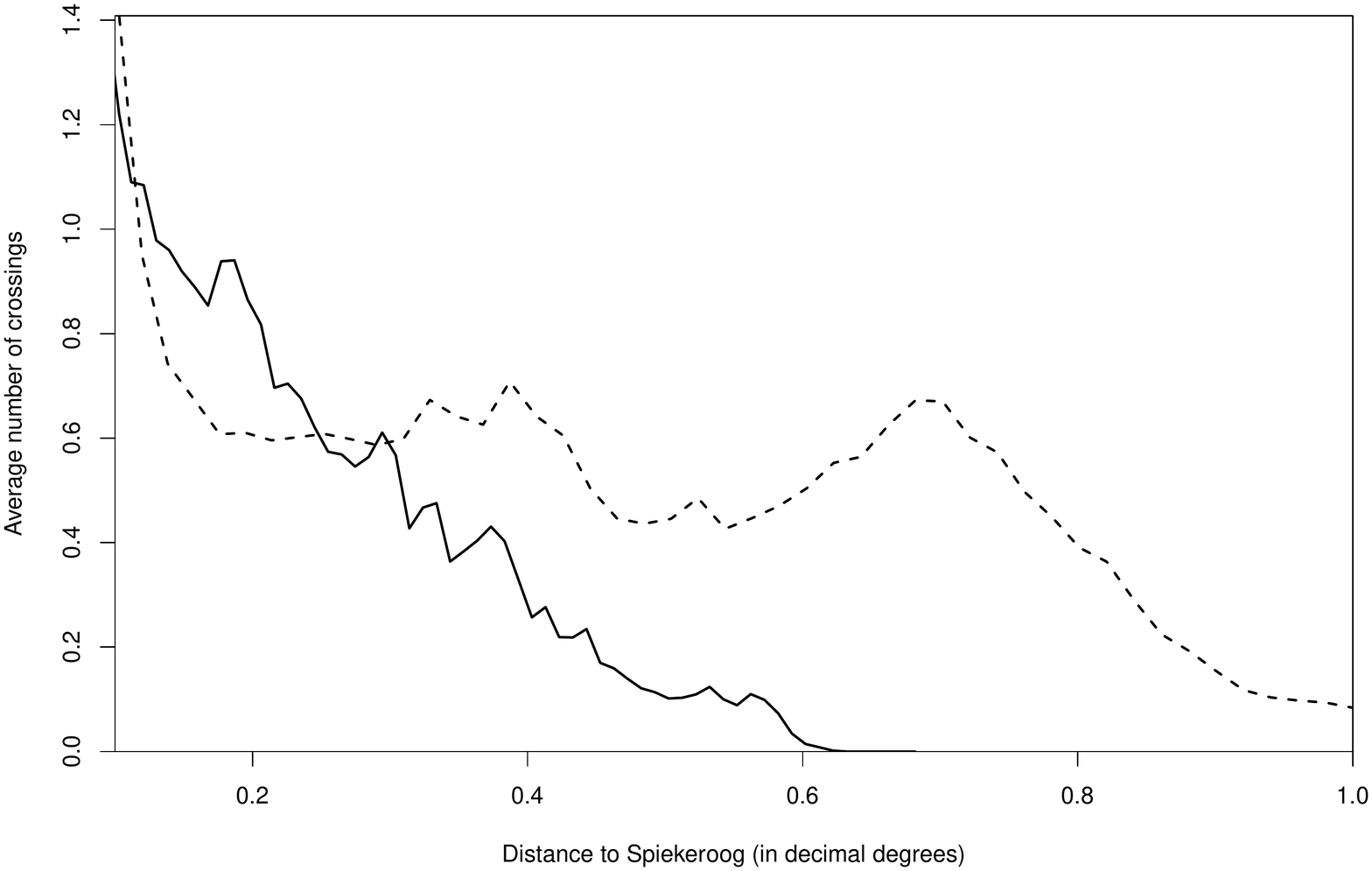}
\caption{Lesser black-backed gulls' trajectories with the two sets of crossings segments $S_{\text{sea}}$ and $S_{\text{inland}}$ (left) and average number of crossings of a daily bird's trajectory in the sea direction (dashed line) and inland direction (solid line) with respect to the distance to Spiekeroog in decimal degrees (right).}
\label{traj_seg}
\end{figure}

\paragraph{Numerical results} In Figure \ref{res_realdata} are displayed the average numbers of crossings of a daily lesser black-backed gull's trajectory for the $S_{\text{sea}}$ and $S_{\text{inland}}$ sets of segments, ordered with respect to the distance of the segment to Spiekeroog for the non-stationary Kac-Rice (left) and the Monte Carlo (right) methods. First, we notice that the average numbers of crossings towards sea and inland follow two distinct trends. The average number of crossings towards inland decreases when the distance to Spiekeroog grows whereas the mean number of crossings towards open sea begins with a decreasing but then shows two episodes of growing before going down to zero.  We also observe that there are crossings towards sea farther and more often than towards inland: the average number of crossings towards sea is above $0.4$ until distance $0.8$ (in decimal degrees) whereas this average number of crossings goes below $0.4$ from half this distance, that is $0.4$ (in decimal degrees).
These phenomena can be explained in light of \cite{Garthe16}. Indeed, the two episodes of crossing growing for marine trips correspond to areas where beam trawlers discard fish \cite[Figure 6]{Garthe16}. These discardings are highly attended by the gulls that travel longer distances than towards inland to reach the discarding areas. Moreover, the gulls certainly stay a while above these locations to feed upon these discards, increasing their average number of crossings at these places. The simpler behavior of gulls on land \cite[p.\,15]{Garthe16} explains the monotone decreasing of the average number of crossings towards inland.

\begin{figure}[ht]
\centering
\includegraphics[width=6.5cm]{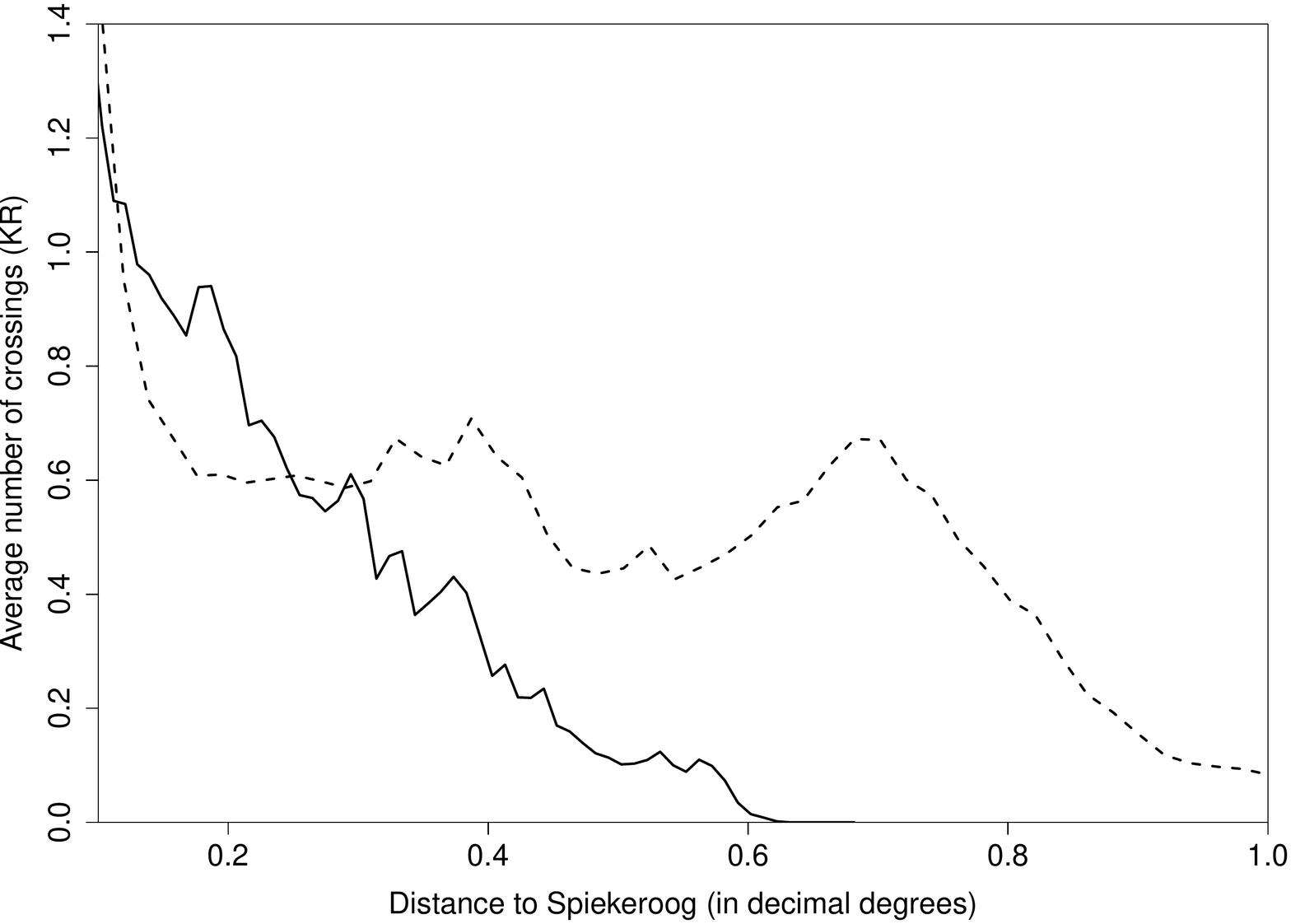}\qquad\qquad\includegraphics[width=6.5cm]{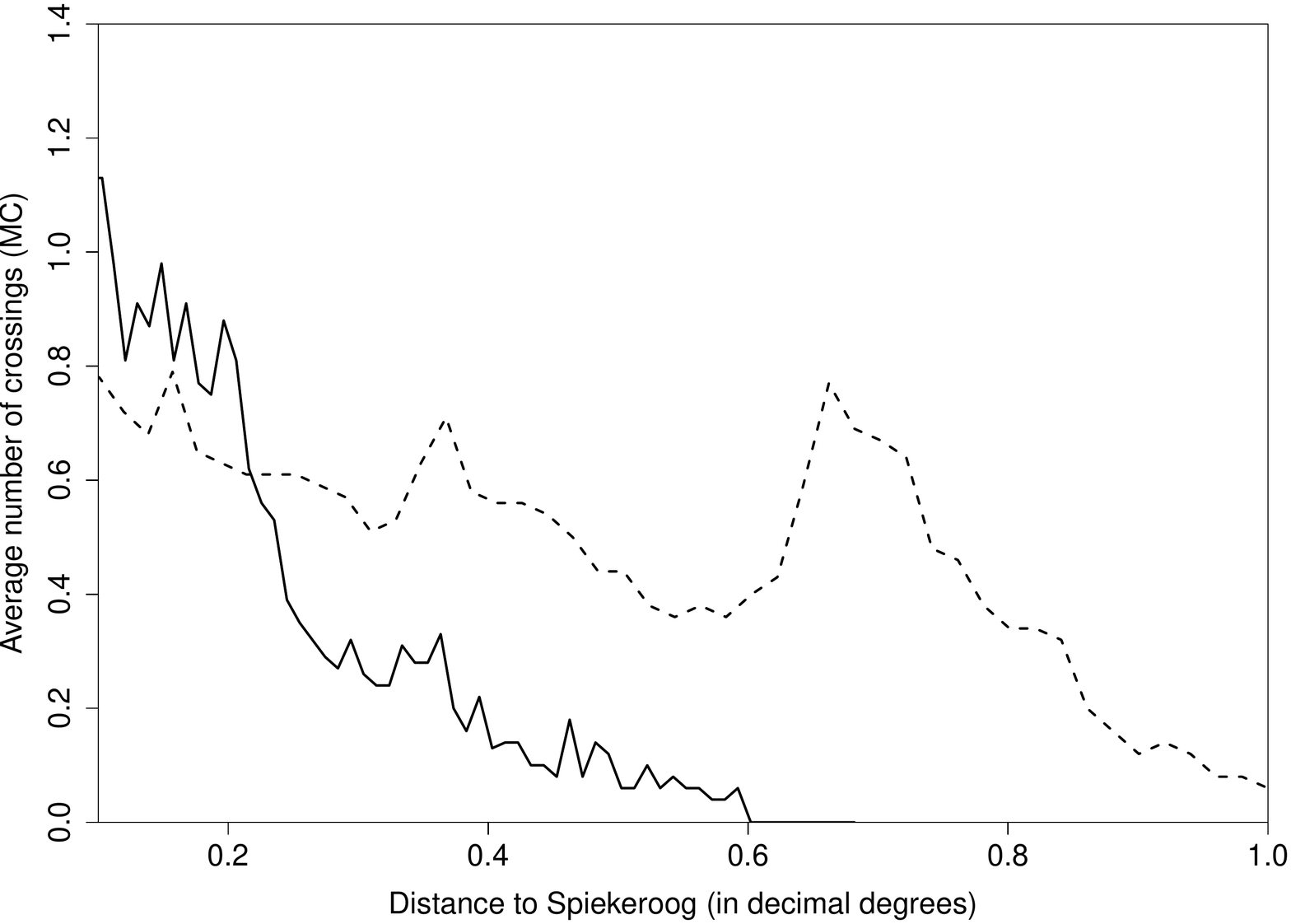}
\caption{Average number of crossings of a daily bird's trajectory in the sea direction (dashed line) and inland direction (solid line) with respect to the distance to Spiekeroog in decimal degrees, for the non-stationary Kac-Rice estimator (left) and the Monte Carlo estimator (right).}
\label{res_realdata}
\end{figure}

\noindent
Figure \ref{res_realdata2} allows for comparisons between the non-stationary Kac-Rice and the Monte Carlo estimators towards earth (left) and towards sea (right). Towards sea, the behaviors are similar except near the net, where the average number of crossings for the non-stationary Kac-Rice estimator is bigger than for the Monte Carlo estimator. Towards earth, we observe that the Monte Carlo method always gives lower estimation results than the non-stationary Kac-Rice method does. This phenomenon is expected in light of the simulation results of Section \ref{s:simu}: in most cases, the Monte Carlo estimator underestimates the average number of crossings. Moreover, we also notice that the Kac-Rice estimation is smoother than the Monte Carlo one. This roughness in the average number of crossing given by the Monte Carlo method may let think to the existence of singularities in the law of the position of the birds, a fact that, far from the net, is highly unlikely. The approximation of the density for the Kac-Rice estimator allows for estimating the average number of crossings also when few data are available, giving smoother and better estimations.

\begin{figure}[ht]
\centering
\includegraphics[width=6.5cm]{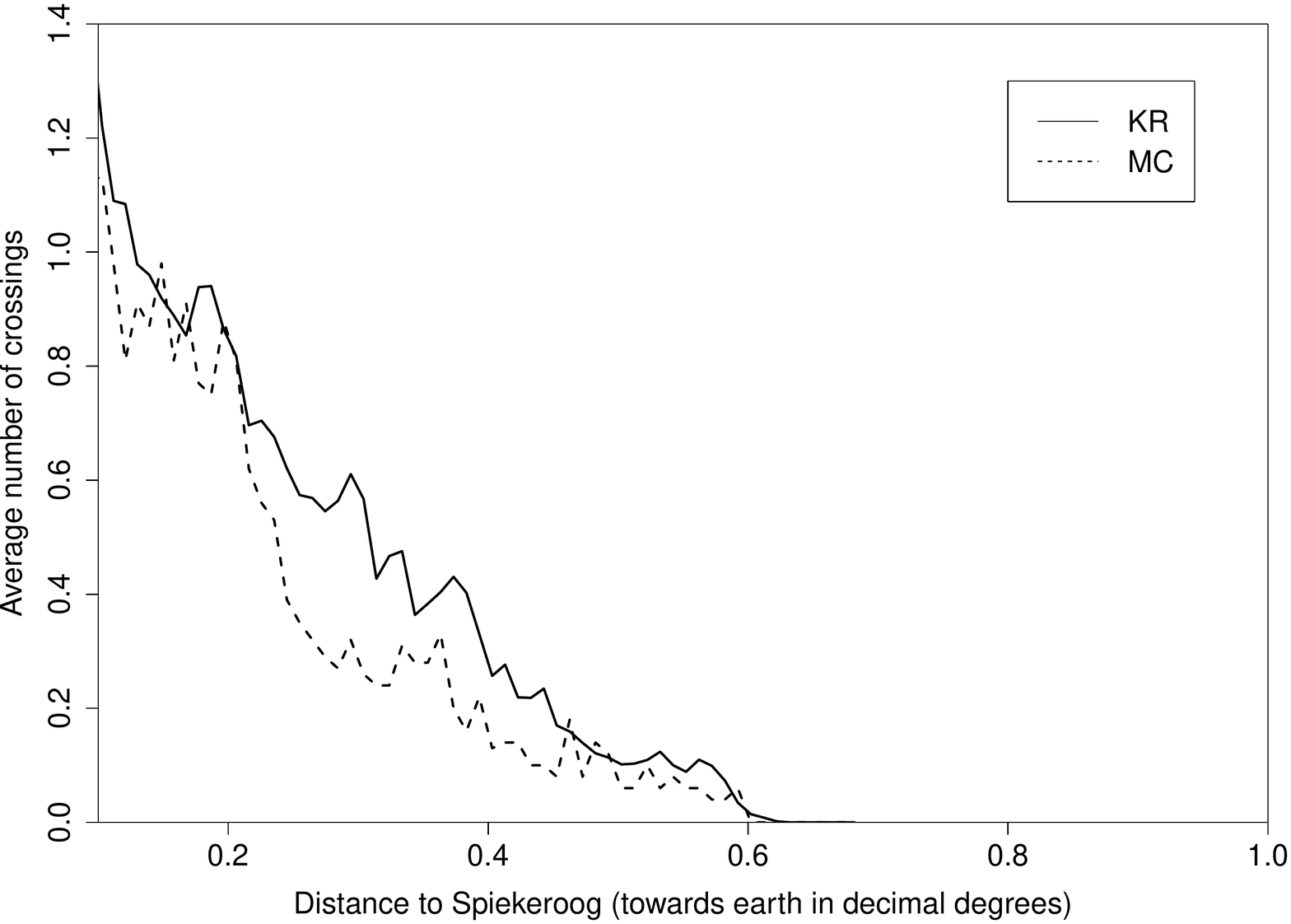}\qquad\qquad\includegraphics[width=6.5cm]{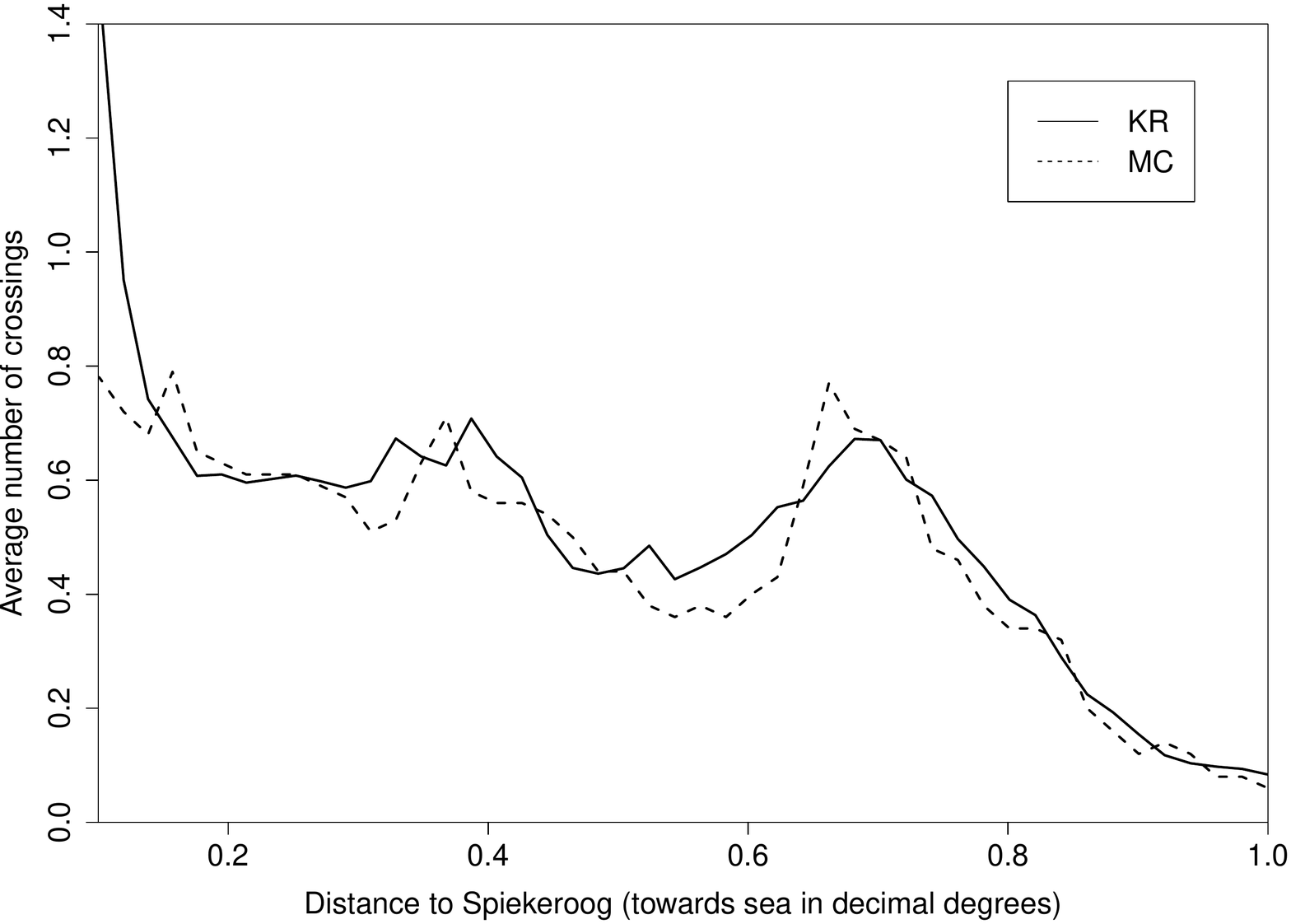}
\caption{Comparison between the two methods for estimating the average number of crossings towards earth (left) and towards sea (right).}
\label{res_realdata2}
\end{figure}


\bibliographystyle{acm}
\bibliography{biblio}


\appendix

\section{Proofs of Section \ref{s:kr}}
\label{app:proof2}

\subsection{Proof of Theorem \ref{crois:d1}}
\label{app:ss:crois:d1}

We start the proof with the same arguments as in the proof of \cite[Theorem 12]{DM15}. Since $r(x)\neq0$, we can apply Kac's counting formula to almost all trajectories of $X$. Indeed, the continuity of $r$ implies that between two successive jumps,  the trajectories of $X$ are $\mathcal{C}^1$ and such that $\dd X(t)/\dd t= r(X(t))$ for almost every $t\in[0,H]$. Moreover, under Assumption \ref{assumption:density}, the level $x$ is almost-surely not reached at $H$ nor at the jump times of the process. Moreover, Assumptions \ref{assump:dd:tangent} and \ref{assumption:density} imply that the trajectories of $X$ have no tangencies at the level $x$. Let us  also remark that the number of crossings of the level $x$ is bounded by the number of jumps of $X$ plus one, and has therefore finite expectation. Thus, we can apply Kac's counting formula \eqref{kac:formula} and have almost-surely
\begin{equation*}
c_x(H)=\lim_{\delta\to0}\frac{1}{2\delta}\int_0^H|r(X(t))|\mathbb{1}_{|X(t)-x|\leq \delta}\dd t.
\end{equation*}
Let us now take another way than in the proof of \cite[Theorem 12]{DM15}. For any $\delta>0$ we have
\begin{align*}
&\frac{1}{2\delta}\int_0^H|r(X(t))|\mathbb{1}_{|X(t)-x|\leq \delta}\dd t\\
=\,&\frac{1}{2\delta}\int_0^H(|r(X(t))|-|r(x)|)\mathbb{1}_{|X(t)-x|\leq \delta}\dd t~+~|r(x)|\frac{1}{2\delta}\int_0^H\mathbb{1}_{|X(t)-x|\leq \delta}\dd t.
\end{align*}
The left-hand-side converges towards $c_x(H)$ by Kac's counting formula. The first term of the right hand side is almost-surely, in absolute value, less or equal to
$$
\frac{1}{2\delta}\int_0^H|r'(\xi_{t,x})||X(t)-x|\mathbb{1}_{|X(t)-x|\leq \delta}\dd t\leq\frac12\int_0^H|r'(\xi_{t,x})|\mathbb{1}_{|X(t)-x|\leq \delta}\dd t,
$$
with $\xi_{t,x}\in(X(t),x)$. The term
$$
\int_0^H|r'(\xi_{t,x})|\mathbb{1}_{|X(t)-x|\leq \delta}\dd t
$$
converges almost-surely towards zero since $X(t)$ has a density for all time $t$, $r'$ is continuous and $\xi_{t,x}$ converges towards $x$ when $\delta$ goes to zero. Finally we have almost surely
$$
c_x(H)=\lim_{\delta\to0}|r(x)|\frac{1}{2\delta}\int_0^H\mathbb{1}_{|X(t)-x|\leq \delta}\dd t,
$$
which states the expected result.

\subsection{Proof of Corollary \ref{cor:dmf}}
\label{app:ss:cor:dmf}

First, let us remark that the fact that $l_x(H)$ has finite expectation is just a consequence of Fatou's lemma and Fubini-Tonelli theorem,
\begin{align*}
\EE(l_x(H))\leq\liminf_{\delta\to0}\frac{1}{2\delta}\EE\bigg(\int_0^H\mathbb{1}_{|X(t)-x|\leq \delta}\dd t\bigg)=\liminf_{\delta\to0}\frac{1}{2\delta}\int_{x-\delta}^{x+\delta}\int_0^Hp_{X(t)}(y)\dd t\,\dd y=\int_0^H p_{X(t)}(x)\dd t.
\end{align*}
Taking expectation on both sides of \eqref{crois:dim:1:formula}, we obtain the first equality. We are left to show that
$$
\EE(l_x(H))=\int_0^H p_{X(s)}(x)\dd s.
$$
For this purpose we first use the dominated convergence theorem
\begin{align*}
\EE(l_x(H))&=\EE\bigg(\lim_{\delta\to0}\frac{1}{2\delta}\int_0^H\mathbb{1}_{|X(t)-x|\leq \delta}\dd t\bigg)\\
&=\lim_{\delta\to0}\frac{1}{2\delta}\int_0^H\PP(|X(t)-x|\leq \delta)\dd t .
\end{align*}
The domination comes from the fact that the term $\frac{1}{2\delta}\int_0^H\mathbb{1}_{|X(t)-x|\leq \delta}\,\dd t$ is bounded, for some $\delta_0$ small enough, by the number of jumps of the process between times $0$ and $H$, which is integrable, times a quantity of order
$$
\left(\inf_{\{x'\,:\,|x'-x|\leq \delta_0\}}|r(x')|\right)^{-1}.
$$
The existence of $\delta_0$ such that $\inf_{\{x';|x'-x|\leq \delta_0\}}|r(x')|>0$ is ensured by the fact that $r(x)\neq0$ and $r$ is continuous. Then, from Fubini-Tonelli theorem, using the fact that $z\mapsto \int_0^Hp_{X(t)}(z)\dd t$ is continuous on a neighborhood of $x$, we have
\begin{align*}
\EE(l_x(H))=\lim_{\delta\to0}\frac{1}{2\delta}\int_0^H\PP(|X(t)-x|\leq \delta)\dd t&=\lim_{\delta\to0}\frac{1}{2\delta}\int_0^H\int_{x-\delta}^{x+\delta}p_{X(t)}(y)\dd y\dd t\\
&=\lim_{\delta\to0}\frac{1}{2\delta}\int_{x-\delta}^{x+\delta}\int_0^Hp_{X(t)}(y)\dd t\dd y\\
&=\int_0^H p_{X(t)}(x)\dd t,
\end{align*}
which ends the proof.

\subsection{Proof of Theorem \ref{th:kr:multidim}}
\label{app:ss:th:kr:multidim}

According to Assumption \ref{assumption:density}, the number of continuous crossings of $S$ by $X$ is almost-surely equal to the number of continuous crossings of $0$ by the process $(\rho(X(t))_{t\in[0,H]}$. Under Assumption \ref{assumption:time:mark} on the integrability of the number of jump times and Assumption \ref{assump:dd:tangent}, the number of continuous crossings $c_S(H)$ has finite expectation. Applying Kac's counting formula (\ref{kac:formula}) to this process, we almost-surely have
$$
c_S(H)=\lim_{\delta\to0}\frac{1}{2\delta}\int_0^H |(r(X(s)),\nu(X(s)))|\mathbb{1}_{(-\delta,\delta)}(\rho(X(s)))\|\nabla \rho(X(s))\|\dd s,
$$
where we recall that $\nu$ is the field of outward normals to $S$, extended continuously on a neighborhood of $S$. Then, by the same arguments as in dimension one (considering the application $x\mapsto (r(x),\nu(x))$ instead of $x\mapsto r(x)$), the exchange of expectation and limit is justified and gives
\begin{align*}
C_S(H)&=\lim_{\delta\to0}\frac{1}{2\delta}\int_0^H \int_{\R^d}|(r(x),\nu(x))|\mathbb{1}_{(-\delta,\delta)}(\rho(x)) p_{X(s)}(x)\dd x\,\|\nabla \rho(x)\|\dd s\\
&=\lim_{\delta\to0}\frac{1}{2\delta} \int_{\cal X}|(r(x),\nu(x))|\left[\int_0^H p_{X(s)}(x)\dd s\right]\,\|\nabla \rho(x)\|\mathbb{1}_{(-\delta,\delta)}(\rho(x))\dd x\\
&=\int_S|(r(x),\nu(x))|\int_0^H p_{X(s)}(x)\dd s\,\sigma_{d-1}(\dd x),
\end{align*}
the last equality being obtained as in \cite[Example 3.7.2, p.\,109]{K12}, using the fact that the map $x\mapsto |(r(x),\nu(x))|\int_0^H p_{X(s)}(x)\dd s$ is continuous (see Assumption \ref{assumption:density}) and thus measurable and bounded on a neighborhood of $S$.


\section{Proofs of Section \ref{sec:sf}}
\label{app:proof3}

\subsection{Proof of Theorem \ref{thm:nonstatkr:consis}}
\label{app:ss:nonstat}

We split the difference $\widehat{C}^{(n,n_H)}_x(H)-{C}_S(H)$ as follows,
$$ \widehat{C}^{(n,n_H)}_S(H)-{C}_S(H) = T_1 + T_2 ,$$
with
\begin{eqnarray*}
T_1 &=& \int_S|(r(x),\nu(x))|\,\frac{H}{n_H-1}\sum_{j=1}^{n_H}[\widehat{p}^{(n)}_{X(h_j)}(x)-{p}_{X(h_j)}(x)]\,\sigma_{d-1}(\dd x) ,\\
T_2 &=& \int_S|(r(x),\nu(x))|\,\left[\frac{H}{n_H-1}\sum_{j=1}^{n_H}p_{X(h_j)}(x)-\int_0^H p_{X(s)}(x)\dd s\right] \,\sigma_{d-1}(\dd x).
\end{eqnarray*}
The second term $T_2$ is deterministic and independent of $n$. According to Assumption \ref{assumption:density}, the map $(x,t)\mapsto p_{X(t)}(x)$ is continuous on the compact set $S\times[0,H]$, thus, the convergence of the rectangle method implies that for any $x\in S$,
$$
\lim_{n_H\to\infty}\frac{H}{n_H-1}\sum_{j=1}^{n_H}p_{X(h_j)}(x)=\int_0^H p_{X(s)}(x)\dd s.
$$
Since $S$ is a compact hyper surface, thus bounded with finite hyper-volume, and the map $x\mapsto|(r(x),n(x))|$ is continuous on $S$, the second term $T_2$ goes to zero when $n_H$ goes to infinity. The almost-sure convergence of the first term to zero is a consequence of \cite[Theorem 1]{DW80}. Indeed, since the map $(x,t)\mapsto p_{X(t)}(x)$ is continuous on the compact set $S\times[0,H]$, we have, with probability one, for any $1\leq j\leq n_H$,
$$
\lim_{n\to\infty}\sup_{x\in S}|\widehat{p}^{(n)}_{X(h_j)}(x)-{p}_{X(h_j)}(x)|=0.
$$
Then, again, since $S$ is a compact hyper surface and the map $x\mapsto|(r(x),\nu(x))|$ is continuous on $S$, the first term $T_1$ goes to zero with probability one.

\subsection{Proof of Theorem \ref{thm:statkr:consis}}
\label{app:ss:stat}

We can write the difference $\widetilde{C}^{(n,n_H)}_S(H)-C_S(H)$ as follows,
$$
\widetilde{C}^{(n,n_H)}_S(H)-C_S(H)=\frac{H}{n}\sum_{i=1}^n \int_S|(r(x),\nu(x))|\left(\widehat{\mu}_{i,n_H}(x)-\mu(x)\right)\sigma_{d-1}(\dd x).
$$
Therefore, for any $\e>0$, we have
$$
\PP_\mu\left(|\widetilde{C}^{(n,n_H)}_S(H)-C_S(H)|\geq\e\right)\leq \sum_{i=1}^n\PP_\mu\left(\sup_{x\in S}|\widehat{\mu}_{i,n_H}(x)-\mu(x)|\geq \frac{\e}{\sigma_{d-1}(S)\sup_{x\in S}|(r(x),\nu(x))|}\right).
$$
By virtue of \cite[Theorem 6]{ADP16}, we have, for any $1\leq i\leq n$,
$$
\lim_{n_H\to\infty}|\widehat{\mu}_{i,n_H}(x)-\mu(x)|=0
$$
in $\PP_\mu$-probability, which states the expected result.

\end{document}